\documentclass[
 onecolumn,
 notitlepage,
 superscriptaddress,
 nofootinbib,
 amsmath,amssymb,
 aps,
]{revtex4-2}

\usepackage{mathpazo}
\usepackage[utf8]{inputenc}
\usepackage[english]{babel} 
\usepackage{graphicx}
\usepackage{dcolumn}
\usepackage{bm}
\usepackage[bookmarks, colorlinks=true, urlcolor=blue, linkcolor=black, citecolor=blue]{hyperref}
\usepackage{mathtools}
\usepackage{physics}
\usepackage{enumitem}
\usepackage[dvipsnames]{xcolor}
\usepackage{float}
\usepackage[capitalise]{cleveref}
\usepackage{amsthm}
\usepackage{thmtools, thm-restate}

\newtheorem{theorem}{Theorem}
\newtheorem{lemma}[theorem]{Lemma}
\newtheorem{corollary}[theorem]{Corollary}
\theoremstyle{definition}
\newtheorem{definition}[theorem]{Definition}

\usepackage{etoolbox}
\usepackage{booktabs}

\usepackage{float}
\makeatletter
\let\newfloat\newfloat@ltx
\makeatother

\usepackage{algorithm}
\usepackage{algpseudocode}
\makeatletter
\renewcommand{\fnum@algorithm}{\fname@algorithm~\thealgorithm}
\makeatother
\algnewcommand\Input{\item[\textbf{Input:}]}
\algnewcommand\Output{\item[\textbf{Output:}]}

\DeclareMathOperator{\sign}{sign}

\renewcommand{\exp}{\ensuremath{\mathrm{exp}}}

\providecommand{\to}{\ensuremath{\Tilde{o}}}

\def\A{\mathcal{A}}

\def\D{\mathcal{D}}

\def\T{\mathcal{T}}

\begin{document}

\title{Supervised Learning with Quantum Hamiltonian Embeddings}
\author{Riccardo Molteni}
\thanks{The first two authors contributed equally to this work.}
\affiliation{IBM Research}
\affiliation{Applied Quantum Algorithms $\langle$aQ$a^L$$\rangle$, Leiden University, The Netherlands}

\author{Minjun Jeon}
\thanks{The first two authors contributed equally to this work.}
\affiliation{IBM Research}
\affiliation{Department of Materials, University of Oxford, Parks Road, Oxford OX1 3PH, United Kingdom}

\author{Paolo Braccia}
\affiliation{IBM Research}

\author{Abhinav Deshpande}
\affiliation{IBM Research}

\author{Vedran~Dunjko}
\affiliation{Applied Quantum Algorithms $\langle$aQ$a^L$$\rangle$, Leiden University, The Netherlands}

\author{Liang Jiang}
\affiliation{The University of Chicago}

\author{Su-un~Lee}
\affiliation{IBM Research}
\affiliation{The University of Chicago}

\author{Swarnadeep Majumder} 
\affiliation{IBM Research}

\author{Javier~Robledo~Moreno}
\affiliation{IBM Research}

\author{Parthasarathy Suryanarayanan}
\affiliation{IBM Research}

\author{Kunal Sharma}
\email{kunals@ibm.com}
\affiliation{IBM Research}

\date{\today}

\begin{abstract}
Machine learning is widely regarded as a promising application in which quantum computing could offer a computational advantage. However, general quantum machine learning models with provable guarantees on performance, trainability, and efficient inference remain limited. In this work, we introduce quantum Hamiltonian embeddings, a supervised learning framework in which a classical training dataset is encoded into a single effective, parameterized Hamiltonian. Training optimizes the parameters of this data-dependent Hamiltonian, whose approximate ground state defines the classification rule. Once the ground state has been approximated using a quantum algorithm, inference can be performed either classically using a classical representation of the learned state, or with a quantum representation of the state. We establish efficient methods for gradient estimation and prove favorable generalization guarantees. Furthermore, under standard cryptographic assumptions, we show that Hamiltonian embeddings can realize hypothesis functions beyond the reach of efficient classical algorithms. We then demonstrate the framework at scale on IBM Heron processors using a standard benchmark dataset. Together, these results establish Hamiltonian embeddings as a new framework for quantum machine learning, combining theoretical guarantees with a path for experimental investigations at scale on quantum hardware.
\end{abstract}

\maketitle
\tableofcontents
\newpage

\section{Introduction}
The prospect of quantum computational advantage in learning tasks has driven significant interest in quantum machine learning (QML)  ~\cite{biamonte2017quantum,chang2025primer}. For instance, in supervised learning, multiple architectures have been proposed~\cite{havlicek2019supervised,schuld2021supervised,perez-salinas2020data, rebentrost2014quantum,  jerbi2024shadows} and provable advantages have been established in particularly constructed settings~\cite{liu2021rigorous, anschuetz2026arbitrary,molteni2026exponential, zhao2026exponential, gil-fuster2025on}.

A broader goal in QML is to develop general models that are efficiently trainable, admit generalization guarantees, cannot be efficiently simulated classically in general, while also being implementable at scale on quantum processors.
However, the difficulty of simultaneously satisfying these requirements is particularly evident for parameterized quantum circuits~\cite{benedetti2019parameterized,cerezo2021variational,schuld2020circuit}. Shallow-depth PQCs are well suited for current quantum processors but can often be simulated efficiently on classical computers \cite{cerezo2025does, gil-fuster2025on}, whereas deeper circuits may be difficult to optimize because of barren plateaus~\cite{mcclean2018barren, cerezo2021cost, larocca2025barren}. Similarly, beyond a few examples \cite{glick2024covariant, henderson2025quantum}, identifying effective quantum kernels for more general datasets remains challenging. 

In this work, we introduce Hamiltonian embedding, a classification algorithm that can be efficiently implemented on a quantum processor, in which the classifier is constructed by solving a physically motivated energy-minimization problem. The method encodes the entire training dataset into a single parameterized Hamiltonian. Each training example contributes a term obtained by dressing a fixed, label-dependent operator, namely by conjugating it with a parameterized unitary that encodes the corresponding input. Then training optimizes the encoding parameters to lower the ground-state energy of this data-dependent Hamiltonian, and an approximate ground state of the optimized Hamiltonian defines the classifier. In particular, for a new input, the trained encoding specifies an observable whose expectation value with respect to the learned ground state determines the predicted label. We identify tasks for which the physical constraints imposed by the learned quantum state provide a useful inductive bias, yielding a provable reduction in sample complexity relative to classically regularized surrogate models.

We then show that the Hamiltonian embedding framework has several properties desirable for QML. During optimization, a quantum algorithm prepares approximations to the required ground states. In the case where the encoding unitaries allow the expectation values of the dressed observables to be efficiently computed from a classical representation of these ground states, both the loss and its gradients with respect to the trainable parameters can be evaluated classically. In particular, we show that the gradients can be expressed as expectation values and evaluated using a parameter-shift rule.

Shallow local encoding unitaries provide one such setting, enabling classical inference after training. For a general family of these encodings, we also show the model avoids barren plateaus for a relevant subset of parameters. Importantly, in this regime the model is not known to be classically simulable, as it involves approximating ground states of local Hamiltonians. Although the shallow encoding may allow for classical surrogates, we exhibit a task on which a specific model instance provably achieves substantially higher test accuracy than a corresponding natural surrogate, at the same training-set size.
Moreover, even with shallow encoding unitaries, we show that the Hamiltonian embedding framework can implement universal approximators for measurable binary target functions. Specifically, for any such function and any desired accuracy, we provide a construction using finitely many qubits such that the hypothesis function implemented by the model achieves an arbitrarily small average approximation error under the uniform input distribution.

Irrespective of the encoding depth, we also show the model admits favorable generalization guarantees for a wide class of encoding unitaries. If training yields a sufficiently low-energy approximate ground state, matrix Hoeffding inequalities and covering number arguments establish generalization to unseen inputs. We show that for fixed generalization error, the required dataset size scales linearly, up to logarithmic factors, with the number of qubits and variational parameters.

Finally, the model can realize hypothesis functions that lie beyond the reach of efficient classical algorithms. Although preparing ground states of local Hamiltonians is hard in the worst case, even for quantum computers, we identify structured instances that encode classically hard learning problems for which the relevant ground states can nevertheless be efficiently approximated on a quantum computer. Under standard cryptographic assumptions, these instances yield explicit learning problems solvable on a quantum processor within the Hamiltonian embedding framework, while remaining classically hard. 

We note that restricting the encoding unitaries to those for which the relevant dressed-observable expectation values can be evaluated classically places Hamiltonian embedding within the \textit{shadow model} framework introduced in~\cite{jerbi2024shadows}.
A shadow QML model uses a quantum computer during training to obtain a quantum advice state, whose classical representation then supports entirely classical inference. Whereas Ref.~\cite{jerbi2024shadows} introduces shadow models as a general framework, our construction provides an explicit method to obtain the advice state, namely as the ground state of the embedding Hamiltonian.

Ground states of data-dependent Hamiltonians have also been used for classification in~\cite{ye2025quantum,umeano2026ground,candelori2026shallow}. In~\cite{ye2025quantum}, the encoding is fixed, and learning consists of preparing a low-energy state of the resulting data-encoding Hamiltonian, without optimizing over a trainable family of Hamiltonians. Our framework instead learns both the advice state and the Hamiltonian whose ground state realizes it. Furthermore, whereas~\cite{candelori2026shallow,umeano2026ground} encode individual inputs into separate Hamiltonians and use their ground states as quantum feature maps, our framework encodes the entire training dataset into a single Hamiltonian. This construction enables energy-based generalization guarantees, efficient gradient estimation, and efficient inference for suitable encodings.

To assess whether the model can be trained to achieve low validation and test error, we study it in numerical experiments on a standard ML task and implement it on IBM Heron processors. The results show successful optimization and competitive predictive performance relative to the quantum and classical baselines considered.

\medskip 

The manuscript is organized as follows.
\begin{itemize}
    \item In Section~\ref{sec: ham embedding model}, we introduce the Hamiltonian embedding model and provide a complete description of the training and prediction procedures. We also discuss the choice of the embedding unitaries and observables, derive an efficient gradient-estimation procedure, and establish the absence of barren plateaus for a relevant subset of parameters. 
    
    \item The main theoretical results are in Section~\ref{sec: properties}. We derive analytical generalization bounds and study classical surrogates of the trained model. We also identify learning tasks for which the physical constraints imposed by the learned quantum ground state act as a useful inductive bias, yielding improved sample efficiency over classically regularized surrogate models. Finally, we exhibit a learning task that is efficiently solvable within the Hamiltonian embedding framework but classically hard under a standard cryptographic assumption.
    
    \item In Section~\ref{sec: numerics}, we present numerical results for the binary Fashion-MNIST classification task. We benchmark the Hamiltonian embedding framework on systems up to 50 qubits against quantum kernels studied in Ref.~\cite{huang2021power} and classical machine-learning baselines. Across this range, the model remains trainable and achieves competitive test performance. We also implement the training procedure on IBM quantum processors, using hardware-generated samples in the Sample-based Krylov Quantum Diagonalization ground-state solver~\cite{yu2025quantum, mikelsen_SKQD2025, Sugisaki_SKQD2025}, and demonstrate successful training on a 50-qubit instance.
\end{itemize}

\section{Hamiltonian embedding model}\label{sec: ham embedding model}

Although our framework can be readily extended to multiclass classification, here we address the problem of binary classification. In particular,
we are given a training set $\mathcal{T}=\{(x_\ell,y_\ell)\}_{\ell=1}^N$, where each input $x_\ell$ is sampled from a distribution $\mathcal{D}$ over $\Omega\subseteq\mathbb{R}^d$, and the corresponding label is generated by an unknown function $f:\Omega\to\{0,1\}$, such that $y_\ell=f(x_\ell)$.
 The goal of the learning algorithm is to infer an approximate map $\Tilde{f}$ that, with high probability, agrees with the true mapping $f$ on previously unseen test data sampled from the same distribution $\mathcal{D}$. 

Given a dataset $\mathcal{T}$, the model constructs a family of dataset-dependent hypothesis functions as follows. On an $n$-qubit Hilbert space, we define an encoding unitary $U(x,\alpha)$ parameterized by both the input $x$ and trainable parameters $\alpha\in\mathbb{R}^M$, together with a family of label-dependent operators $\{O(y)\}_{y}$, assigning a measurement operator $O(y)$ to each class label $y\in\{0,1\}$. Then, for each data point $z=(x,y)$, with $y=f(x)$, we define a corresponding Hamiltonian term as $ H_{z}(\alpha) = U^{\dagger}(x,\alpha)O(y)U(x,\alpha)$. The data Hamiltonian on the whole dataset is defined as 
\begin{align}\label{eq: encoding_hamiltonian_T}
    H_{\T}(\alpha) &= \frac{1}{N}\left( \sum_{\ell=1}^N \left (I - U^{\dagger}(x_\ell,\alpha)O(y_\ell)U(x_\ell,\alpha)\right)\right)\\
    & = I - \frac{1}{N}\sum_\ell U^{\dagger}(x_\ell,\alpha)O(y_\ell)U(x_\ell,\alpha)\\
    & = I - \frac{1}{N}\sum_\ell H_{z_\ell}(\alpha)~.
\end{align}

For each $\alpha$, we define $\ket{\psi_{{\scriptscriptstyle\mathcal{T}},\alpha}}$ to be the ground state of $H_{\mathcal{T}}(\alpha)$. The state $\ket{\psi_{{\scriptscriptstyle\mathcal{T}},\alpha}}$ is then used to construct the corresponding hypothesis function implemented by the model. 
In particular, given a new input point $x'$, we apply the encoding unitary $U(x',\alpha)$ to $\ket{\psi_{{\scriptscriptstyle\mathcal{T}},\alpha}}$ and measure it with respect to $O(y)$, for $y\in \{0,1\}$. As we show in Section~\ref{sec: observable O}, we can choose \(O(0)+O(1)=I\), so that the expectation values of these positive operators on the encoded state define a normalized probability distribution over the labels. Specifically, the probability $P_{\psi_{\scalebox{0.8}{$\scriptscriptstyle\mathcal{T}$},\alpha},\alpha}(x',y)\in[0,1]$ assigned to label $y$ for input $x'$ is given by the following expression:
\begin{equation}
\label{eq: probability of y}
P_{\psi_{\scalebox{0.8}{$\scriptscriptstyle\mathcal{T}$},\alpha},\alpha}(x',y)
=
\operatorname{Tr}\!\left[
O(y)\,U(x',\alpha)
|\psi_{{\scriptscriptstyle\mathcal{T}},\alpha}\rangle
\langle\psi_{{\scriptscriptstyle\mathcal{T}},\alpha}|
U^\dagger(x',\alpha)
\right].
\end{equation}

The corresponding hypothesis function assigns a label $\tilde y\in\{0,1\}$ to the input $x'$ as follows:
\begin{align}\label{eq: prediction rule}
\tilde{f}_{\mathcal{T},\alpha}(x')
=
\operatorname*{arg\,max}_{y\in\{0,1\}}
P_{\psi_{\scalebox{0.8}{$\scriptscriptstyle\mathcal{T}$},\alpha},\alpha}(x',y).
\end{align}
For a given training set $\mathcal{T}$, the resulting family of candidate hypotheses is therefore
$
\mathcal H_{\mathcal T}
=
\{\tilde f_{\mathcal T,\alpha}:\alpha\in\mathbb{R}^M\}$.
We note that the family of hypothesis functions depends on the training data through the ground states of \(H_{\mathcal T}(\alpha)\). Moreover, $\mathcal H_{\mathcal T}$ is a subset of the broader, data-independent family of hypothesis functions obtained by allowing arbitrary normalized $n$-qubit states in the same classifier construction.

The model selects a hypothesis from $\mathcal H_{\mathcal T}$ by choosing $\alpha$ according to the following training procedure. Let us consider the case when $0\preceq O(y)\preceq I$ for every $y\in\{0,1\}$, then $ 0 \preceq \frac{1}{N}\sum_\ell H_{z_\ell}(\alpha)\preceq I$. This implies that the energy of $H_\mathcal{T}(\alpha)$ is always lower bounded by zero for any $\alpha$. In particular, the energy can reach the minimum value zero if the ground state $\ket{\psi_{{\scriptscriptstyle\mathcal{T}},\alpha}}$ of $H_{\mathcal{T}}(\alpha)$ is such that $\bra{\psi_{{\scriptscriptstyle\mathcal{T}},\alpha}}U^{\dagger}(x,\alpha)O(y)U(x,\alpha)\ket{\psi_{{\scriptscriptstyle\mathcal{T}},\alpha}}=1$ for every training point $(x_\ell,y_\ell)$ in $\mathcal{T}$.
The optimal parameter vector \(\alpha^\star\) is then chosen to minimize the ground-state energy over the family of Hamiltonians $\{H_{\mathcal T}(\alpha)\}_{\alpha}$.
As we see below, this guarantees the correct classification of the training data. Moreover, in Section~\ref{sec: generalization bound}, we show that our model also satisfies favorable generalization bounds.

In particular, let us define the empirical energy cost function as:
\begin{align}\label{eq:cost-function}
\widehat{C}_{\mathcal{T}}(\psi,\alpha) \equiv \mathbb{E}_{z\in \T} \left[\ell_{\text{energy}}(\psi, H_z(\alpha))\right] = \langle \psi | H_{\T}(\alpha) | \psi \rangle , 
\end{align}
where $\ell_\text{energy}$ is the energy loss, defined as:
\begin{align}\label{eq: energy_loss}
   \ell_{\text{energy}}(\psi, H_z(\alpha))=1-\langle \psi | H_{z}(\alpha) | \psi \rangle . 
\end{align}
 Then the training objective is to find the parameter value $\alpha^\star$ and the state $\psi_{{\scriptscriptstyle\mathcal{T}},\alpha^\star}$ that minimize the cost function $\widehat C_\mathcal{T}(\psi,\alpha)$.  We note that the algorithm may not reach the exact ground state of the training Hamiltonian $H_{\mathcal{T}}(\alpha^\star)$. In other words, the learned state $\psi_{{\scriptscriptstyle\mathcal{T}},\alpha^\star}$ may only be an approximate ground state. However, as we show in Section~\ref{sec: generalization bound}, it is sufficient for the learned state to achieve a low energy with respect to the training Hamiltonian: a small empirical energy cost already guarantees good performance on unseen input points.
In practice, this optimization is carried out through an alternating nested procedure. Starting from an initial value of $\alpha$, the algorithm first computes the corresponding ground state of $H_{\mathcal T}(\alpha)$ and then updates $\alpha$ to decrease the cost. Repeating this procedure yields a sequence of parameter values and associated ground states until convergence to a final pair $(\alpha^\star,\psi_{{\scriptscriptstyle\mathcal{T}},\alpha^\star})$. 
We detail the whole learning procedure in Algorithm~\ref{alg:ham-embed}. Although reaching zero energy loss guarantees correct classification on the training dataset, a decrease in the empirical energy cost $\widehat{C}_{\mathcal{T}}$ as $\alpha$ is updated does not necessarily lead to a simultaneous increase in classification accuracy. We discuss this mismatch in Section~\ref{sec: loss function}, where we also consider alternative loss functions that more directly target the classification error.

Once training is complete, the learned parameters $\alpha^\star$ and ground state $|\psi_{{\scriptscriptstyle\mathcal{T}},\alpha^\star}\rangle$ define the hypothesis used to predict labels for unseen inputs. An important feature of the Hamiltonian embedding framework is that, for particular classes of encoding unitaries $U(x,\alpha)$, these predictions can be evaluated entirely classically. As discussed in Section~\ref{sec: classical inference}, this is achieved by constructing a classical surrogate $\sigma(\psi_{{\scriptscriptstyle\mathcal{T}},\alpha^\star})$ of the learned quantum model from a finite number of measurements on the quantum device. The corresponding classical inference procedure is summarized in Algorithm~\ref{alg:ham-embed-inference}.

\begin{algorithm}[htbp]
\caption{Supervised Training via Hamiltonian Embedding}
\label{alg:ham-embed}
\begin{algorithmic}[1]

\Input
\Statex \hspace{\algorithmicindent}
1. Training set
$\mathcal T=\{(x_\ell,y_\ell)\}_{\ell=1}^N$,
$x_\ell\in\mathbb{R}^d$ and $y_\ell\in\{0,1\}$.
\Statex \hspace{\algorithmicindent}
2. Encoding unitary $U(x,\alpha)$.
\Statex \hspace{\algorithmicindent}
3. Desired training energy $\varepsilon\geq0$,
learning rate $\eta>0$, maximum number of iterations $T_{\max}$.

\Output Trained parameters $\alpha^\star$ and final state
$\ket{\psi_{{\scriptscriptstyle\mathcal{T}},\alpha^\star}}$.

\vspace{2mm}

\State
$O(y)\gets \frac{1}{2}\left[I+(-1)^yA\right]$
\Comment{Construct the label observable}

\State
$\alpha_0\gets\mathrm{Unif}([-1,1]^M)$
\Comment{Initialize the variational parameters}

\State $t\gets0$

\While{$t<T_{\max}$}

    \For{$\ell=1$ to $N$}
        \State
        $z_\ell\gets(x_\ell,y_\ell)$
        \State
        $H_{z_\ell}(\alpha_t)
        \gets
        U^\dagger(x_\ell,\alpha_t)
        O(y_\ell)
        U(x_\ell,\alpha_t)$
        \Comment{Embed the training point}
        \State Store a classical representation of
        $H_{z_\ell}(\alpha_t)$
    \EndFor

    \State
    $H_{\mathcal T}(\alpha_t)
    \gets
    I-\frac{1}{N}
    \sum_{\ell=1}^N H_{z_\ell}(\alpha_t)$
    \Comment{Construct the current embedding Hamiltonian}

    \State
    $\displaystyle
    \ket{\psi_{{\scriptscriptstyle\mathcal{T}},\alpha_t}}
    \approx
    \arg\min_{\ket{\psi}}
    \bra{\psi}H_{\mathcal T}(\alpha_t)\ket{\psi}$
    \Comment{Find an approximate ground state}

    \State
    $E_t
    \gets
    \bra{\psi_{{\scriptscriptstyle\mathcal{T}},\alpha_t}}
    H_{\mathcal T}(\alpha_t)
    \ket{\psi_{{\scriptscriptstyle\mathcal{T}},\alpha_t}}$
    \Comment{Evaluate the current training energy}

    \If{$E_t\leq\varepsilon$}
        \State \textbf{break}
    \EndIf

    \For{$\ell=1$ to $N$}
        \State
        $m_{\alpha_t,\psi_{{\scriptscriptstyle\mathcal{T}},\alpha_t}}(x_\ell,y_\ell)
        \gets
        \bra{\psi_{{\scriptscriptstyle\mathcal{T}},\alpha_t}}
        H_{z_\ell}(\alpha_t)
        \ket{\psi_{{\scriptscriptstyle\mathcal{T}},\alpha_t}}$
        \Comment{Evaluate the correct label score}
    \EndFor

    \State
    Evaluate
    $\widehat C_{\mathcal T}
    (\psi_{{\scriptscriptstyle\mathcal{T}},\alpha_t},\alpha_t)$
    according to Eq.~\eqref{eq:cost-function}

    \For{$j=1$ to $M$}
        \State
        $\displaystyle
        g_j(t)
        \gets
        \left.
        \frac{\partial}{\partial\alpha_j}
        \widehat C_{\mathcal T}
        \left(\psi_{{\scriptscriptstyle\mathcal{T}},\alpha_t},\alpha\right)
        \right|_{\alpha=\alpha_t}$
        \Comment{\parbox[t]{6.8cm}{
        Estimate the gradient via parameter shift,
        keeping the state fixed}}
    \EndFor

    \State
    $g(t)\gets
    \bigl(g_1(t),\ldots,g_M(t)\bigr)$

    \State
    $\alpha_{t+1}
    \gets
    \alpha_t-\eta g(t)$
    \Comment{Update the variational parameters}

    \State $t\gets t+1$

\EndWhile

\State
$\alpha^\star\gets\alpha_t$

\State
$H_{\mathcal T}(\alpha^\star)
\gets
I-\frac{1}{N}\sum_{\ell=1}^N
H_{z_\ell}(\alpha^\star)$

\State
$\displaystyle
\ket{\psi_{{\scriptscriptstyle\mathcal{T}},\alpha^\star}}
\approx
\arg\min_{\ket{\psi}}
\bra{\psi}H_{\mathcal T}(\alpha^\star)\ket{\psi}$

\State
\Return $(\alpha^\star,\ket{\psi_{{\scriptscriptstyle\mathcal{T}},\alpha^\star}})$

\end{algorithmic}
\end{algorithm}

\begin{algorithm}[htbp] \caption{Classical inference via Hamiltonian embedding} \label{alg:ham-embed-inference} \begin{algorithmic}[1] 
\Input \Statex \hspace{\algorithmicindent} 1. New input point $x'$. 
\Statex \hspace{\algorithmicindent} 2. Trained parameters $\alpha^\star$. 
\Statex \hspace{\algorithmicindent} 3. Classical representation $ \sigma(\psi_{{\scriptscriptstyle\mathcal{T}},\alpha^\star})$ of the learned state $\ket{\psi_{{\scriptscriptstyle\mathcal{T}},\alpha^\star}}$.
\Statex \hspace{\algorithmicindent} 4. Label-dependent observable $O(y)$ for $y\in\{0,1\}$. 
\Output Predicted label $\tilde y\in\{0,1\}$. 
\vspace{2mm} 

\For{$y\in\{0,1\}$} 

\State  $ \widetilde O(x',y) \leftarrow U^\dagger(x',\alpha^\star)O(y)U(x',\alpha^\star) $ \Comment{Construct observable in \cref{eq: dressed O(y)} for $x'$ associated with label $y$} 
\State  $ \widetilde O(x',y) \leftarrow \sum_i c_i(x',y,\alpha^\star)P_i $ \Comment{Decompose $O'$ into local Pauli operators} 
\State $ p_y(x') \leftarrow \sum_i c_i(x',y,\alpha^\star) \operatorname{Tr}\!\left[P_i\, \sigma  (\psi_{{\scriptscriptstyle\mathcal{T}},\alpha^\star})\right] $ \Comment{Estimate the label score using the classical representation $\sigma(\psi_{{\scriptscriptstyle\mathcal{T}},\alpha^\star})$ } 
\EndFor 
\State  $ \tilde y \leftarrow \arg\max_{y\in\{0,1\}} p_y(x') $   \Comment{Output the label with largest score}
\State \Return $\tilde y$ 
\end{algorithmic} 
\end{algorithm}

\subsection{Construction of the model}\label{sec: construction model}

The Hamiltonian embedding framework provides substantial freedom in how the encoding Hamiltonian is constructed, allowing problem-specific choices to be incorporated. In the following, we outline practical guidelines for selecting the label-dependent observable $O(y)$, the encoding unitary $U(x,\alpha)$, the choice of the loss function and the ground-state optimization method. 

\subsubsection{Choice of the observable}\label{sec: observable O}
The label-dependent observable $O(y)$ is measured at inference time to determine which label $y\in\{0,1\}$ should be assigned to a given input $x$. Recall that the energy loss in \cref{eq: energy_loss} is minimized when, for every training point $(x_\ell,y_\ell)$ with $y_\ell=f(x_\ell)$, the learned ground state $\ket{\psi_{{\scriptscriptstyle\mathcal{T}},\alpha^\star}}$ satisfies
\begin{align}\label{eq: perfect_classification}
\operatorname{Tr}[O(y_\ell) U(x_\ell,\alpha^\star)|\psi_{{\scriptscriptstyle\mathcal{T}},\alpha^\star}\rangle \langle \psi_{{\scriptscriptstyle\mathcal{T}},\alpha^\star}|U^{\dagger}(x_\ell,\alpha^\star)]=1.
\end{align}
This occurs when, for every training example $x_\ell$, the encoded state $U(x_\ell,\alpha^\star)\ket{\psi_{{\scriptscriptstyle\mathcal{T}},\alpha^\star}}$ lies in the eigenvalue-$1$ eigenspace of the observable $O(y)$ associated with the correct label $y\in\{0,1\}$. We note that Eq.~\eqref{eq: perfect_classification} provides a sufficient, but not necessary, condition for correct classification. Since prediction in Eq.~\eqref{eq: prediction rule} is performed by selecting the label associated with the largest expectation value, it is enough that the score assigned to the correct label exceeds that of the other label.
To give these observables a direct probabilistic interpretation, we can impose the natural condition
$O(0)+O(1)=I$,
so that their expectation values in \cref{eq: probability of y} define a normalized probability distribution over the two classes. This condition is equivalent to parameterizing the two label observables in terms of a single Hermitian operator $A$ as
\begin{equation}\label{eq: observable O}
    O(y)=\frac{I+(-1)^y A}{2},
\end{equation}
where $y\in\{0,1\}$ labels the two classes. We take $A$ to be normalized such that its spectrum lies in $[-1,1]$, with extremal eigenvalues $\lambda_{\min}(A)=-1$ and $\lambda_{\max}(A)=1$. We note that an operator $A$ with this normalization can always be obtained whenever the extremal eigenvalues $\lambda_{\min}$ and $\lambda_{\max}$ of the original Hermitian operator are known, by applying the affine transformation
$
A \mapsto \frac{2A-(\lambda_{\max}+\lambda_{\min})I}{\lambda_{\max}-\lambda_{\min}}$.
 The choice of $A$ plays a central role in the construction of the model, as it determines how the two classes are represented in Hilbert space and can therefore be tailored to the learning task under consideration. 
 We note however that the condition $O(0)+O(1)=I$ is not necessary to define the model as labels can be assigned by comparing the expectation values of the two operators, without interpreting these values as probabilities.
 As we show in Appendix~\ref{sec: example_entanglement} and Appendix~\ref{app: provable sep surrogate}, the choice of $A$, and $O(y)$ in general, can influence properties of the learned ground state, such as its entanglement structure, and can thereby impact the classical simulability of the resulting model.

From a practical point of view, the choice of the operator $A$ and the corresponding observable $O(y)$ affects the form of the final embedding Hamiltonian. Since conjugation by a constant-depth unitary $U(x,\alpha)$ increases the support of a local observable only by a constant factor, choosing $A$ to be local ensures that the resulting embedding Hamiltonian remains local and with polynomially many terms.

In our numerical experiments, we choose the operator $A$ to be a simple diagonal Hamiltonian of the form  
\begin{equation}\label{eq: A experiment}
A
=
\frac{1}{n}
\sum_{i=1}^n Z_i.
\end{equation}
With this convention, the model perfectly classifies the training set whenever, for each training example, the learned ground state $\ket{\psi_{{\scriptscriptstyle\mathcal{T}},\alpha^\star}}$ is mapped by the data-encoding unitary to a state in the eigenspace of $O(y)$ associated with the correct label.
In the particular case of the choice in Eq.~\eqref{eq: A experiment}, a zero energy cost is obtained when, for each training point $(x_\ell,y_\ell)\in\mathcal{T}$
\begin{align}
U(x_\ell,\alpha^\star)\ket{\psi_{{\scriptscriptstyle\mathcal{T}},\alpha^\star}}
=
\begin{cases}
\ket{0^n} & y_\ell=0,\\
\ket{1^n} & y_\ell=1.
\end{cases}
\end{align}
As discussed above, the choice of $A$ can influence the properties of the ground states learned by the model. For example, in Appendices~\ref{sec: example_entanglement} and~\ref{app: provable sep surrogate}, we present alternative choices of $O(y)$ that lead to highly entangled learned ground states.

\subsubsection{Choice of the unitary embedding}

As discussed above, the embedded Hamiltonian $H_\T(\alpha)$ is determined by the choice of the parameterized encoding unitary $U(x,\alpha)$. This choice influences the expressiveness of the model, and different encoding strategies may be better suited to different learning tasks.
Since training requires the ground state of the embedding Hamiltonian to be prepared at each optimization step, the corresponding quantum ground-state preparation procedure should be efficient.
A necessary condition is that the embedding Hamiltonian admit an efficient classical description suitable as input to the quantum ground-state preparation algorithm. For example, when $O(y)$ can be expressed as a sum of local Pauli operators, choosing $U(x,\alpha)$ to have constant depth ensures that each data-dependent term $H_{(x,y)}(\alpha)=U^\dagger(x,\alpha)O(y)U(x,\alpha)$ remains of bounded locality $O(1)$. In Appendix~\ref{app:shallow-universality}, we show that the model achieves universality even with constant-depth encodings $U$ and local $O(y)$, using a finite number of qubits whose scaling however need not be efficient.
We stress, however, that a constant-depth encoding unitary is not strictly required by our framework.  For instance, ground-state solvers based on quantum phase estimation only require
the ability to efficiently implement controlled time evolution under the
target Hamiltonian. Therefore, the encoding unitaries $U(x,\alpha)$ may
have depth growing with the system size, provided that they admit an
efficient circuit description from which the corresponding Hamiltonian
evolution can be implemented efficiently. In particular, this description
need not be restricted to a decomposition in terms of Pauli operators.

Another important consideration in the choice of the parameterized unitary $U(x,\alpha)$ is that the cost function should admit efficient differentiation through the parameter-shift rule, as will be discussed in Section~\ref{sec: barren plateaus and gradient}. For this to hold, the trainable parameters should enter the circuit through gates generated by operators with a finite, typically two-point, spectrum, such as Pauli rotations of the form $e^{-i\theta P/2}$ with $P^2=I$.

Keeping all of this in consideration, in the numerical experiments in this work we choose the following encoding. The unitary $U(x,\alpha)$ is taken to be a constant-depth circuit composed of layers of bounded-locality Pauli rotations, whose angles depend on both the input $x\in\mathbb{R}^d$ and the trainable parameters $\alpha\in\mathbb{R}^M$. Concretely, we consider circuits of the form:

\begin{equation}\label{eq: encoding U}
U(x,\alpha)
=
\prod_{l=1}^{K}
\prod_{g=1}^{G}
\exp\!\left[
-\frac{i}{2}
\left(
\alpha_{g,0}^{l}
+
\sum_{r=1}^{d}
\alpha_{g,r}^{l} x^r
\right)
P_g^l
\right],
\end{equation}
where $K=O(1)$ denotes the number of layers, $G=O(n)$ the number of
parameterized gates per layer, and $P_g^\ell$ are bounded-locality Pauli
operators. The products are taken according to the prescribed ordering
of the gates in the circuit.  In the numerical experiments, we instantiate Eq.~\eqref{eq: encoding U} with $K=3$ with each layer applying, on every qubit $j$, an independently parameterized $Z_j\!-\!Y_j\!-\!Z_j$ sequence of single-qubit rotations, followed by a nearest-neighbor $Z_jZ_{j+1}$ rotation. This architecture satisfies the locality requirement by construction: conjugating a one-local observable through a constant-depth local circuit only expands its support to a constant-size light cone. At the same time, because every trainable gate is a Pauli rotation, derivatives of the cost can be evaluated using the parameter-shift rule. 

Finally, we note that the choice of the encoding unitary $U(x,\alpha)$, together with the observables $O(y)$, determines whether inference on new input points can be performed entirely classically. This is the case, for instance, when $U(x,\alpha)$ has constant depth and $O(y)$ is local. Efficient classical inference is not restricted to this setting, and certain classes of deeper encoding unitaries also admit efficient classical surrogates. We discuss these different regimes in Section~\ref{sec: classical inference}. Furthermore, in the classically hard task discussed in Section~\ref{sec: hard problem}, neither the embedded Hamiltonian nor the dressed observable $\widetilde O(y)$ is local. Nonetheless, the construction still admits an efficient quantum procedure for preparing the relevant ground state, as well as a problem-specific classical procedure for representing and classically evaluating the corresponding hypothesis.

\subsubsection{Choice of the loss function}
\label{sec: loss function}
The Hamiltonian embedding framework naturally defines the energy loss introduced in Eq.~\eqref{eq: energy_loss}. For fixed training parameters $\alpha$ and state $\psi_{{\scriptscriptstyle\mathcal{T}},\alpha}$, its empirical average over the training set is precisely the energy of the corresponding embedding Hamiltonian $H_{\mathcal{T}}(\alpha)$. This choice is particularly convenient, since minimizing the loss is directly connected to finding low-energy states of the embedded Hamiltonian $H_\mathcal{T}(\alpha)$. Moreover, reaching zero energy loss on a training example guarantees that the corresponding point is classified correctly. However, the energy loss is not identical to the classification error and, in particular, a decrease in energy is not necessarily accompanied by a simultaneous increase in classification accuracy.
To see this, let us consider the prediction rule in Eq.~\eqref{eq: prediction rule}. In order for a point $x$ to be correctly classified, it is sufficient that the model finds a state $\psi_{{\scriptscriptstyle\mathcal{T}},\alpha^\star}$ such that $P_{\psi_{{\scriptscriptstyle\mathcal{T}},\alpha^\star},\alpha^\star}(x,y)>0.5$ if $y$ is the correct label. This happens whenever the margin $m_{\alpha^\star,\psi_{{\scriptscriptstyle\mathcal{T}},\alpha^\star}}(x,y)= \langle \psi_{{\scriptscriptstyle\mathcal{T}},\alpha^\star}|U^{\dagger}(x,\alpha^\star)O(y) U(x,\alpha^\star)|\psi_{{\scriptscriptstyle\mathcal{T}},\alpha^\star}\rangle $ is larger than $0.5$. The training classification error is then defined as  
\begin{equation}
\widehat R_{\mathcal T}(\psi,\alpha)
=
\frac{1}{N}
\sum_{\ell=1}^{N}
\mathbb{1}
\left[
m_{\alpha,\psi}(x_\ell,y_\ell)\leq 0.5
\right],
\label{eq: empirical classification error}
\end{equation}
which only depends on whether each example lies on the correct side of the decision boundary. We recall that, at each optimization step $t$, the state $\psi_{{\scriptscriptstyle\mathcal{T}},\alpha_t}$ is held fixed while the encoding parameters $\alpha$ are updated. For a fixed state $\psi$, minimizing the energy loss $\ell_{\mathrm{energy}}$ in Eq.~\eqref{eq: energy_loss} favors parameters $\alpha$ that increase the margins $m_{\alpha,\psi}(x,y)$ of the training examples $(x,y)\in\mathcal{T}$. However, since the energy loss $\ell_{\mathrm{energy}}$ is linear in $m_{\alpha,\psi}(x,y)$, it continues to reward increases in the margin even after a sample has already been classified correctly, i.e. when $m_{\alpha,\psi}(x,y)>0.5$. Consequently, the model may continue to lower the energy by increasing the confidence assigned to points that are already correctly classified, without this necessarily leading to a corresponding reduction in the empirical classification error of Eq.~\eqref{eq: empirical classification error}.

Directly minimizing Eq.~\eqref{eq: empirical classification error} is not well suited to gradient-based optimization, since the error $\hat{R}_{\mathcal{T}}$ changes only when the score $m_{\alpha,\psi}(x_\ell,y_\ell)$ of a training example crosses the classification threshold, and therefore provides no useful gradient for most parameter values.
As in classical supervised learning, one can instead optimize a continuous surrogate of the $0$--$1$ loss that more directly reflects the classification objective. We introduce two alternative losses as possible extensions aimed more directly at reducing classification error, although the resulting parameter updates need not decrease the Hamiltonian energy.

For example, let us define a centered margin $\tilde m_{\alpha,\psi}(x,y)=2m_{\alpha,\psi}(x,y)-1$ so that $\tilde m_{\alpha,\psi}\in [-1,1]$. We then consider the logistic loss
\begin{equation}
\ell_{\mathrm{log}}(\tilde m)
=
\log\left(1+e^{-\beta \tilde m}\right),
\label{eq: logistic loss}
\end{equation}
where $\beta>0$ controls the sharpness of the loss around the decision boundary, and the hinge loss
\begin{equation}
\ell_{\mathrm{hinge}}^{(\gamma)}(\tilde m)
=
\max\{0,\gamma-\tilde m\},
\label{eq: hinge loss}
\end{equation}
where $\gamma>0$ specifies a desired classification margin. In contrast to the energy loss, the hinge loss assigns no additional penalty to examples whose margin already exceeds $\gamma$, thereby focusing the optimization on misclassified points and on correctly classified points lying close to the decision boundary. Similarly, the logistic loss progressively suppresses the contribution of examples with large positive margins while assigning greater weight to misclassified examples and to points close to the decision boundary.

More generally, given a surrogate loss $\ell_{\mathrm{cls}}$, we define the corresponding empirical classification-oriented cost as
\begin{equation}
\widehat C_{\mathcal T}^{\mathrm{cls}}(\psi,\alpha)
=
\frac{1}{N}
\sum_{\ell=1}^{N}
\ell_{\mathrm{cls}}
\left(
\tilde m_{\psi,\alpha}(x_\ell,y_\ell)
\right).
\label{eq: classification surrogate cost}
\end{equation}
Importantly, introducing this alternative cost does not modify the embedding Hamiltonian or the ground-state preparation step. At optimization iteration $t$, the state is still obtained as an approximate ground state of the embedding Hamiltonian,
\begin{equation}
\ket{\psi_{{\scriptscriptstyle\mathcal{T}},\alpha_t}}
\approx
\arg\min_{\ket{\psi}}
\bra{\psi}
H_{\mathcal T}(\alpha_t)
\ket{\psi}.
\end{equation}
The resulting state is then kept fixed while the encoding parameters are updated using the gradient of the chosen loss,
\begin{equation}
\alpha_{t+1}
=
\alpha_t
-
\eta
\nabla_\alpha
\widehat C_{\mathcal T}^{\mathrm{cls}}
\left(
\psi_{{\scriptscriptstyle\mathcal{T}},\alpha_t},\alpha
\right)
\bigg|_{\alpha=\alpha_t}.
\label{eq: classification loss parameter update}
\end{equation}
We note that the Hamiltonian embedding algorithm in Algorithm~\ref{alg:ham-embed} is recovered by choosing $\ell_{\mathrm{cls}}=\ell_{\mathrm{energy}}$.
Furthermore, one may also interpolate between the two objectives through
\begin{equation}
\widehat C_{\mathcal T}^{(\lambda)}
=
(1-\lambda)
\widehat C_{\mathcal T}
+
\lambda
\widehat C_{\mathcal T}^{\mathrm{cls}},
\qquad
\lambda\in[0,1],
\label{eq: combined loss}
\end{equation}
which allows the parameter optimization to balance the search for low-energy Hamiltonians with a loss that more directly targets classification performance.

The choice of the outer-loop loss therefore provides an additional degree of freedom in the Hamiltonian embedding framework: the embedding Hamiltonian determines the state used by the model, while the loss function determines how the parameters are optimized within the family of embedding Hamiltonians. Unless stated otherwise, the numerical results presented in Section~\ref{sec: numerics} are obtained using the canonical energy-based objective in Eq.~\eqref{eq: energy_loss}, which we find to be sufficient to achieve good classification performance on the tasks considered. We also note that replacing the energy cost in Eq.~\eqref{eq:cost-function} with either the hinge or logistic cost does not increase the complexity of gradient estimation. Indeed, for a generic differentiable surrogate loss $\ell_{\text{cls}}(m)$, the derivative with respect to a variational parameter $\alpha_j$ can be written using the chain rule as
\begin{equation}
    \frac{\partial}{\partial\alpha_j}
    \ell_{\mathrm{cls}}
    \left(\widetilde m_{\alpha,\psi}(x,y)\right)
    =
    \ell_{\mathrm{cls}}'
    \left(\widetilde m_{\alpha,\psi}(x,y)\right)
    \frac{\partial\widetilde m_{\alpha,\psi}(x,y)}
         {\partial\alpha_j}.
\end{equation}
Since
$
    \widetilde m_{\alpha,\psi}(x,y)
    =
    2
    \langle\psi|
    U^\dagger(x,\alpha)O(y)U(x,\alpha)
    |\psi\rangle
    -1$,
we have
$
    \frac{\partial\widetilde m_{\alpha,\psi}(x,y)}
         {\partial\alpha_j}
    =
    2
    \frac{\partial}{\partial\alpha_j}
    \langle\psi|
    U^\dagger(x,\alpha)O(y)U(x,\alpha)
    |\psi\rangle$.

The first factor is a simple classical function of the margin, while the second corresponds to the derivative of the expectation value
$
m_{\alpha,\psi}(x,y)
=
\langle \psi |
U^\dagger(x,\alpha) O(y) U(x,\alpha)
|\psi\rangle $.

As discussed in Section~\ref{sec: barren plateaus and gradient}, this derivative can be efficiently estimated through the parameter-shift rule. Therefore, the same quantum measurements used to estimate gradients of the energy cost can also be used to evaluate the gradients of the alternative surrogate losses, followed by a simple classical post-processing step.

\subsubsection{Ground-state approximation}
At the heart of the Hamiltonian embedding approach is the use of a quantum computer to prepare the ground state of the embedding Hamiltonian $H_\T(\alpha)$ for varying parameter values $\alpha$. Several ground-state preparation methods can be employed, and the most suitable choice depends on the structure of $H_\T(\alpha)$ and on the capabilities of the available quantum device.
In the general setting of local Hamiltonians, a standard approach to ground-state preparation is quantum phase estimation (QPE)~\cite{cleve1998quantum,abrams1999quantumalgorithmprovidingexponential}. Given an efficient implementation of the time evolution generated by $H_\T(\alpha)$ and an initial state with non-negligible overlap with the ground state, QPE resolves the energy eigenvalues and projects the state onto the corresponding eigenspaces, producing the ground state with probability proportional to the initial overlap. When the Hamiltonian is gapped and the overlap is sufficiently large, QPE provides a systematic procedure for preparing the ground state with controllable precision.

The Sample-based Krylov Quantum Diagonalization (SKQD) algorithm offers an alternative that avoids phase estimation and also has rigorous guarantees in certain regimes~\cite{yu2025quantum}.
In the SKQD algorithm, time-evolution circuits are applied at multiple time steps and measured in a chosen basis, such as the computational basis, producing a set of sampled bitstrings that defines a reduced subspace onto which the Hamiltonian is projected and classically diagonalized. This procedure is compatible with current quantum hardware and is robust against device noise, provided the sampled subspace retains sufficient support on the low-energy sector.

Both methods require an initial reference state with non-negligible overlap with the true ground state, ideally at least inverse-polynomial in the system size. In practice, the applicability of these ground-state solvers therefore depends on the embedding Hamiltonians remaining in a region of parameter space where reference states with sufficiently large ground-state overlap can be efficiently obtained. For structured Hamiltonian families, such reference states may be constructed by mean-field or tensor-network methods such as the density matrix renormalization group (DMRG)~\cite{white1992density,schollwock2005density}. Identifying families of embedding Hamiltonians that simultaneously admit efficient quantum ground-state preparation while remaining difficult to simulate classically remains an important open question.

Finally, the ground state of the data Hamiltonian could also be approximated through adiabatic preparation. Assuming that a ground state of $H_{\mathcal{T}}(\alpha_0)$ is available, one could obtain a ground state of $H_{\mathcal{T}}(\alpha_1)$ entirely quantumly by following an adiabatic path in parameter space. Efficient preparation requires, in particular, that the spectral gap along this path remain bounded below by an inverse polynomial in the system size. Exploring choices of the embedding $U$ and observables $O$ for which such gapped paths can be guaranteed is an interesting direction for future work.

\subsection{Efficient gradient estimation and absence of barren plateaus}\label{sec: barren plateaus and gradient}
One of the main bottlenecks in training variational quantum ans\"atze, for example parameterized quantum circuits, is the presence of barren plateaus in the optimization landscape~\cite{mcclean2018barren,qi2023barren}. In this regime, gradients become exponentially small with the system size. Since gradients are estimated from measurement outcomes on a quantum computer, identifying the correct descent direction may then require exponentially many measurement shots. In the Hamiltonian embedding model, gradients can be efficiently evaluated through the parameter-shift rule. Furthermore, for a class of Pauli-rotation encodings satisfying the conditions stated below, we can ensure that a relevant subset of trainable parameters avoids barren plateaus. We also note that, in contrast to standard PQC settings, the use of a shallow encoding circuit does not by itself imply efficient classical simulability of the framework, as finding the ground state of the corresponding local data Hamiltonian may remain hard.

The gradient expression can be derived directly from the cost function in \cref{eq:cost-function}. As described in Algorithm~\ref{alg:ham-embed}, the optimization proceeds through a nested procedure. Given an initial parameter value $\alpha_0$, the algorithm first computes the ground state of the encoding Hamiltonian $H_{\mathcal{T}}(\alpha_0)$, outputting $\psi_{{\scriptscriptstyle\mathcal{T}},\alpha_0}$. The state is then kept fixed while the parameters are updated by minimizing the cost function $\widehat{C}(\psi_{{\scriptscriptstyle\mathcal{T}},\alpha_0},\alpha)$. In particular, the algorithm searches for a new parameter value $\alpha_1$ for which the energy of $H_{\mathcal{T}}(\alpha_1)$ evaluated on $\psi_{{\scriptscriptstyle\mathcal{T}},\alpha_0}$ is minimized. As derived in Appendix~\ref{app: gradient_cost}, this update is guided by the gradient of the cost $\widehat{C}(\psi_{{\scriptscriptstyle\mathcal{T}},\alpha_0},\alpha)$ in \cref{eq:cost-function} with respect to parameter $\alpha_j$. Recalling the structure of $H_\T(\alpha)$ from \cref{eq: encoding_hamiltonian_T}, we obtain:
\begin{equation}\label{eq: derivative_cost}
\frac{\partial}{\partial \alpha_j}
\widehat{C}(\psi_{{\scriptscriptstyle\mathcal{T}},\alpha_0},\alpha)
=
-\frac{1}{N}
\sum_{\ell=1}^N
\bra{\psi_{{\scriptscriptstyle\mathcal{T}},\alpha_0}}
\frac{\partial H_{z_\ell}(\alpha)}{\partial \alpha_j}
\ket{\psi_{{\scriptscriptstyle\mathcal{T}},\alpha_0}},
\end{equation}
where $N=|\T|$ is the size of the full training set. 
In the encoding unitary $U(x,\alpha)$ in Eq.~\eqref{eq: encoding U}, each parameterized gate has a rotation angle that depends linearly on the input features. For notational simplicity, we omit the constant bias term in the following derivation. The bias parameters are treated identically by regarding them as coefficients of an additional constant feature $x^0=1$. For an input $x_i=(x_i^1,\ldots,x_i^d)$, the angle of the $j$-th parameterized gate is
\begin{equation}
    \theta_j(x_i,\alpha)=\sum_{k=1}^d \alpha_{j,k} x_i^k
\end{equation}
where $\alpha_{j,k}$ is the trainable coefficient associated with the $k$-th feature in the $j$-th gate. We can therefore decompose the full encoding unitary around this gate as
\begin{equation}
 U(x_i,\alpha) = U_{L_j}(x_i,\alpha)\, U_j\!\left(\theta_j(x_i,\alpha)\right)\, U_{R_j}(x_i,\alpha),
\end{equation}
where $U_{R_j}(x_i,\alpha)$ contains the gates applied before the $j$-th parameterized gate, while $U_{L_j}(x_i,\alpha)$ contains those applied after it. Since $U_j$ is a Pauli rotation gate, its dependence on $\theta_j$, and consequently on each parameter $\alpha_{j,k}$, can be differentiated using the standard parameter-shift rule.
 Since differentiating the gate angle with respect to $\alpha_{j,k}$ contributes a factor $x_i^k$, and assuming that $U_j$ is a Pauli rotation gate, the
parameter-shift rule gives
\begin{equation}
\frac{\partial \widehat{C}}{\partial \alpha_{j,k}}
=
-\frac{1}{2N}
\sum_{i=1}^{N}
x_i^k
\left[
\widehat{C}_i^{(j,+)}
-
\widehat{C}_i^{(j,-)}
\right],
\end{equation}
where
\begin{align}
\widehat{C}_i^{(j,+)}
&=
\bra{\psi_{{\scriptscriptstyle\mathcal{T}},\alpha_0}}
U_{R_j}^\dagger
U_j^\dagger\!\left(\sum_k \alpha_{j,k} x_i^k+\frac{\pi}{2}\right)
U_{L_j}^\dagger
O(y_i)
U_{L_j}
U_j\!\left(\sum_k \alpha_{j,k} x_i^k+\frac{\pi}{2}\right)
U_{R_j}
\ket{\psi_{{\scriptscriptstyle\mathcal{T}},\alpha_0}},
\\
\widehat C_i^{(j,-)}
&=
\bra{\psi_{{\scriptscriptstyle\mathcal{T}},\alpha_0}}
U_{R_j}^\dagger
U_j^\dagger\!\left(\sum_k \alpha_{j,k} x_i^k-\frac{\pi}{2}\right)
U_{L_j}^\dagger
O(y_i)
U_{L_j}
U_j\!\left(\sum_k \alpha_{j,k} x_i^k-\frac{\pi}{2}\right)
U_{R_j}
\ket{\psi_{{\scriptscriptstyle\mathcal{T}},\alpha_0}}.
\end{align}
where $O(y_i)$ is the same labeled local observable as in \cref{eq: observable O}. After the parameter-update step, the algorithm returns a new parameter value $\alpha_{t+1}$. It then constructs the updated Hamiltonian $H_{\mathcal T}(\alpha_{t+1})$ and computes a new ground state $\psi_{{\scriptscriptstyle\mathcal{T}},\alpha_{t+1}}$. This state is used to evaluate the updated cost $\widehat C_{\mathcal T}(\psi_{{\scriptscriptstyle\mathcal{T}},\alpha_{t+1}},\alpha_{t+1})$, and the optimization loop continues.
A key point is that, during each parameter update, the current ground state is kept fixed. Thus, when computing the gradient with respect to $\alpha$, we do not differentiate through the dependence of $\psi_{{\scriptscriptstyle\mathcal{T}},\alpha}$ on $\alpha$. At iteration $t$, the algorithm first finds $\psi_{{\scriptscriptstyle\mathcal{T}},\alpha_t}$ for the Hamiltonian $H_{\mathcal T}(\alpha_t)$ and then updates the parameters by minimizing the frozen-state cost $\widehat C_{\mathcal T}(\psi_{{\scriptscriptstyle\mathcal{T}},\alpha_t},\alpha)$. This update rule is variationally well motivated. If the new parameters satisfy $\widehat C_{\mathcal T}(\psi_{{\scriptscriptstyle\mathcal{T}},\alpha_t},\alpha_{t+1}) \leq \widehat C_{\mathcal T}(\psi_{{\scriptscriptstyle\mathcal{T}},\alpha_t},\alpha_t)$, then, after recomputing the ground state $\psi_{{\scriptscriptstyle\mathcal{T}},\alpha_{t+1}}$, the energy cannot increase, i.e., $\widehat C_{\mathcal T}(\psi_{{\scriptscriptstyle\mathcal{T}},\alpha_{t+1}},\alpha_{t+1}) \leq \widehat C_{\mathcal T}(\psi_{{\scriptscriptstyle\mathcal{T}},\alpha_t},\alpha_t)$. Moreover, although the ground state $\psi_{{\scriptscriptstyle\mathcal{T}},\alpha}$ does depend implicitly on $\alpha$, the Hellmann-Feynman theorem implies that, at each iteration, the derivative of the ground-state energy at the current parameter value can be evaluated while keeping the corresponding ground state fixed, as in Eq.~\eqref{eq: derivative_cost}. Further details are provided in Appendix~\ref{app: gradient_cost}, where we also analyze the case of approximate ground states.

As in the inference step, the gradient can also be estimated classically once an approximate ground state has been obtained from the quantum device. In particular, a classical-shadow representation of ${\psi_{{\scriptscriptstyle\mathcal{T}},\alpha_t}}$ can be used to estimate the expectation values of the two shifted observables
\begin{equation}
O_{i,j}^{\pm}
=
U_{R_j}^\dagger
U_j^\dagger\!\left(\sum_k \alpha_{j,k} x_i^k \pm \frac{\pi}{2}\right)
U_{L_j}^\dagger
O(y_i)
U_{L_j}
U_j\!\left(\sum_k \alpha_{j,k} x_i^k \pm \frac{\pi}{2}\right)
U_{R_j}.
\end{equation}
For constant depth $U(x,\alpha)$ and local $O(y)$, even the dressed observables remain local, or at most low-weight, so their expectation values can be efficiently inferred from the local Pauli classical shadows. Classical shadows can estimate expectation values of local observables to inverse-polynomial precision with only polynomially many measurements. 

In order to train the model efficiently, the cost-function gradients must remain sufficiently large to provide a useful optimization signal, and in particular should not vanish exponentially with the system size as in barren-plateau regimes. For a general class of shallow encodings, we show in Appendix~\ref{app: barren plateaus} that, once the energy cost is strictly below $1/2$ by at least an inverse-polynomial amount, a situation that we observe in our numerical experiments and that corresponds to a non-trivial energy advantage, the gradient variance of the frozen-state cost admits an inverse-polynomial lower bound for at least an inverse-polynomial fraction of parameters. In particular, in case of shallow encodings with linearly many parameters per layer, a constant energy advantage is sufficient to have a lower bound for a constant fraction of trainable parameters. This guarantees a useful training direction whenever the assumptions of the theorem are satisfied.

\begin{theorem}[Absence of barren plateaus (informal)]
Consider a Hamiltonian embedding model defined on $n$ qubits with an observable $O(y)$ as in Eq.~\eqref{eq: observable O} with $A=n^{-1}\sum_{i=1}^{n}Z_i$, $M$ trainable parameters $\alpha\in\mathbb R^M$, and an encoding $U(x,\alpha)$ of depth $K$ whose final layer applies one independently trainable $X$- or $Y$-rotation to each qubit, with rotation angles independent of the input $x$.
At iteration $t$, keep the prepared state $|\psi_{{\scriptscriptstyle\mathcal{T}},\alpha_t}\rangle$
fixed and define
\begin{equation}
\widehat C_{\mathcal{T}}(\psi_{{\scriptscriptstyle\mathcal{T}},\alpha_t},\alpha)
=
\langle\psi_{{\scriptscriptstyle\mathcal{T}},\alpha_t}|H_{\mathcal T}(\alpha)|\psi_{{\scriptscriptstyle\mathcal{T}},\alpha_t}\rangle.
\end{equation}
If $\widehat C_{\mathcal{T}}(\psi_{{\scriptscriptstyle\mathcal{T}},\alpha_t},\alpha_t)\le \tfrac12-\gamma$, for some $\gamma>0$,
then at least $2\gamma^2n$ final-layer parameters satisfy
\begin{equation}
\operatorname{Var}_{\alpha_j}
\left[
\frac{\partial \widehat C_\mathcal{T}(\psi_{{\scriptscriptstyle\mathcal{T}},\alpha_t},\alpha)}{\partial\alpha_j}
\right]
\ge
\frac{\gamma^2}{4n^2}.
\end{equation}
Here, the variance is taken over the corresponding rotation
angle uniformly distributed in $[0,2\pi]$, while the dataset,
prepared state $\psi_{{\scriptscriptstyle\mathcal{T}},\alpha_t}$, and all other parameters remain fixed at their iteration-$t$ values.
Consequently, for encodings $U$ of $K=\Theta(1)$ layer, $M=Kn$ parameters and
$\gamma=\Theta(1)$, a constant fraction of trainable parameters has inverse-polynomial gradient variance.
\end{theorem}
The precise assumptions on the encoding $U$ and the complete proof of the theorem are given in Appendix~\ref{app: barren plateaus}.

\section{Properties of the Hamiltonian embedding model}\label{sec: properties}
The Hamiltonian embedding model exhibits several desirable properties for a quantum machine learning model. First, we show that once trained, the model generalizes well to unseen inputs. Moreover, for certain families of encoding unitaries and label observables, the model can be evaluated entirely classically once the optimal ground state has been obtained during training.
Finally, we identify a machine learning task believed to be classically hard, yet efficiently solvable on a quantum device using the Hamiltonian embedding framework.

\subsection{Generalization bounds}\label{sec: generalization bound}

Let $\ket{\psi_{{\scriptscriptstyle\mathcal{T}},\alpha^\star}}$ be the ground state found by the algorithm after the training phase. The label for a new point $x'$ is assigned by the hypothesis function $\widetilde{f}_{\scriptscriptstyle\mathcal{T},\alpha^\star}$ in \cref{eq: prediction rule}.  We define the total classification risk as the probability the model misclassifies a point $x$ sampled from the input distribution $\mathcal{D}$ over the whole input space $\Omega\subseteq\mathbb{R}^d$:
\begin{align}\label{eq: total risk}
    R(\psi,\alpha)=\Pr_{x\sim\mathcal{D}}[\widetilde{f}_{\scriptscriptstyle\mathcal{T},\alpha}(x)\neq f(x)],
\end{align}
where we recall that $f$ is the target labeling function and $\widetilde{f}_{\scriptscriptstyle\mathcal{T},\alpha}$ is the model implemented function for the Hamiltonian parameters $\alpha\in \Theta\subseteq\mathbb{R}^M$ and the corresponding reached ground state $\psi_{{\scriptscriptstyle\mathcal{T}},\alpha}.$  
We show that the total classification risk $R(\psi,\alpha)$ can be controlled through the expected energy cost over the full input distribution, computed as $C(\psi,\alpha)=\bra{\psi}H_{\mathcal{D}}(\alpha)\ket{\psi}$, where $H_{\mathcal{D}}(\alpha):=\mathbb{E}_{z\sim \mathcal{D}} [I-H_z(\alpha)]$ is the population Hamiltonian at parameter $\alpha$.  Specifically, $H_{\mathcal{D}}$ represents the average encoding Hamiltonian associated with a random labeled input drawn from the full data distribution. In order to connect the total risk to the training procedure, we first bound the difference between the total cost energy and the empirical cost in \cref{eq:cost-function}:
\begin{align}\label{eq: energy_gap}
    C(\psi_{{\scriptscriptstyle\mathcal{T}},\alpha^\star},\alpha^\star) - \widehat{C}(\psi_{{\scriptscriptstyle\mathcal{T}},\alpha^\star},\alpha^\star)=\bra{\psi_{{\scriptscriptstyle\mathcal{T}},\alpha^\star}}H_{\mathcal{D}}(\alpha^\star )-H_{\mathcal{T}}(\alpha^\star)\ket{\psi_{{\scriptscriptstyle\mathcal{T}},\alpha^\star}}.
\end{align}
 This difference is the energy generalization gap. The training algorithm minimizes the empirical energy cost $\widehat{C}$ and obtains Hamiltonian parameters $\alpha^\star$ and an associated ground state $\ket{\psi_{{\scriptscriptstyle\mathcal{T}},\alpha^\star}}$. We then show that a small empirical cost together with a small generalization gap implies a small classification error on unseen input points.
In particular, we find that the probability of observing a large deviation $\norm{H_{\mathcal{D}}(\alpha^\star)-H_{\mathcal{T}}(\alpha^\star)}_\infty$
can be bounded with a dependence that scales linearly with the number of qubits $n$ and the number of trainable parameters $M$, up to logarithmic factors. The proof relies on a matrix Hoeffding inequality, which controls the energy deviation for a fixed parameter value, together with a covering-number argument, which extends the bound uniformly over the full parameter space. For the purpose of deriving a uniform generalization bound, we restrict the parameter search to a compact domain $\Theta=[-1,1]^M$. This restriction is specific to the theoretical result below and is not required by the Hamiltonian embedding framework itself.

In particular, we obtain the following theorem. 

\begin{restatable}[Uniform energy-based generalization bound]{theorem}{genbound}\label{thm: energy gen bound}
Assume a Hamiltonian embedding model such that $0\preceq H_z(\alpha)\preceq I$ and that there exists a constant $L>0$ such that $\norm{H_z(\alpha)-H_z(\beta)}_\infty
\le
L\norm{\alpha-\beta}_\infty$
for all points $z=(x,y)\in\Omega \times \{0,1\}$ and all $\alpha,\beta\in[-1,1]^M$. Then, for every $\varepsilon\in(0,4L]$, with probability at least $1-\delta$ over the draw of $N$-size training datasets,
\begin{equation}
\sup_{\alpha\in[-1,1]^M}
\norm{H_{\mathcal{T}}(\alpha)-H_{\mathcal{D}}(\alpha)}_\infty
\leq
\sqrt{
\frac{8}{N}
\left[
\log\left(\frac{2D}{\delta}\right)
+
M\log\left(\frac{24L}{\varepsilon}\right)
\right]
}
+
\frac{\varepsilon}{4}.
\label{eq:uniform-main}
\end{equation}
where $D=2^n$ is the Hilbert space dimension.\end{restatable}

The derived energy generalization gap bound can be used to control the population risk of the trained model. In particular, suppose that the training procedure reaches an empirical energy cost $\widehat{C}(\psi_{{\scriptscriptstyle\mathcal{T}},\alpha^\star},\alpha^\star)$. Then, the number of training samples required to ensure that the population risk is upper bounded by twice the empirical energy cost, up to a constant additive error, scales linearly with the number of qubits and the number of trainable parameters, up to logarithmic factors. This is stated in the following corollary of Theorem~\ref{thm: energy gen bound}.

\begin{restatable}[Bound on the total risk]{corollary}{totrisk}\label{cor:classification-risk-bound} 
Assume the setting of Theorem~\ref{thm: energy gen bound}. Let $(\psi_{{\scriptscriptstyle\mathcal{T}},\alpha^\star},\alpha^\star)$ be the output of the training procedure. Let  $\widehat C(\psi_{{\scriptscriptstyle\mathcal{T}},\alpha^\star},\alpha^\star) = \bra{\psi_{{\scriptscriptstyle\mathcal{T}},\alpha^\star}} H_{\mathcal T}(\alpha^\star) \ket{\psi_{{\scriptscriptstyle\mathcal{T}},\alpha^\star}} $ be the empirical energy reached by the trained model, and let $ C(\psi_{{\scriptscriptstyle\mathcal{T}},\alpha^\star},\alpha^\star) = \bra{\psi_{{\scriptscriptstyle\mathcal{T}},\alpha^\star}} H_{\mathcal{D}}(\alpha^\star) \ket{\psi_{{\scriptscriptstyle\mathcal{T}},\alpha^\star}}$ be its population energy. Moreover, let $ R(\psi_{{\scriptscriptstyle\mathcal{T}},\alpha^\star},\alpha^\star) = \Pr_{x\sim \mathcal{D}} \left[ \tilde f_{\psi_{{\scriptscriptstyle\mathcal{T}},\alpha^\star},\alpha^\star}(x)\neq f(x) \right]$ be the classification risk of the trained model. Then, with probability at least $1-\delta$ over the draw of the training set $\mathcal T$, the classification risk of the trained model can be bounded by 
\begin{equation}\label{eq: risk_bound}
 R(\psi_{{\scriptscriptstyle\mathcal{T}},\alpha^\star},\alpha^\star) \le 2\widehat C(\psi_{{\scriptscriptstyle\mathcal{T}},\alpha^\star},\alpha^\star) + \varepsilon ,
 \end{equation}
provided that the number of training samples satisfies 
\begin{equation}\label{eq: sample_complexity}
 N \ge \frac{128}{\varepsilon^2} \left[ \log\left(\frac{2D}{\delta}\right) + M\log\left(\frac{24L}{\varepsilon}\right) \right]. 
\end{equation}
\end{restatable}

For simplicity, we stated the bound for variational parameters $\alpha\in[-1,1]^M$. We note that the same argument applies to any bounded parameter domain
$\Theta=[-R,R]^M$, with the corresponding modification of the covering
number.
Finally, in Appendix~\ref{app: energy gen bounds} we note that many natural encoding unitaries $U(x,\alpha)$, including the one we propose in Section~\ref{sec: construction model}, satisfy the Lipschitz assumption on the terms $H_z(\alpha)$ required by Theorem~\ref{thm: energy gen bound}. Thus, a sample complexity linear in $n$ and $M$, for fixed $\epsilon,\delta$, up to logarithmic factors is enough to guarantee a classification error less than $\varepsilon$ in the case of perfect training, i.e. when $\alpha^\star$ is such that $H_{\mathcal{T}}(\alpha^\star)$ has zero ground energy and the true ground state of $H_{\mathcal{T}}(\alpha^\star)$ is found during the training. 

We note that, if the ground states of the data Hamiltonian admit a polynomial-size circuit representation, generalization guarantees could also be obtained from existing QML bounds~\cite{caro2022generalization}, whose sample complexity is controlled by the number of parameterized gates in the circuit. However, our bound does not depend on the circuit complexity required to prepare the ground state, but only on the number $M$ of parameters entering the embedding-Hamiltonian family and on the number of qubits $n$. In particular we notice that, even when the ground state is efficiently preparable, the number of parameterized gates required by such a preparation circuit may be substantially larger than $M$.

\subsection{Classical inference and surrogate models}\label{sec: classical inference}

An interesting feature of the Hamiltonian embedding model is that, for specific choices of the encoding unitaries $U(x,\alpha)$, inference can be carried out entirely classically, thereby removing the need to access a quantum device at deployment. In these settings, the role of the quantum computer is confined to the training stage, where it is used to prepare a suitable advice state by minimizing the energy of the data-embedding Hamiltonian. 
 Quantum models that allow for classical inference after quantum training were introduced in~\cite{jerbi2024shadows} under the name \textit{shadow models}. In a shadow model, the quantum computer provides a quantum advice state obtained during training. We can then see the Hamiltonian embedding model as a special case of shadow model in which the advice state is obtained as the ground state of the data-dependent Hamiltonian. Once the trained ground state has been obtained, the structure of the encoding unitary $U(x,\alpha^\star)$ can be exploited to evaluate the prediction rule in Eq.~\eqref{eq: probability of y} without further access to the quantum device. In this section, we present different classes of encoding unitaries $U(x,\alpha)$ for which inference can be performed efficiently classically, together with the corresponding construction of classical surrogate models. Finally, in Section~\ref{sec: dequantizability surrogate} we also identify settings in which the learned ground state acts as a useful inductive bias, enabling quantum Hamiltonian embedding to solve the task with fewer training examples than a classically regularized surrogate.

\subsubsection{Fourier surrogates}
Let us consider the case of a constant depth encoding unitary $U(x,\alpha)$. A first example of a classical surrogate for evaluating the trained model is based on the Fourier decomposition of the final circuit appearing in Eq.~\eqref{eq: probability of y}. Indeed, parameterized quantum circuits are known to admit a representation as generalized trigonometric polynomials~\cite{schuld2021effectdataencodingexpressive}, whose set of frequencies is determined solely by the data-encoding gates. Specifically, let $x\in\mathbb{R}^d$ be the input points, then any parameterized quantum circuit of the kind 
\begin{equation}
    f_\alpha(x)=\bra{\psi}U^\dagger(x,\alpha)MU(x,\alpha)\ket{\psi}
\end{equation}
can be equivalently expressed in the form 
\begin{equation}\label{eq: fourier_coeff}
    f_\alpha(x)
    =
    \sum_{\boldsymbol{\omega}\in\Omega_\alpha}
    c_{\boldsymbol{\omega}}(\alpha)
    e^{i\boldsymbol{\omega}\cdot x},
\end{equation}
where $\Omega_\alpha$ denotes the set of accessible frequencies and
$c_{\boldsymbol{\omega}}(\alpha)\in\mathbb{C}$ are the corresponding
Fourier coefficients. The set of accessible frequencies $\Omega_\alpha$ only depends on the gates $U(x,\alpha)$ used to encode the inputs $x$. In our model, where $U(x,\alpha)$ is shallow and the measured observables are local, the number of distinct frequencies in $\Omega_\alpha$ can be shown to scale polynomially with the input size.

In the case where the frequency spectrum $\Omega_\alpha$ is polynomially large, recent results have shown that a classical representation of the circuit can be learned efficiently. In particular, we can use the results in~\cite{schreiber2023classical} to learn an approximation for the coefficients $c_\omega(\alpha)$ in Eq.~\eqref{eq: fourier_coeff} and construct a classical surrogate of the form: 
\begin{equation}\label{eq: surr_fourier_coeff}
    g_\alpha(x)
    =
    \sum_{\boldsymbol{\omega}\in\Omega_\alpha}
    \tilde c_{\boldsymbol{\omega}}(\alpha)
    e^{i\boldsymbol{\omega}\cdot x},
\end{equation}
with the guarantee that 
\begin{equation}
    \Pr\left[
    \sup_x
    \left|
    f_\alpha(x)-g_\alpha(x)
    \right|
    \leq \epsilon
    \right]
    \geq 1-\delta,
\end{equation}
using polynomially many queries in $1/\delta$ and $1/\epsilon$ of the original quantum circuit.

\subsubsection{Pauli surrogates}\label{sec: Pauli surrogates}
For constant depth encoding $U(x,\alpha)$, we can define a different classical surrogate model by exploiting the locality of the observables considered at inference time. In particular, rather than learning the Fourier coefficients of the final circuit, we construct a classical representation of the trained model by expanding the dressed observable in a basis of local Pauli operators and estimating their expectation values on the learned ground state using classical shadows.
When $A$ in Eq.~\eqref{eq: observable O} is a sum of local Pauli operators, these estimates provide a classical description sufficient to evaluate the inference rule entirely through classical-shadow post-processing~\cite{huang2020predicting}. In particular, once the optimized state $\ket{\psi_{{\scriptscriptstyle\mathcal{T}},\alpha^\star}}$ is obtained during training, the probability $P_{\psi_{{\scriptscriptstyle\mathcal{T}},\alpha^\star},\alpha^\star}(x',y)$ in Eq.~\eqref{eq: probability of y} for a new input $x'$ can be estimated from a classical-shadow representation of the learned ground state. Indeed, let us consider the dressed observable constructed by conjugating the original observable in Eq.~\eqref{eq: observable O} with the trained encoding unitary for a new input point $x'$:
\begin{equation}\label{eq: dressed O(y)}
    \tilde O(x',y) =
\frac{I+(-1)^y\tilde A}{2},
\qquad
\tilde A(x)= U^\dagger (x',\alpha^\star )AU(x',\alpha ^\star). 
\end{equation}
 Since we choose $A$, and therefore $O(y)$, to be local, conjugating it by a constant-depth local encoding unitary $U(x,\alpha)$ only spreads its support within a constant-size light cone. Therefore, the dressed observable $\widetilde O(x',y)=U^\dagger(x',\alpha^\star)O(y)U(x',\alpha^\star)$ can then be efficiently decomposed into local Pauli operators. 
Indeed, if $U$ is a constant-depth circuit composed of local gates, 
the backward light cone of each measured term contains at most 
$k=O(1)$ qubits. Each such term can therefore be expanded into at most 
$4^k$ Pauli operators. Since the observable contains $O(n)$ local terms, 
we can write
\begin{equation}
    \label{eq: decomposition Otilde}
    \widetilde A(x')
    =
    \sum_{i=1}^{K}
    c_i(x',\alpha^\star) P_i,
\end{equation}
with $K=O(n4^k)$, where the $P_i$ are Pauli observables of weight at most 
$k$. Moreover, the coefficients $c_i(x',\alpha^\star)$ can be computed
efficiently classically by simulating the constant-size backward light
cones. Once the approximate ground state $\ket{\psi_{{\scriptscriptstyle\mathcal{T}},\alpha^\star}}$ has been prepared quantumly, a classical shadow built from 
\begin{equation}
    S
    =
    O\!\left(
    \frac{
        3^k \max_x \|c(x,\alpha^\star)\|_1^2
    }{\epsilon^2}
    \log\frac{K}{\delta}
    \right)
\end{equation} 
measurements suffices to estimate $P_{\psi_{{\scriptscriptstyle\mathcal{T}},\alpha^\star},\alpha^\star}(x',y)$ to additive error $\epsilon$, uniformly over all $x'\in\Omega$, with probability at least $1-\delta$~\cite{huang2020predicting}.

\subsubsection{Deep Clifford encodings}

A constant depth $U(x,\alpha)$ is not the only case in which the Hamiltonian embedding model allows for an entirely classical inference after training. As an example, let us consider an encoding unitary of the form $U(x,\alpha)=U_1(x,\alpha)VU_2(\alpha)$, where $U_2(\alpha)$ is independent of the input $x$, $V$ is a deep Clifford circuit, and $U_1(x,\alpha)$ is a shallow-depth unitary. Similarly to the Pauli-surrogate case, we can define the dressed observable obtained by conjugating $O(y)$ with the shallow part of the encoding,
\begin{equation}
\widetilde O_1(x,y)
=
U_1^\dagger(x,\alpha^\star)O(y)U_1(x,\alpha^\star)
=
\sum_{i=1}^{K} c_i(x,\alpha^\star)P_i,
\end{equation}
where $K=\operatorname{poly}(n)$ and the $P_i$ are local Pauli observables. Since $V$ is a Clifford circuit, conjugation by $V$ maps every Pauli observable to another Pauli observable,
\begin{equation}
V^\dagger P_i V = Q_i.
\end{equation}
Although the resulting operators $Q_i$ may have support on a number of qubits growing with the system size, the number of terms in the decomposition remains polynomial. The prediction probability rule of \cref{eq: probability of y} can therefore be written as
\begin{equation}
P_{\psi_{{\scriptscriptstyle\mathcal{T}},\alpha^\star},\alpha^\star}(x,y)
=
\sum_{i=1}^{K}
c_i(x,\alpha^\star)
m_i,
\qquad
m_i
=
\bra{\psi_{{\scriptscriptstyle\mathcal{T}},\alpha^\star}}
U_2^\dagger(\alpha^\star)Q_iU_2(\alpha^\star)
\ket{\psi_{{\scriptscriptstyle\mathcal{T}},\alpha^\star}}.
\end{equation}
Crucially, since $U_2$ is independent of $x$, the coefficients $m_i$ need to be estimated only once after training. Since the operators $Q_i$ may have large Pauli weight, product-Pauli classical shadows may not be efficient. However, as only polynomially many such observables appear, their expectation values can be measured directly on the quantum device and stored classically.
Predictions for new inputs can then be obtained entirely classically by computing the coefficients $c_i(x,\alpha^\star)$ and evaluating the above sum.

\subsubsection{Dequantizability of the model}\label{sec: dequantizability surrogate}
The existence of efficient classical surrogates for the trained quantum model raises the natural question of whether the full learning procedure can itself be reproduced classically, potentially eliminating the need for quantum computation altogether. For example, let us consider the Pauli surrogate. Suppose that the dressed observable $\tilde O$ admits a decomposition into polynomially many Pauli operators, as in Eq.~\eqref{eq: decomposition Otilde}. The prediction function of the trained quantum model can then be written as
\begin{equation}
    \tilde f_{\alpha^\star}(x)
    =
    \sum_{i}^{K} c_i(x,\alpha^\star)m_i(\alpha^\star),
    \label{eq:trained_quantum_model}
\end{equation}
where $m_i(\alpha^\star)=\bra{\psi_{{\scriptscriptstyle\mathcal{T}},\alpha^\star}}P_i\ket{\psi_{{\scriptscriptstyle\mathcal{T}},\alpha^\star}}$ are the expectation values of the Pauli operators $P_i$ on the learned ground state, while $c_i(x,\alpha^\star)$ are the coefficients appearing in the Pauli decomposition of the dressed observable as in Eq~\eqref{eq: decomposition Otilde}. 

We can then define a fully classical surrogate of the form
\begin{equation}\label{eq: pauli surrogate base}
    g_{\alpha}(x)
    =
    \sum_i^K c_i(x,\alpha)\tilde m_i(\alpha),
\end{equation}
where the coefficients $\tilde{m}_i(\alpha)\in \mathbb{R}$ are treated as trainable classical parameters and are not required to correspond to expectation values of Pauli operators on any quantum state.

In the quantum model of \cref{eq:trained_quantum_model}, the moments $m_i(\alpha^\star)$ are not independent parameters, but must be jointly compatible with a single physical quantum state. These compatibility constraints can in principle act as a useful inductive bias, although they do not by themselves imply either easier optimization or improved generalization relative to classical alternatives. To investigate this question, in Appendix~\ref{app: provable sep surrogate} we construct a learning problem for which the benefit of these physical compatibility constraints can be established analytically, yielding a provable separation in the number of training examples required by Hamiltonian embedding and a corresponding classically regularized surrogate. We complement this result in Appendix~\ref{app: counter surrogate} with a finite-size numerical example exhibiting the same qualitative behavior. Importantly, both tasks are classically tractable, and the purpose of these examples is to isolate the generalization advantage of the physical compatibility constraints imposed by the learned ground state rather than to demonstrate a computational quantum advantage.

A more constrained surrogate can be obtained when each dressed readout term $\tilde A_j=U(x,\alpha)^\dagger P_jU(x,\alpha)$ in Eq.~\eqref{eq: decomposition Otilde} is supported on a fixed region $L_j$ of the $n$ qubits, independent of $x$ and $\alpha$, with $|L_j|=O(\log n)$. Writing the observable $A$ in Eq.~\eqref{eq: observable O} as $A=\sum_j^K a_jP_j$, we define a classical surrogate as
\begin{equation}\label{eq: density surrogate}
g_{\sigma,\alpha }(x)
=\sum_j^K a_j\operatorname{Tr}\left[
\sigma_j \tilde A_j
\right],
\end{equation}
where each $\sigma_j$ is a density matrix on the fixed subsystem $L_j$, with no global compatibility constraints, and $\tilde A_j(x,\alpha)$ has support on $L_j$. At fixed $\alpha$, from the energy-loss definition in Eq.~\eqref{eq:cost-function}, the empirical energy becomes $\widehat C_{\mathcal T}(\sigma,\alpha)
=\frac12-\frac12\sum_{j=1}^{K}a_j \operatorname{Tr}[\sigma_jB_j(\alpha)]$, where $B_j(\alpha)=N^{-1}\sum_{\ell=1}^{N}s_\ell\tilde A_j(x_\ell,\alpha)$ and $s_\ell=(-1)^{y_\ell}$. Since the matrices $\sigma_j$ can be optimized independently, energy minimization reduces to $K$ eigenproblems by choosing each $\sigma_j$ as a projector onto a top eigenvector of $a_jB_j(\alpha)$. These matrices have dimension at most $2^{|L_j|}=\operatorname{poly}(n)$, so the inner optimization runs in polynomial-time when the dressed terms can be computed efficiently and achieves an energy lower or equal to the ground-state energy of the original embedding Hamiltonian $H_{\mathcal T}(\alpha)$. Under $\sum_{j=1}^{K}|a_j|\le1$ and uniform operator-norm Lipschitz bounds on the individual dressed terms $\tilde A(x,\alpha )$, the same proof strategy of Theorem~\ref{thm: energy gen bound} extends to this relaxation guaranteeing also the same generalization bounds. However, we note that this surrogate construction may not remain efficient when the encoding unitary $U(x,\alpha)$ is such that the supports of each dressed term $\tilde A_j(x,\alpha)$ depend on $x$. In this case, the support of their union over all training points may involve all $n$ qubits, making the dimension $2^{|L_j|}$ of $\sigma_j$ scale exponentially. For instance, the construction in Appendix~\ref{app: provable sep surrogate} also exhibits such input-dependent supports, providing an example in which the efficiency guarantee for this surrogate no longer applies.  Finally, we note that the polynomial-time training statements of the surrogates do not cover the outer optimization over $\alpha$, which is not convex and could benefit from the inductive bias induced by the ground state of the embedding Hamiltonian.

In particular, a separate question is whether the full Hamiltonian embedding learning procedure can itself be efficiently reproduced classically. The existence of efficient classical surrogates for inference does not, in general, imply that the corresponding model can also be learned efficiently classically. In particular, one natural route beyond the shallow-surrogate regime is to use encodings whose dressed observables have growing light cones and do not admit a polynomial-sized, efficiently evaluable feature representation. Increasing the depth of $U$ can therefore invalidate the particular Pauli- and Fourier-surrogate arguments used above. Depth alone, however, is not sufficient, since structured deep circuits may remain classically tractable and low-bond-dimension learned states may still permit efficient tensor-network evaluation.
The relevant question is whether the task-dependent moments defining the classifier can be learned or computed classically, rather than circuit depth or entanglement alone. The construction in the following section provides an explicit setting in which the trained hypothesis is efficiently
evaluable but learning it classically is nevertheless conjectured to be hard under cryptographic assumptions.

\subsection{Encoding a classically hard problem in the Hamiltonian embedding model}\label{sec: hard problem}

When introducing a new quantum machine learning method, it is natural to ask whether it can be applied to problems that lie beyond the reach of classical algorithms. In the Hamiltonian embedding framework, the training procedure requires finding the ground state of an encoding Hamiltonian. Since ground-state preparation is, in general, hard even for quantum computers, it is not immediately clear whether this approach can provide a rigorous quantum speedup. In this section we present a learning task that is believed to be classically intractable, yet can be expressed within our Hamiltonian-embedding framework and solved using a quantum computer. The learning problem is the binary classification task described in~\cite{jerbi2024shadows}, where input $x\in\mathbb{Z}_N=\{0,1,...,N-1\}$ are labeled in two classes $y=\{0,1\}$ by the following labeling function:
\begin{equation}\label{eq: dcr_labeling_function}
y(x)=f_{N,s}(x)=
\begin{cases}
1, & \text{if } \sqrt[3]{x}\mod N \in \left[s,\, s+\dfrac{N-1}{2}\right],\\[6pt]
0, & \text{otherwise}.
\end{cases}
\end{equation}
where $N=pq$ is a large $n$-bit integer product of the two primes $p$ and $q$ and the interval in \cref{eq: dcr_labeling_function} is interpreted cyclically in $\mathbb{Z}_N$.
In particular, the labeling function partitions inputs into two classes according to whether their discrete cube root (DCR) modulo $N$ lies in a specified half-interval, determined by $s$. Under the DCR assumption, computing discrete cube roots modulo large semiprimes $N$ is believed to be classically intractable~\cite{kearns1994introduction}, and the associated learning problem is also believed to be classically hard~\cite{jerbi2024shadows}. Notably, the DCR function is a trapdoor function: there exists a key $d\in\mathbb{Z}_N$ that allows one to evaluate the function efficiently. More precisely, for the DCR function, there exists a $d$ such that $\sqrt[3]{x}\bmod N = x^d \bmod N$.
Since modular exponentiation can be computed efficiently classically, evaluating the DCR function on different inputs $x$ becomes easy once $d$ is known. However, determining the key $d$ requires factoring the integer $N$, which is believed to be computationally hard for classical algorithms for large integers $N$.

The Hamiltonian embedding framework can reproduce the DCR labeling function in \cref{eq: dcr_labeling_function} as its hypothesis function. More precisely, one can construct an encoding Hamiltonian whose ground state induces, through the prediction rule in \cref{eq: prediction}, a hypothesis $\widetilde{f}$ that well approximates the DCR labeling function. In particular, we consider the following encoding unitary $U(x)$ defined over the Hilbert space $\mathcal{H}=\mathrm{span}\{|d,s\rangle : d,s\in \mathbb Z_N\}\otimes \mathbb C^2$ and the observable $O(y)$:
\begin{equation}\label{eq: U(x) DCR}
   U(x)\,|d,s\rangle|b\rangle
=
|d,s\rangle\,|b\oplus g_{d,s}(x)\rangle \quad \&\quad O(y)=I\otimes |y\rangle\langle y|,
\qquad x,d,s\in\mathbb{Z}_N,y\in\{0,1\}.
\end{equation}
where $g_{d,s}$ is defined as:
\begin{equation}\label{eq: dcr_key_function}
g_{d,s}(x)=
\begin{cases}
1, & \text{if } x^d\mod N \in \left[s,\, s+\dfrac{N-1}{2}\right],\\[6pt]
0, & \text{otherwise}
\end{cases}
\end{equation}
and it is classically computable as modular exponentiation can be computed efficiently. We then construct the embedding Hamiltonian $H_{\text{DCR}}$ as   
\begin{equation}\label{eq: dcr_hamiltonian}
    H_{\text{DCR}}=\sum_\ell H(x_\ell,y_\ell).
\end{equation}
We note that the encoding Hamiltonian $H_{\mathrm{DCR}}$ admits an efficient classical description, since the encoding unitary $U$ in \cref{eq: U(x) DCR} can be implemented by a polynomial-size classical reversible circuit, so that each Hamiltonian term associated with a training point $(x_\ell,y_\ell)\in\mathcal{T}$ can be efficiently specified as:
\begin{equation}\label{eq: Hz DCR}
H(x_\ell,y_\ell)
=
I-U(x_\ell)^\dagger O(y_\ell)U(x_\ell)+I\otimes |1\rangle\langle 1|.
\end{equation}
The ground space of $H_{\mathrm{DCR}}$ contains states $\ket{d,s,0}$ associated with hypotheses $g_{d,s}$ that are consistent with the training data. As shown in Appendix~\ref{app: DCR}, a quantum algorithm can efficiently find a ground state $\ket{d^\star,s',0}$, whose corresponding hypothesis $g_{d^\star,s'}$ has small population error on the DCR learning task. In particular, following the construction in Section~\ref{sec: ham embedding model}, for a new point $x$ the predicted label is the value of $y$ that maximizes $\operatorname{Tr}\!\left[
O(y)\,
U(x)\,
|d^\star,s',0\rangle\langle d^\star,s',0|
\,U(x)^\dagger
\right]$.  Therefore, we can prove the following theorem:

\begin{theorem}[Encoding of a hard problem (informal)]\label{thm: dcr advantage}
There exists a classification problem for which the Hamiltonian embedding framework yields a classically efficiently describable data-embedding Hamiltonian whose ground states induce, through the prediction rule in Eq.~\eqref{eq: prediction rule}, hypotheses with arbitrarily small classification error. Moreover, a quantum computer can efficiently find a corresponding ground state and construct the associated hypothesis. In contrast, assuming the classical hardness of computing discrete cube roots, no efficient classical learner can achieve the same guarantee with non-negligible probability.
\end{theorem}

The full specification of the result, together with the formal proof, can be found in Appendix~\ref{app: DCR}.
Finally, we note a distinction between Hamiltonian embedding and some earlier examples of quantum models that efficiently solve learning tasks believed to be classically intractable~\cite{liu2021rigorous, anschuetz2026arbitrary, molteni2026exponential, zhao2026exponential, gil-fuster2025on}, concerning where the hard computation takes place. In these examples, the quantum models are explicitly constructed to incorporate the correct hard computation required to solve the specific task into their circuit architectures. In the discrete-logarithm-based constructions of~\cite{liu2021rigorous,gil-fuster2025on}, for example, the encoding circuits incorporate discrete-logarithm computations, while the remaining training optimization then becomes easy. In our construction, the hard computation instead takes place during data-dependent quantum optimization through ground-state preparation, while the encoding circuits themselves perform efficiently computable classical operations. Here ground-state preparation is treated as a subroutine.
The encoding circuits therefore need not be designed a priori to perform a specific hard computation and, in fact, may also be classically tractable, as is the case in the classically hard learning task presented here.

\section{Experiments}\label{sec: numerics}

We next test whether the optimization in Algorithm~\ref{alg:ham-embed} remains feasible at nontrivial system sizes and whether the resulting low-energy states achieve low classification error on a standard benchmark dataset.
In particular, we evaluate the Hamiltonian embedding algorithm on the binary classification task of Fashion-MNIST dress and shirt images, with inputs reduced by PCA to $n$ features for an $n$-qubit model. We first train the model using a classical iterative solver as a proxy~\cite{kirby2026observationimprovedaccuracyclassical} to approximate the ground state at each iteration. This benchmark uses $100$ random seeds for each of the $25$ even system sizes $n=2,4,\ldots,50$, with each seed determining both the data split and the initial model parameters. We compare test accuracy against classical machine-learning baselines and the fidelity and projected quantum kernels in~\cite{huang2021power}.

We then implement the full hybrid training procedure, using sample-based Krylov quantum diagonalization on IBM Heron processors for ground-state approximation. The selection criteria for the hardware runs are given below and in Appendix~\ref{app: qpu}.

\subsection{The model instance used in the experiments}\label{sec: numerics setup}

All Hamiltonian embedding experiments use the same model architecture and training objective introduced in Section~\ref{sec: ham embedding model}, which we summarize here and discuss in detail in Appendix~\ref{app: numerics setup}. The encoding unitary is the constant-depth circuit in Eq.~\eqref{eq: encoding U} with $K=3$ layers. We choose $K=3$ to explore nontrivial Hamiltonians while keeping the number of Pauli terms computationally tractable. Each layer applies three single-qubit rotations, generated by $Z$, $Y$, and $Z$, to every qubit, followed by a ring of nearest-neighbor two-qubit $ZZ$ rotations. The circuit therefore contains $12n$ parameterized gates, whose trainable coefficients  are $\alpha\in\mathbb{R}^{12n\times(d+1)}$. For an input $x\in\mathbb{R}^d$, the angle of gate $g$ is $\theta_g(x,\alpha)=\sum_{r=1}^{d}\alpha_{g,r}x_r+\alpha_{g,0}$. This is equivalent to the encoding in Eq.~\eqref{eq: encoding U} with a constant feature appended to the input. We use the one-local readout
$A=\frac{1}{n}\sum_{i=1}^n Z_i$ from Section~\ref{sec: construction model}. The outer optimization minimizes the energy loss defined in Eq.~\eqref{eq: energy_loss}, which is sufficient to achieve good classification performance in our experiments. 

Two features make the experiments feasible at these system sizes. First, the light-cone argument in Section~\ref{sec: construction model} shows that each data point contributes a linear combination of the same $\mathcal{O}(n)$ bounded-weight Pauli observables to the embedding Hamiltonian $H_{\T}(\alpha)$. We obtain these strings by propagating the readout observable through the circuit in the Heisenberg picture and combining identical terms.  Second, as discussed in Section~\ref{sec: barren plateaus and gradient}, the state is held fixed during each parameter update, so the gradient is given by the expectation value of $\partial_\alpha H_{\T}(\alpha)$ in the current state.

We optimize the parameters using Adam~\cite{kingma2015adam} together with an $\ell_2$ penalty. At each training epoch, we evaluate the model on the validation set. After training, we keep the variational parameters and associated learned quantum state from the epoch with the highest validation accuracy, and evaluate the resulting model once on the held-out test~set.

We first study Hamiltonian embedding framework using a classical method for approximating ground states that we call ``classical iterative solver''~\cite{kirby2026observationimprovedaccuracyclassical}.
In these simulations, we approximate the ground state at each training iteration using an iterative solver based on repeated applications of $H_{\T}(\alpha)$, deterministically retaining the bitstrings with the highest probabilities, akin to a power iteration with truncation. We reserve the term ``HE'' for the hardware runs using the real-time SKQD protocol.

For the hardware runs, we approximate the ground state at each training iteration using SKQD~\cite{yu2025quantum}. Starting from a reference state, we generate a sequence of Krylov states by Trotterized real-time evolution under the current training Hamiltonian and measure them in the computational basis. The sampled bitstrings span a reduced subspace onto which we project $H_{\T}(\alpha)$ and diagonalize it classically.

The classical iterative solver and SKQD may select different basis states and therefore produce different reduced subspaces. Both methods nevertheless project the full embedding Hamiltonian onto the selected subspace, so the lowest projected energy remains a variational upper bound on the true ground-state energy. The classical solver provides a useful reference for the optimization and classification performance of Hamiltonian embedding, although its performance need not predict that of SKQD.

\subsection{Results}\label{sec: numerics fashion}

\begin{figure}[t]
  \centering
  \includegraphics[width=\textwidth]{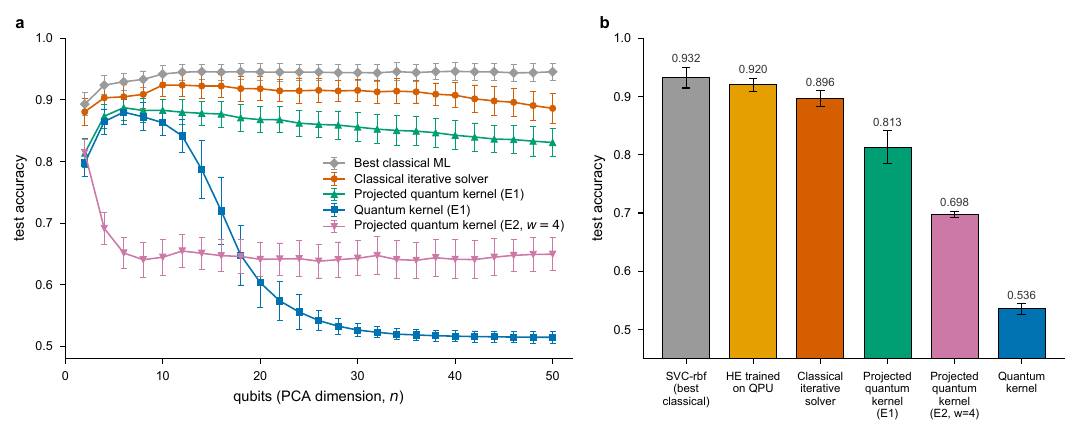}
  \caption{
  a) Test accuracy versus system size for various methods on a binary classification problem from the Fashion-MNIST dataset.
 The methods are: best-performing classical baseline (gray), classical iterative solver (orange) used in HE framework, quantum kernel based on a product-state embedding (blue), $1$-reduced density matrix ($1$-RDM) projected quantum kernel with the same embedding \cite{huang2021power} (green), and the $1$-RDM projected quantum kernel based on the E2 embedding of Ref.~\cite{huang2021power} (pink).
  b) Held-out test accuracy of the five runs with the highest validation accuracy at $n=50$. In this plot, we additionally show the accuracy obtained by the HE model trained on hardware (yellow).
  }
  \label{fig: qpu train overlay}
\end{figure}

We compare Hamiltonian embedding with the fidelity and projected quantum kernels studied in Ref.~\cite{huang2021power}, as well as with classical machine-learning methods. We use the same image-classification benchmark studied in Ref.~\cite{huang2021power}, which is a binary classification of Fashion-MNIST images~\cite{xiao2017fashion} into the classes \emph{dress} and \emph{shirt}. We first use principal component analysis (PCA) to reduce the input dimension to the number of qubits $n$. For each run, a pool of $800$ samples is divided into $448$ training, $112$ validation, and $240$ held-out test examples. We first split the data $70/30$ into training and test sets, and then split the training set $80/20$ into the final training and validation sets. At each $n$, all methods use the same random seeds and therefore the same data splits. 

Figure~\ref{fig: qpu train overlay}(a) reports the mean test accuracy across the $100$ runs at each system size, with error bars showing one standard deviation. We note that the best classical ML baseline, which corresponds to the support vector classifier with a radial basis function kernel, performs the best, reaching an accuracy of approximately $0.946$.
The classical iterative solver, which may be viewed as a proxy for the Hamiltonian embedding framework, achieves a higher mean test accuracy than the best reported quantum kernel baseline, with margins ranging from $0.018$ to $0.067$. The mean accuracy increases from  $0.880$ at $n=2$ to approximately $0.924$ at intermediate system sizes before gradually decreasing to $0.886$ at $n=50$.
The fidelity kernel based on the product-state embedding E1 from~\cite{huang2021power} decreases from $0.798$ to $0.515$, whereas the projected $1$-RDM kernel based on the same embedding decreases more gradually from a maximum of $0.887$ to $0.831$.

We also study a projected quantum kernel constructed from the $1$-RDMs of feature states generated using a simplified version of the E2 embedding from Refs.~\cite{havlicek2019supervised,huang2021power}. These feature states are prepared using circuits inspired by the instantaneous quantum polynomial-time (IQP) model, which we call IQP-like circuits. To make their classical simulation tractable, we restrict the interactions to bands of neighboring qubits of width $w=4$, allowing the circuits and resulting kernels to be evaluated using matrix product states. The test accuracy of this kernel decreases rapidly with $n$ and levels off near $0.65$, above that of the fidelity kernel at large $n$.

We next turn to the hardware results. In these runs, we prepare the SKQD Krylov states described above on IBM Heron processors and measure them in the computational basis. We then project $H_{\T}(\alpha_t)$ in the sampled bitstring subspace and diagonalize it classically. We evaluate the gradients and perform the Adam updates classically, following Algorithm~\ref{alg:ham-embed}. 

We select five hardware runs based on the results obtained using the classical iterative solver to select data splits and initial model parameters $\alpha$ that lead to high validation accuracy.
As shown in \cref{fig: qpu train overlay}(b), Hamiltonian embedding is competitive with the best classical baseline and outperforms the other quantum methods considered. All $n=50$ runs were performed on \texttt{ibm\_pittsburgh}, a $156$-qubit IBM Heron~r2 processor. The $50$ data qubits were mapped to a linear chain chosen greedily to minimize two-qubit gate~errors. Moreover, we apply dynamical decoupling and measurement twirling but no additional error mitigation. 

 These experiments demonstrate that Hamiltonian embedding models can be studied at scale on current quantum processors.

\section{Discussion}

In this work, we introduced Hamiltonian embedding, a QML framework for classification in which the entire training set is encoded in a single parameterized Hamiltonian. Training minimizes the ground-state energy or a more general loss function over the encoding parameters, with the loss function estimated by a quantum algorithm at each iteration. The ground state of the optimized Hamiltonian then defines a classifier for unseen inputs. For shallow encoding circuits, measurements of the prepared states yield classical representations that allow the loss and its gradients to be evaluated efficiently. Moreover, the representation of the final state supports inference without further quantum computation. Therefore, the role of quantum processors is restricted to approximating ground states during training. We further showed, analytically and numerically, that the physical compatibility constraints imposed by the learned state can provide a useful inductive bias for certain tasks. Finally, under standard cryptographic assumptions, Hamiltonian embedding can encode certain learning tasks that are intractable for efficient classical algorithms.

For generic learning tasks, however, it remains unclear when the inductive bias of Hamiltonian embedding offers an advantage over alternative classical or quantum methods. This question is particularly difficult for classical datasets, where symmetry-based arguments may not apply directly \cite{glick2024covariant,ragone2022representation}, and will likely require heuristic model design together with systematic, task-specific benchmarking. Practical use also requires the model to remain trainable and to generalize as the system size grows. For fixed generalization error, our analytical bounds show that the required number of training samples scales linearly with the number of qubits and trainable parameters, up to logarithmic factors. Our numerical results provide complementary evidence for trainability. On binary Fashion-MNIST, Hamiltonian embedding can be trained at scale, including implementations on IBM quantum processors, while maintaining competitive test performance as the system size increases. In the same benchmarks, the generalization performance of the quantum-kernel method of Ref.~\cite{huang2021power} deteriorates with system size, whereas Hamiltonian embedding also achieves higher accuracy than the projected quantum kernels designed to mitigate this behavior.

Hamiltonian embedding therefore offers a promising alternative to existing QML models. Several aspects of the framework remain unexplored. In particular, the freedom to choose both the encoding and the readout observable may enable constructions with improved trainability or expressivity. A central open question is which learning problems are naturally suited to Hamiltonian embedding and whether it can offer advantages over classical methods beyond cryptographically motivated constructions.

\section*{Acknowledgments}
We acknowledge discussions with Aram Harrow, Will Kirby, Sarah Sheldon, Kevin Smith, and Minh Tran. R.M. and V.D. acknowledge the support of the Dutch Research Council (NWO) through the NWO/NWA project Divide and Quantum (Project No. 1389.20.241). V.D. also acknowledges the support of the European Union’s Horizon Europe program through the ERC CoG BeMAIQuantum (Grant No. 101124342). 
This work was also supported by the Dutch National Growth Fund (NGF) as part of the QDNL program.

We used large language models to assist with numerical simulations and parts of the analytical work, including selected steps in the proofs. The authors critically reviewed and verified all incorporated suggestions and take full responsibility for the accuracy and integrity of this work.

\bibliographystyle{alphaurl}
\bibliography{Ref}

\newcommand{\etalchar}[1]{$^{#1}$}
\begin{thebibliography}{CMM{\etalchar{+}}26}

\bibitem[ACGS88]{alexi1988rsa}
Werner Alexi, Benny Chor, Oded Goldreich, and Claus~P. Schnorr.
\newblock {{RSA}} and {{Rabin Functions}}: {{Certain Parts}} are as {{Hard}} as the {{Whole}}.
\newblock {\em SIAM Journal on Computing}, 17(2):194--209, 1988.
\newblock \href {https://doi.org/10.1137/0217013} {\path{doi:10.1137/0217013}}.

\bibitem[AG26]{anschuetz2026arbitrary}
Eric~R. Anschuetz and Xun Gao.
\newblock Arbitrary {{Polynomial Separations}} in {{Trainable Quantum Machine Learning}}.
\newblock {\em Quantum}, 10:1976, 2026.
\newblock \href {https://doi.org/10.22331/q-2026-01-20-1976} {\path{doi:10.22331/q-2026-01-20-1976}}.

\bibitem[AL99]{abrams1999quantumalgorithmprovidingexponential}
Daniel~S. Abrams and Seth Lloyd.
\newblock Quantum {{Algorithm Providing Exponential Speed Increase}} for {{Finding Eigenvalues}} and {{Eigenvectors}}.
\newblock {\em Physical Review Letters}, 83(24):5162--5165, 1999.
\newblock \href {https://doi.org/10.1103/PhysRevLett.83.5162} {\path{doi:10.1103/PhysRevLett.83.5162}}.

\bibitem[BLSF19]{benedetti2019parameterized}
Marcello Benedetti, Erika Lloyd, Stefan Sack, and Mattia Fiorentini.
\newblock Parameterized quantum circuits as machine learning models.
\newblock {\em Quantum Science and Technology}, 4(4):043001, 2019.
\newblock \href {https://doi.org/10.1088/2058-9565/ab4eb5} {\path{doi:10.1088/2058-9565/ab4eb5}}.

\bibitem[BWP{\etalchar{+}}17]{biamonte2017quantum}
Jacob Biamonte, Peter Wittek, Nicola Pancotti, Patrick Rebentrost, Nathan Wiebe, and Seth Lloyd.
\newblock Quantum machine learning.
\newblock {\em Nature}, 549(7671):195--202, 2017.
\newblock \href {https://doi.org/10.1038/nature23474} {\path{doi:10.1038/nature23474}}.

\bibitem[CAB{\etalchar{+}}21]{cerezo2021variational}
M.~Cerezo, Andrew Arrasmith, Ryan Babbush, Simon~C. Benjamin, Suguru Endo, Keisuke Fujii, Jarrod~R. McClean, Kosuke Mitarai, Xiao Yuan, Lukasz Cincio, and Patrick~J. Coles.
\newblock Variational quantum algorithms.
\newblock {\em Nature Reviews Physics}, 3(9):625--644, 2021.
\newblock \href {https://doi.org/10.1038/s42254-021-00348-9} {\path{doi:10.1038/s42254-021-00348-9}}.

\bibitem[CC25]{chang2025primer}
Su~Yeon Chang and M.~Cerezo.
\newblock A {{Primer}} on {{Quantum Machine Learning}}.
\newblock (arXiv:2511.15969), 2025.
\newblock \href {https://arxiv.org/abs/2511.15969} {\path{arXiv:2511.15969}}.

\bibitem[CEMM98]{cleve1998quantum}
R.~Cleve, A.~Ekert, C.~Macchiavello, and M.~Mosca.
\newblock Quantum algorithms revisited.
\newblock {\em Proceedings of the Royal Society of London. Series A: Mathematical, Physical and Engineering Sciences}, 454(1969):339--354, 1998.
\newblock \href {https://doi.org/10.1098/rspa.1998.0164} {\path{doi:10.1098/rspa.1998.0164}}.

\bibitem[CHC{\etalchar{+}}22]{caro2022generalization}
Matthias~C. Caro, Hsin-Yuan Huang, M.~Cerezo, Kunal Sharma, Andrew Sornborger, Lukasz Cincio, and Patrick~J. Coles.
\newblock Generalization in quantum machine learning from few training data.
\newblock {\em Nature Communications}, 13(1):4919, 2022.
\newblock \href {https://doi.org/10.1038/s41467-022-32550-3} {\path{doi:10.1038/s41467-022-32550-3}}.

\bibitem[CLG{\etalchar{+}}25]{cerezo2025does}
M.~Cerezo, Martin Larocca, Diego {Garc{\'i}a-Mart{\'i}n}, N.~L. Diaz, Paolo Braccia, Enrico Fontana, Manuel~S. Rudolph, Pablo Bermejo, Aroosa Ijaz, Supanut Thanasilp, Eric~R. Anschuetz, and Zo{\"e} Holmes.
\newblock Does provable absence of barren plateaus imply classical simulability?
\newblock {\em Nature Communications}, 16(1):7907, 2025.
\newblock \href {https://doi.org/10.1038/s41467-025-63099-6} {\path{doi:10.1038/s41467-025-63099-6}}.

\bibitem[CMM{\etalchar{+}}26]{candelori2026shallow}
Luca Candelori, Swarnadeep Majumder, Antonio Mezzacapo, Javier~Robledo Moreno, Kharen Musaelian, Santhanam Nagarajan, Sunil Pinnamaneni, Kunal Sharma, and Dario Villani.
\newblock Shallow-circuit supervised learning on a quantum processor.
\newblock {\em arXiv preprint arXiv:2601.03235}, 2026.

\bibitem[CSV{\etalchar{+}}21]{cerezo2021cost}
M.~Cerezo, Akira Sone, Tyler Volkoff, Lukasz Cincio, and Patrick~J. Coles.
\newblock Cost {{Function Dependent Barren Plateaus}} in {{Shallow Parametrized Quantum Circuits}}.
\newblock {\em Nature Communications}, 12(1):1791, 2021.
\newblock \href {https://doi.org/10.1038/s41467-021-21728-w} {\path{doi:10.1038/s41467-021-21728-w}}.

\bibitem[GGC{\etalchar{+}}24]{glick2024covariant}
Jennifer~R. Glick, Tanvi~P. Gujarati, Antonio~D. C{\'o}rcoles, Youngseok Kim, Abhinav Kandala, Jay~M. Gambetta, and Kristan Temme.
\newblock Covariant quantum kernels for data with group structure.
\newblock {\em Nature Physics}, 20(3):479--483, 2024.
\newblock \href {https://doi.org/10.1038/s41567-023-02340-9} {\path{doi:10.1038/s41567-023-02340-9}}.

\bibitem[GGPD25]{gil-fuster2025on}
Elies {Gil-Fuster}, Casper Gyurik, Adrian {P{\'e}rez-Salinas}, and Vedran Dunjko.
\newblock On the relation between trainability and dequantization of variational quantum learning models.
\newblock In {\em The Thirteenth International Conference on Learning Representations}, 2025.
\newblock URL: \url{https://openreview.net/forum?id=TdqaZbQvdi}.

\bibitem[HBK{\etalchar{+}}25]{henderson2025quantum}
Laura~J. Henderson, Kerstin Beer, Salini Karuvade, Riddhi Gupta, Angela White, and Sally Shrapnel.
\newblock Quantum advantage without exponential concentration: {{Trainable}} kernels for symmetry-structured data.
\newblock 2025.
\newblock \href {https://arxiv.org/abs/2509.14337} {\path{arXiv:2509.14337}}.

\bibitem[HBM{\etalchar{+}}21]{huang2021power}
Hsin-Yuan Huang, Michael Broughton, Masoud Mohseni, Ryan Babbush, Sergio Boixo, Hartmut Neven, and Jarrod~R. McClean.
\newblock Power of data in quantum machine learning.
\newblock {\em Nature Communications}, 12(1):2631, 2021.
\newblock \href {https://doi.org/10.1038/s41467-021-22539-9} {\path{doi:10.1038/s41467-021-22539-9}}.

\bibitem[HCT{\etalchar{+}}19]{havlicek2019supervised}
Vojtech Havlicek, Antonio~D. C{\'o}rcoles, Kristan Temme, Aram~W. Harrow, Abhinav Kandala, Jerry~M. Chow, and Jay~M. Gambetta.
\newblock Supervised learning with quantum enhanced feature spaces.
\newblock {\em Nature}, 567(7747):209--212, 2019.
\newblock \href {https://doi.org/10.1038/s41586-019-0980-2} {\path{doi:10.1038/s41586-019-0980-2}}.

\bibitem[HKP20]{huang2020predicting}
Hsin-Yuan Huang, Richard Kueng, and John Preskill.
\newblock Predicting many properties of a quantum system from very few measurements.
\newblock {\em Nature Physics}, 16(10):1050--1057, 2020.
\newblock \href {https://doi.org/10.1038/s41567-020-0932-7} {\path{doi:10.1038/s41567-020-0932-7}}.

\bibitem[JGM{\etalchar{+}}24]{jerbi2024shadows}
Sofiene Jerbi, Casper Gyurik, Simon~C. Marshall, Riccardo Molteni, and Vedran Dunjko.
\newblock Shadows of quantum machine learning.
\newblock {\em Nature Communications}, 15(1):5676, 2024.
\newblock \href {https://doi.org/10.1038/s41467-024-49877-8} {\path{doi:10.1038/s41467-024-49877-8}}.

\bibitem[KB15]{kingma2015adam}
Diederik~P. Kingma and Jimmy Ba.
\newblock Adam: {{A}} method for stochastic optimization.
\newblock In {\em International {{Conference}} on {{Learning Representations}}}, 2015.

\bibitem[KPM{\etalchar{+}}26]{kirby2026observationimprovedaccuracyclassical}
William Kirby, Bibek Pokharel, Javier~Robledo Moreno, Kevin~C. Smith, Sergey Bravyi, Abhinav Deshpande, Constantinos Evangelinos, Bryce Fuller, James~R. Garrison, Ben Jaderberg, Caleb Johnson, Petar Jurcevic, Su-un Lee, Simon Martiel, Mario Motta, Seetharami Seelam, Oles Shtanko, Kevin~J. Sung, Minh Tran, Vinay Tripathi, Kazuhiro Seki, Kazuya Shinjo, Han Xu, Lukas Broers, Tomonori Shirakawa, Seiji Yunoki, Kunal Sharma, and Antonio Mezzacapo.
\newblock Observation of {{Improved Accuracy}} over {{Classical Sparse Ground-State Solvers}} using a {{Quantum Computer}}.
\newblock (arXiv:2603.03496), 2026.
\newblock \href {https://arxiv.org/abs/2603.03496} {\path{arXiv:2603.03496}}.

\bibitem[KV94]{kearns1994introduction}
Michael~J. Kearns and Umesh Vazirani.
\newblock {\em An {{Introduction}} to {{Computational Learning Theory}}}.
\newblock The MIT Press, 1994.
\newblock \href {https://doi.org/10.7551/mitpress/3897.001.0001} {\path{doi:10.7551/mitpress/3897.001.0001}}.

\bibitem[LAT21]{liu2021rigorous}
Yunchao Liu, Srinivasan Arunachalam, and Kristan Temme.
\newblock A rigorous and robust quantum speed-up in supervised machine learning.
\newblock {\em Nature Physics}, 17(9):1013--1017, 2021.
\newblock \href {https://doi.org/10.1038/s41567-021-01287-z} {\path{doi:10.1038/s41567-021-01287-z}}.

\bibitem[LTW{\etalchar{+}}25]{larocca2025barren}
Mart{\'i}n Larocca, Supanut Thanasilp, Samson Wang, Kunal Sharma, Jacob Biamonte, Patrick~J. Coles, Lukasz Cincio, Jarrod~R. McClean, Zo{\"e} Holmes, and M.~Cerezo.
\newblock Barren plateaus in variational quantum computing.
\newblock {\em Nature Reviews Physics}, 7(4):174--189, 2025.
\newblock \href {https://doi.org/10.1038/s42254-025-00813-9} {\path{doi:10.1038/s42254-025-00813-9}}.

\bibitem[MBS{\etalchar{+}}18]{mcclean2018barren}
Jarrod~R. McClean, Sergio Boixo, Vadim~N. Smelyanskiy, Ryan Babbush, and Hartmut Neven.
\newblock Barren plateaus in quantum neural network training landscapes.
\newblock {\em Nature Communications}, 9(1):4812, 2018.
\newblock \href {https://doi.org/10.1038/s41467-018-07090-4} {\path{doi:10.1038/s41467-018-07090-4}}.

\bibitem[MGD26]{molteni2026exponential}
Riccardo Molteni, Casper Gyurik, and Vedran Dunjko.
\newblock Exponential quantum advantages in learning quantum observables from classical data.
\newblock {\em npj Quantum Information}, 12(1):19, 2026.
\newblock \href {https://doi.org/10.1038/s41534-025-01162-2} {\path{doi:10.1038/s41534-025-01162-2}}.

\bibitem[MN25]{mikelsen_SKQD2025}
Mathias Mikkelsen and Yuya~O. Nakagawa.
\newblock Quantum-selected configuration interaction with time-evolved state.
\newblock {\em Physical Review Research}, 7(4):043043, 2025.
\newblock \href {https://doi.org/10.1103/75pv-hbrx} {\path{doi:10.1103/75pv-hbrx}}.

\bibitem[PCGL20]{perez-salinas2020data}
Adri{\'a}n {P{\'e}rez-Salinas}, Alba {Cervera-Lierta}, Elies {Gil-Fuster}, and Jos{\'e}~I. Latorre.
\newblock Data re-uploading for a universal quantum classifier.
\newblock {\em Quantum}, 4:226, 2020.
\newblock \href {https://doi.org/10.22331/q-2020-02-06-226} {\path{doi:10.22331/q-2020-02-06-226}}.

\bibitem[QWZ{\etalchar{+}}23]{qi2023barren}
Han Qi, Lei Wang, Hongsheng Zhu, Abdullah Gani, and Changqing Gong.
\newblock The barren plateaus of quantum neural networks: Review, taxonomy and trends.
\newblock {\em Quantum Information Processing}, 22(12):435, 2023.
\newblock \href {https://doi.org/10.1007/s11128-023-04188-7} {\path{doi:10.1007/s11128-023-04188-7}}.

\bibitem[RBN{\etalchar{+}}22]{ragone2022representation}
Michael Ragone, Paolo Braccia, Quynh~T. Nguyen, Louis Schatzki, Patrick~J. Coles, Frederic Sauvage, Martin Larocca, and M.~Cerezo.
\newblock Representation {{Theory}} for {{Geometric Quantum Machine Learning}}.
\newblock 2022.
\newblock \href {https://arxiv.org/abs/2210.07980} {\path{arXiv:2210.07980}}.

\bibitem[RML14]{rebentrost2014quantum}
Patrick Rebentrost, Masoud Mohseni, and Seth Lloyd.
\newblock Quantum {{Support Vector Machine}} for {{Big Data Classification}}.
\newblock {\em Physical Review Letters}, 113(13):130503, 2014.
\newblock \href {https://doi.org/10.1103/PhysRevLett.113.130503} {\path{doi:10.1103/PhysRevLett.113.130503}}.

\bibitem[SBSW20]{schuld2020circuit}
Maria Schuld, Alex Bocharov, Krysta~M. Svore, and Nathan Wiebe.
\newblock Circuit-centric quantum classifiers.
\newblock {\em Physical Review A}, 101(3):032308, 2020.
\newblock \href {https://doi.org/10.1103/PhysRevA.101.032308} {\path{doi:10.1103/PhysRevA.101.032308}}.

\bibitem[Sch05]{schollwock2005density}
U.~Schollw{\"o}ck.
\newblock The density-matrix renormalization group.
\newblock {\em Reviews of Modern Physics}, 77(1):259--315, 2005.
\newblock \href {https://doi.org/10.1103/RevModPhys.77.259} {\path{doi:10.1103/RevModPhys.77.259}}.

\bibitem[SEM23]{schreiber2023classical}
Franz~J. Schreiber, Jens Eisert, and Johannes~Jakob Meyer.
\newblock Classical {{Surrogates}} for {{Quantum Learning Models}}.
\newblock {\em Physical Review Letters}, 131(10):100803, 2023.
\newblock \href {https://doi.org/10.1103/PhysRevLett.131.100803} {\path{doi:10.1103/PhysRevLett.131.100803}}.

\bibitem[Sho97]{shor1997polynomial}
Peter~W. Shor.
\newblock Polynomial-{{Time Algorithms}} for {{Prime Factorization}} and {{Discrete Logarithms}} on a {{Quantum Computer}}.
\newblock {\em SIAM Journal on Computing}, 26(5):1484--1509, 1997.
\newblock \href {https://doi.org/10.1137/S0097539795293172} {\path{doi:10.1137/S0097539795293172}}.

\bibitem[SKI{\etalchar{+}}25]{Sugisaki_SKQD2025}
Kenji Sugisaki, Shu Kanno, Toshinari Itoko, Rei Sakuma, and Naoki Yamamoto.
\newblock Hamiltonian simulation-based quantum-selected configuration interaction for large-scale electronic structure calculations with a quantum computer.
\newblock {\em Physical Chemistry Chemical Physics}, 27(38):20869--20884, 2025.
\newblock \href {https://doi.org/10.1039/d5cp02202a} {\path{doi:10.1039/d5cp02202a}}.

\bibitem[SP21]{schuld2021supervised}
Maria Schuld and Francesco Petruccione.
\newblock {\em Quantum {{Models}} as {{Kernel Methods}}}, pages 217--245.
\newblock Springer International Publishing, Cham, 2021.
\newblock \href {https://doi.org/10.1007/978-3-030-83098-4_6} {\path{doi:10.1007/978-3-030-83098-4_6}}.

\bibitem[SSM21]{schuld2021effectdataencodingexpressive}
Maria Schuld, Ryan Sweke, and Johannes~Jakob Meyer.
\newblock The effect of data encoding on the expressive power of variational quantum machine learning models.
\newblock {\em Physical Review A}, 103(3):032430, 2021.
\newblock \href {https://doi.org/10.1103/PhysRevA.103.032430} {\path{doi:10.1103/PhysRevA.103.032430}}.

\bibitem[Tro12]{tropp2012user}
Joel~A. Tropp.
\newblock User-{{Friendly Tail Bounds}} for {{Sums}} of {{Random Matrices}}.
\newblock {\em Foundations of Computational Mathematics}, 12(4):389--434, 2012.
\newblock \href {https://doi.org/10.1007/s10208-011-9099-z} {\path{doi:10.1007/s10208-011-9099-z}}.

\bibitem[UK26]{umeano2026ground}
Chukwudubem Umeano and Oleksandr Kyriienko.
\newblock Ground state-based quantum feature maps.
\newblock {\em APL Quantum}, 3(1), 2026.

\bibitem[Ver26]{vershynin2026high}
Roman Vershynin.
\newblock {\em High-Dimensional Probability}.
\newblock Cambridge University Press, New York, 1 edition, 2026.

\bibitem[Whi92]{white1992density}
Steven~R. White.
\newblock Density matrix formulation for quantum renormalization groups.
\newblock {\em Physical Review Letters}, 69(19):2863--2866, 1992.
\newblock \href {https://doi.org/10.1103/PhysRevLett.69.2863} {\path{doi:10.1103/PhysRevLett.69.2863}}.

\bibitem[XRV17]{xiao2017fashion}
Han Xiao, Kashif Rasul, and Roland Vollgraf.
\newblock Fashion-{{MNIST}}: A {{Novel Image Dataset}} for {{Benchmarking Machine Learning Algorithms}}.
\newblock 2017.
\newblock \href {https://arxiv.org/abs/1708.07747} {\path{arXiv:1708.07747}}.

\bibitem[YGH{\etalchar{+}}25]{ye2025quantum}
Qi~Ye, Shuangyue Geng, Zizhao Han, Weikang Li, L.-M. Duan, and Dong-Ling Deng.
\newblock Quantum automated learning with provable and explainable trainability.
\newblock (arXiv:2502.05264), 2025.
\newblock \href {https://arxiv.org/abs/2502.05264} {\path{arXiv:2502.05264}}.

\bibitem[YMI{\etalchar{+}}25]{yu2025quantum}
Jeffery Yu, Javier~Robledo Moreno, Joseph~T. Iosue, Luke Bertels, Daniel Claudino, Bryce Fuller, Peter Groszkowski, Travis~S. Humble, Petar Jurcevic, William Kirby, Thomas~A. Maier, Mario Motta, Bibek Pokharel, Alireza Seif, Amir Shehata, Kevin~J. Sung, Minh~C. Tran, Vinay Tripathi, Antonio Mezzacapo, and Kunal Sharma.
\newblock Quantum-{{Centric Algorithm}} for {{Sample-Based Krylov Diagonalization}}.
\newblock (arXiv:2501.09702), 2025.
\newblock \href {https://arxiv.org/abs/2501.09702} {\path{arXiv:2501.09702}}.

\bibitem[ZZN{\etalchar{+}}26]{zhao2026exponential}
Haimeng Zhao, Alexander Zlokapa, Hartmut Neven, Ryan Babbush, John Preskill, Jarrod~R. McClean, and Hsin-Yuan Huang.
\newblock Exponential quantum advantage in processing massive classical data.
\newblock (arXiv:2604.07639), 2026.
\newblock \href {https://arxiv.org/abs/2604.07639} {\path{arXiv:2604.07639}}.

\end{thebibliography}
\appendix

\newpage

\section{Energy-based generalization bounds}\label{app: energy gen bounds}
The bound on the energy generalization gap stated in Theorem~\ref{thm: energy gen bound} follows from two main technical ingredients. First, for any fixed value of $\alpha$, we bound the probability that the deviation $
\norm{H_{\mathcal{D}}(\alpha)-H_{\mathcal{T}}(\alpha)}
$
is large over the random draw of the training set $\mathcal{T}$. Second, we extend this fixed-parameter concentration bound uniformly over the entire parameter space, thereby obtaining control of the deviation simultaneously for all $\alpha$ in the parameter space.

Consider a classification problem with training points $\mathcal{T}=\{x_\ell,y_\ell\}_{\ell=1}^N$ where $x$ is sampled from an underlying distribution $\mathcal{D}$ over an input space $\Omega$ and $y_\ell\in\{0,1\}$. Let $ H_{\mathcal{T}}(\alpha)$ denote the training Hamiltonian defined in \cref{eq: encoding_hamiltonian_T}, constructed from an arbitrary encoding unitary $U(x,\alpha)$ parameterized by $\alpha\in\Theta =[-1,1]^M$.
For any labeled input $z=(x,y)$, we define $H'_z(\alpha)=I-H_z(\alpha)$, where $H_z(\alpha)=U^\dagger(x,\alpha)O(y)U(x,\alpha)$. We then define the population Hamiltonian as $H_{\mathcal{D}}(\alpha)=\mathbb{E}_{z\sim\mathcal{D}}[H'_z(\alpha)]$, namely, the average encoding Hamiltonian associated with a random labeled input drawn from the data distribution.
The key idea for the proof of Theorem~\ref{thm: energy gen bound} relies on controlling how much the empirical energy cost optimized during training can differ from the expected energy of the resulting trained state when averaged over the underlying data distribution. This amounts to controlling the deviation between the training Hamiltonian $H_{\mathcal{T}}(\alpha)$ and the population Hamiltonian $H_{\mathcal{D}}(\alpha)$.

Specifically, we first show that for any fixed $\alpha$, the following lemma holds.
\begin{lemma}\label{lemma: energy_bound}
Let $H_{\T}=I-\frac{1}{N}\sum_\ell H_{z_\ell}(\alpha)$ be a Hamiltonian embedding a training set $\T$ as in \cref{eq: encoding_hamiltonian_T}, such that $0\preceq H_z(\alpha)\preceq I$ for all points $z=(x,y)\in\Omega \times \{0,1\}$ and all $\alpha\in\Theta$, and let $H_{\mathcal{D}}(\alpha)=\mathbb{E}_{z\sim\mathcal{D}}[H'_z(\alpha)]$ be the population Hamiltonian, with $H'_z(\alpha)=I-H_z(\alpha)$. Then, for every $\alpha\in\Theta$, with probability at least $1-\delta$ over the draw of $N$-size training datasets $\T=\{(x_\ell,y_\ell)\}_{\ell=1}^N$,
\begin{equation}\label{eq: lemma_ham_difference_bound}
\norm{H_{\mathcal{T}}(\alpha)-H_{\mathcal{D}}(\alpha)}_{\infty}
\leq
\sqrt{
\frac{8}{N}
\log\left(\frac{2D}{\delta}\right)
}.
\end{equation}
where $D=2^n$ is the Hilbert space dimension.
\end{lemma}

The proof of Lemma~\ref{lemma: energy_bound} relies on the matrix Hoeffding inequality. We report here the relative theorem, the proof can be found in~\cite{tropp2012user}.
\begin{theorem}[Matrix Hoeffding (Theorem 1.3 in~\cite{tropp2012user})]\label{thm: matrix hoeffding}
Consider a finite sequence $\{X_k\}$ of independent, random, self-adjoint matrices with dimension $D$, and let $\{A_k\}$ be a sequence of fixed self-adjoint matrices. Assume that each random matrix satisfies
\begin{equation}
\mathbb{E}_k[X_k]=0
\qquad\text{and}\qquad
X_k^2 \preceq A_k^2 \text{ almost surely } 
\end{equation}

Then, for all $t\ge 0$,
\begin{equation}
\Pr\left\{
\lambda_{\max}\left(\sum_k X_k\right)\ge t
\right\}
\le
D\cdot e^{-t^2/8\sigma^2}, \text{ where }
\sigma^2
:=
\left\|
\sum_k A_k^2
\right\|_\infty.
\end{equation}
\end{theorem}

In our proof of Lemma~\ref{lemma: energy_bound} we will actually need to bound the operator norm of the sum of random matrices, not just the largest eigenvalue. However, we can obtain a bound on the operator norm of a sum of random matrices $\sum_k X_k$ by applying Theorem~\ref{thm: matrix hoeffding} twice: one to $\sum_k X_k$ and one to $-\sum_k X_k$ as explained in the following corollary.

\begin{corollary}[Matrix Hoeffding for operator norm (Corollary of Theorem 1.3 in~\cite{tropp2012user})]\label{cor: matrix hoeffding}
Consider a finite sequence $\{X_k\}$ of independent, random, self-adjoint matrices with dimension $D$, and let $\{A_k\}$ be a sequence of fixed self-adjoint matrices. Assume that each random matrix satisfies
\begin{equation}
\mathbb{E}_k[X_k]=0
\qquad\text{and}\qquad
X_k^2 \preceq A_k^2 \text{ almost surely } 
\end{equation}

Then, for all $t\ge 0$,
\begin{equation}
\Pr\left\{
\norm{\sum_k X_k}_\infty\ge t
\right\}
\le
2D\cdot e^{-t^2/8\sigma^2}, \text{ where }
\sigma^2
:=
\left\|
\sum_k A_k^2
\right\|_\infty.
\end{equation}
\end{corollary}
\begin{proof}
For a self-adjoint matrix $A$, the operator norm satisfies
$
\norm{A}_{\infty}
=
\max\left\{
\lambda_{\max}(A),
\lambda_{\max}(-A)
\right\}.
$
Therefore,
$
\norm{A}_{\infty}\ge t
$
is equivalent to
\begin{equation}
\lambda_{\max}(A)\ge t
\qquad\text{or}\qquad
\lambda_{\max}(-A)\ge t .
\end{equation}

Applying Tropp's one-sided matrix Hoeffding bound to $\sum_k X_k$ gives
\begin{equation}
\Pr\left[
\lambda_{\max}\left(\sum_k X_k\right)\ge t
\right]
\le
D\exp\left(-\frac{t^2}{8\sigma^2}\right).
\end{equation}

Applying the same bound to $-\sum_k X_k$ gives
\begin{equation}
\Pr\left[
\lambda_{\max}\left(-\sum_k X_k\right)\ge t
\right]
\le
D\exp\left(-\frac{t^2}{8\sigma^2}\right).
\end{equation}

Hence, by the union bound,
\begin{equation}
\begin{aligned}
\Pr\left[
\norm{\sum_k X_k}_{\infty}\ge t
\right]
&\le
\Pr\left[
\lambda_{\max}\left(\sum_k X_k\right)\ge t
\right]
+
\Pr\left[
\lambda_{\max}\left(-\sum_k X_k\right)\ge t
\right] \\
&\le
2D\exp\left(-\frac{t^2}{8\sigma^2}\right).
\end{aligned}
\end{equation}

\end{proof}

We are now ready to give the complete proof of our Lemma~\ref{lemma: energy_bound} 
\begin{proof}[Proof of Lemma~\ref{lemma: energy_bound}]
We define $X_z(\alpha) = H'_{z}(\alpha) - H_{\mathcal{D}}(\alpha)$. Note that $H_{\T}(\alpha) - H_{\mathcal{D}}(\alpha) = \frac{1}{N} \sum_{z\in\T} X_z(\alpha)$ and that   
\begin{equation}
  \mathbb{E}_{z\sim \D} [X_z(\alpha)]= \mathbb{E}_{z\sim \D} [H'_z(\alpha)]-H_{\mathcal{D}}(\alpha)=H_{\mathcal{D}}(\alpha)-H_{\mathcal{D}}(\alpha)=0  .
\end{equation}
 Moreover, $-I\preceq X_z \preceq I$  since we assumed $0\preceq H'_{z}\preceq I$ for any point $z$ and therefore also $0\preceq H_{d}\preceq I$. Hence, $\Vert X_z \Vert_{\infty} \leq 1$ for any data point $z$.  We can then bound the difference in norm between $H_\T(\alpha)$ and  $H_{\mathcal{D}}(\alpha)$ by
 \begin{equation}\label{eq: norm_difference}
     \norm{H_\T(\alpha)-H_{\mathcal{D}}(\alpha)}_\infty\leq \norm{\frac{1}{N} \sum_{z\in\T} X_z(\alpha)}_\infty
 \end{equation}
 In order to prove the lemma is then sufficient to bound the quantity $\norm{\frac{1}{N} \sum_{z\in\T} X_z(\alpha)}_\infty$. Since the samples $z\in\mathcal T$ are drawn independently, the matrices $X_z(\alpha)$ are also independent for every fixed $\alpha$. Moreover, under the assumption $0\preceq H_z(\alpha)\preceq I$, they satisfy $\norm{X_z(\alpha)}_\infty \le 1$ for all $z$. We can therefore apply the matrix Hoeffding inequality of Corollary~\ref{cor: matrix hoeffding}. In particular, in our case $X_k^2\preceq I$ and therefore $\sigma^2=\norm{\sum_{k=1}^N I}_{\infty}=N$ to obtain: 
 \begin{equation}
    \Pr_{\T}\!\left[\left\|\frac{1}{N}\sum_{i=1}^{N} X_i\right\|_\infty \ge \varepsilon\right]
\leq 
2D \exp\!\left(-\frac{N\varepsilon^2}{8}\right),
\end{equation}
 where $D =2^n$ is the Hilbert space dimension. Finally, after setting the right hand to be $\delta$, we get 
\begin{equation}
   \operatorname{Pr}_{\T}\left[ \left\|\frac{1}{N}\sum_{i=1}^{N} X_i\right\|_\infty \leq \sqrt{\frac{8}{N}\log\left(\frac{2 D}{\delta}\right)}\right] \geq 1-\delta~.
\end{equation}
which, combined with \cref{eq: norm_difference} proves the bound stated in Lemma~\ref{lemma: energy_bound}.

\end{proof}

Lemma~\ref{lemma: energy_bound} bounds the probability of having a large difference in operator norm between the training and population Hamiltonian for a fixed $\alpha$. Importantly, this probability is taken over different random datasets $\T$. In general, the set of trained parameter $\alpha^\star$ will depend on the particular training set used in the training, therefore when averaging over different $\T$ this dependence must be taken into account. When the Hamiltonian embedding is Lipschitz continuous in $\alpha$, the fixed-$\alpha$ bound can be extended uniformly over the whole parameter space. This is done by first applying the bound to all points of a finite covering net and then using the Lipschitz property to control the deviation at parameters outside the net. This is shown in the full proof of Theorem~\ref{thm: energy gen bound}, that we restate here for convenience.

\genbound*

\begin{proof}
Let the parameters take values in $\Theta=[-1,1]^M$. Equip the parameter space with the $\ell_\infty$ norm
$
\norm{\alpha-\beta}_\infty
:=
\max_{j=1,\ldots,M}|\alpha_j-\beta_j|$. Then $[-1,1]^M$ is exactly the radius-one $\ell_\infty$ ball in $\mathbb{R}^M$.
Let $\eta\in(0,1]$. An $\eta$-net $\mathcal A_\eta\subset\Theta$ is a finite set such that for every $\alpha\in\Theta$, there exists $\beta\in\A_\eta$ satisfying
$
\norm{\alpha-\beta}_\infty\le\eta$.
It is known~\cite{vershynin2026high} that the covering numbers of a ball $B_R$ of radius $R>0$ around a fixed point in $\mathbb{R}^M$, for $0\leq\eta\leq R$ is:
\begin{equation}
\mathcal{N}(B_R,||.||,\eta)=|\A_\eta|
\le
\left(1+\frac{2R}{\eta}\right)^M
\le
\left(\frac{3}{\eta}\right)^M.
\label{eq:covering-number}
\end{equation}
Where in the last step we consider $R=1$, as we are considering a ball of radius $R=1$ around 0 (that corresponds to $\Theta = [-1,1]^M$).
For each fixed $\beta\in\A_\eta$, \cref{eq: lemma_ham_difference_bound} gives
$
\text{Prob}_{\T}\!\left[
\norm{H_{\T}(\beta)-H_{\mathcal{D}}(\beta)}\ge t
\right]
\le
2D\exp\!\left(-\frac{Nt^2}{8}\right)$. We now bound the probability that there is a large deviation for at least a point in the net.
Define the bad event
\begin{equation}
E_\beta
:=
\left\{
\norm{H_{\T}(\beta)-H_{\mathcal{D}}(\beta)}\ge t
\right\}.
\end{equation}
The event that at least one net point is bad is $
\bigcup_{\beta\in\A_\eta}E_\beta$.
The union bound does not require the events $E_\beta$ to be independent. It gives
\begin{align}
\text{Prob}_{\T}\!\left[
\sup_{\beta\in\A_\eta}
\norm{H_{\T}(\beta)-H_{\mathcal{D}}(\beta)}\ge t
\right]
&\le
\sum_{\beta\in\A_\eta}\text{Prob}_{\T}(E_\beta)\\
&\le
2D|\A_\eta|\exp\!\left(-\frac{Nt^2}{8}\right)\\
&\le
2D\left(\frac{3}{\eta}\right)^M
\exp\!\left(-\frac{Nt^2}{8}\right).
\label{eq:union-bound}
\end{align}
To make the right-hand side at most $\delta$, it is enough to choose
\begin{equation}
t
=
\sqrt{
\frac{8}{N}
\left[
\log\!\left(\frac{2D}{\delta}\right)
+
M\log\!\left(\frac{3}{\eta}\right)
\right]
}.
\label{eq:t-net}
\end{equation}
Hence, with probability at least $1-\delta$,
\begin{equation}\label{eq: ham_difference_net}
\sup_{\beta\in\A_\eta}
\norm{H_{\T}(\beta)-H_{\mathcal{D}}(\beta)}
\le
\sqrt{
\frac{8}{N}
\left[
\log\!\left(\frac{2D}{\delta}\right)
+
M\log\!\left(\frac{3}{\eta}\right)
\right]
}.
\end{equation}

The finite-net bound is not yet enough: the final trained parameter $\alpha^\star(\T)$ need not lie exactly on the grid of points in $\mathcal{A}_\eta$. However, we can extend the bound to a supremum of all the parameter points in $\alpha \in [-1,1]^M$ by assuming regularity of the Hamiltonian family. In particular, Lipschitz continuity implies the existence of a $L\geq 0$ such that for every data point $z$ and every $\alpha,\beta\in[-1,1]^M$ it holds that $\norm{H_z(\alpha)-H_z(\beta)}_\infty
\le
L|\alpha-\beta|_\infty$. We notice that both the training Hamiltonian $H_\T(\alpha)$ and the full population Hamiltonian inherit the same Lipschitz constant:
\begin{align}
\norm{H_T(\alpha)-H_T(\beta)}
&\le
\frac{1}{N}\sum_{i=1}^{N}
\norm{H_{z_i}(\alpha)-H_{z_i}(\beta)}\\
&\le
L\norm{\alpha-\beta}_\infty.
\label{eq:lipschitz-empirical}
\end{align}
and, similarly,
\begin{align}
\norm{H_{\mathcal{D}}(\alpha)-H_{\mathcal{D}}(\beta)}
&\le
\mathbb{E}_z\!\left[
\norm{H_z(\alpha)-H_z(\beta)}
\right]\\
&\le
L\norm{\alpha-\beta}_\infty.
\label{eq:lipschitz-population}
\end{align}
Define $ \Delta_T(\alpha):=H_{\T}(\alpha)-H_{\mathcal{D}}(\alpha)$.
Combining Eqs.~\eqref{eq:lipschitz-empirical} and \eqref{eq:lipschitz-population},
\begin{align}
\norm{\Delta_T(\alpha)-\Delta_T(\beta)}
&\le
\norm{H_T(\alpha)-H_T(\beta)}
+
\norm{H_{\mathcal{D}}(\alpha)-H_{\mathcal{D}}(\beta)}\\
&\le
2L\norm{\alpha-\beta}_\infty.
\label{eq:lipschitz-delta}
\end{align}
Now take an arbitrary $\alpha\in[-1,1]^M$. Choose $\beta\in\A_\eta$ such that $ \norm{\alpha-\beta}_\infty\le\eta$.
Then
\begin{align}
\norm{\Delta_T(\alpha)}
&\le
\norm{\Delta_T(\beta)}
+
\norm{\Delta_T(\alpha)-\Delta_T(\beta)}\\
&\le
\norm{\Delta_T(\beta)}+2L\eta.
\end{align}
Taking the supremum over all $\alpha$ gives
\begin{equation}\label{eq: lipschitz error }
\sup_{\alpha\in[-1,1]^M}\norm{\Delta_T(\alpha)}
\le
\sup_{\beta\in\A_\eta}\norm{\Delta_T(\beta)}
+2L\eta.
\end{equation}

We can now conclude the proof of the Theorem by inserting \cref{eq: ham_difference_net} in \cref{eq: lipschitz error } in order to obtain:
\begin{equation}
    \sup_{\alpha\in[-1,1]^M}
\norm{H_{\mathcal{T}}(\alpha)-H_{\mathcal{D}}(\alpha)}
\leq
\sqrt{
\frac{8}{N}
\left[
\log\left(\frac{2D}{\delta}\right)
+
M\log\left(\frac{24L}{\varepsilon}\right)
\right]
}
+
\frac{\varepsilon}{4}.
\end{equation}
where we chose $\eta=\frac{\varepsilon}{8L}$.
This concludes the proof of Theorem~\ref{thm: energy gen bound}.

\end{proof}

Theorem~\ref{thm: energy gen bound} provides a bound on the energy generalization gap defined in \cref{eq: energy_gap}, it holds for any $\alpha$ and therefore even for the $\alpha^\star$ obtained after the training. The energy generalization gap can be directly used to bound the total risk $R$ defined in \cref{eq: total risk} as we show in Corollary~\ref{cor:classification-risk-bound}.

\totrisk*
\begin{proof}
Theorem~\ref{thm: energy gen bound} bounds the energy generalization gap for any parameter $\alpha\in\Theta=[-1,1]^M$ and, in particular, for the $\alpha^\star$ obtained after the model is trained. In particular, for any $\epsilon>0$,  with probability larger than $1-\delta$ over different training data $\T$, it holds that 
\begin{equation}
    \sup_{\alpha\in[-1,1]^M}
\norm{H_{\mathcal{T}}(\alpha)-H_{\mathcal{D}}(\alpha)}
\leq
\sqrt{
\frac{8}{N}
\left[
\log\left(\frac{2D}{\delta}\right)
+
M\log\left(\frac{24L}{\varepsilon}\right)
\right]
}
+
\frac{\varepsilon}{4}.
\end{equation}
Suppose now that the number of training samples satisfies
\begin{equation}\label{eq: appendix sample size}
N
\ge
\frac{128}{\varepsilon^2}
\left[
\log\left(\frac{2D}{\delta}\right)
+
M\log\left(\frac{24L}{\varepsilon}\right)
\right],
\end{equation}
then Theorem~\ref{thm: energy gen bound} implies that, with probability $1-\delta$,  $ \sup_{\alpha\in[-1,1]^M}
\norm{H_{\mathcal{T}}(\alpha)-H_{\mathcal{D}}(\alpha)}
\leq \frac{\varepsilon}{2}$. Recalling the definition of total cost, we obtain:
\begin{align}\label{eq: appendix bound_cost}
C(\psi_{{\scriptscriptstyle\mathcal{T}},\alpha^\star},\alpha^\star)
&=
\bra{\psi_{{\scriptscriptstyle\mathcal{T}},\alpha^\star}}
H_{\mathcal{D}}(\alpha^\star)
\ket{\psi_{{\scriptscriptstyle\mathcal{T}},\alpha^\star}} \\
&\le
\bra{\psi_{{\scriptscriptstyle\mathcal{T}},\alpha^\star}}
H_{\mathcal T}(\alpha^\star)
\ket{\psi_{{\scriptscriptstyle\mathcal{T}},\alpha^\star}}
+
\|H_{\mathcal{D}}(\alpha^\star)-H_{\mathcal T}(\alpha^\star)\| \\
&\le
\widehat C(\psi_{{\scriptscriptstyle\mathcal{T}},\alpha^\star},\alpha^\star)
+
\frac{\varepsilon}{2}.
\end{align}
It now remains to relate the population energy to the total classification error $R$. Recall that, for a data point $z=(x,y)$ embedded through the Hamiltonian $H_z(\alpha)$, the energy loss evaluated on the state $\psi$ is defined in \cref{eq: energy_loss} as:
\begin{equation}
   \ell_{\text{energy}}(\psi, H_z(\alpha))=1-\langle \psi | H_{z}(\alpha) | \psi \rangle . 
\end{equation}
 Moreover, the model prediction rule for a point $x$ is given by \cref{eq: prediction rule}, which we also state here for convenience:
\begin{equation}
    \widetilde{f}_{\psi_{{\scriptscriptstyle\mathcal{T}},\alpha^\star},\alpha^\star}(x')
    =
    \operatorname*{argmax}_{\tilde y\in\{0,1\}}
    P_{\psi_{{\scriptscriptstyle\mathcal{T}},\alpha^\star},\alpha^\star}(x',\tilde y).
\end{equation}
where $    P_{\psi_{{\scriptscriptstyle\mathcal{T}},\alpha},\alpha}(x,y) = \operatorname{Tr}\left(O(y) U(x,\alpha)|\psi_{{\scriptscriptstyle\mathcal{T}},\alpha}\rangle \langle \psi_{{\scriptscriptstyle\mathcal{T}},\alpha}|U^{\dagger}(x,\alpha)\right)$. We notice that if $(x,y)$ is misclassified, then the model assigns at least as much probability to the wrong label as to the true label, in other words
\begin{equation}
    \ell_{\mathrm{energy}}(\psi,H_z(\alpha))
=
1-
\bra{\psi_{{\scriptscriptstyle\mathcal{T}},\alpha}}H_z(\alpha)\ket{\psi_{{\scriptscriptstyle\mathcal{T}},\alpha}}
\ge
\frac12.
\end{equation}
for misclassified point $z$. This directly implies that $\mathbf 1\{\tilde f_{\alpha,\psi_{{\scriptscriptstyle\mathcal{T}},\alpha}}(x)\neq f(x)\}
\le
2\ell_{\mathrm{energy}}(\psi,H_z(\alpha))$ for any point $z=(x,y)$.
Therefore, taking the expectation over different point $z$ for $x\sim\mathcal{D}$ and remembering that $R(\psi,\alpha)=\Pr_{x\sim\mathcal{D}}[\widetilde{f}_{\psi_{{\scriptscriptstyle\mathcal{T}},\alpha},\alpha}(x)\neq f(x)]$, we obtain:
\begin{align}
R(\psi_{{\scriptscriptstyle\mathcal{T}},\alpha^\star},\alpha^\star)
&\le
2C(\psi_{{\scriptscriptstyle\mathcal{T}},\alpha^\star},\alpha^\star) \\
&\le
2\widehat C(\psi_{{\scriptscriptstyle\mathcal{T}},\alpha^\star},\alpha^\star)
+
\varepsilon.
\end{align}
where in the second equation we used the bound on the population cost in \cref{eq: appendix bound_cost} given by the training cost plus the energy generalization gap assuming the sample size given in \cref{eq: appendix sample size}.
    
\end{proof}

\subsection{Lipschitz assumption}
In Theorem~\ref{thm: energy gen bound}, we assumed that the parameterized Hamiltonians used to encode the data satisfy a Lipschitz continuity condition. Namely, we required that there exists a constant $L>0$ such that $\norm{H_z(\alpha)-H_z(\beta)}
\le
L\norm{\alpha-\beta}_\infty$ for any $\alpha,\beta\in[-1,1]^M$ and all the labeled input $z=(x,y)$. We now show that this is a natural assumption, and that it holds for common choices of encoding unitaries $U(x,\alpha)$ used in the construction of $H_z(\alpha)=U^\dagger(x,\alpha)O(y)U(x,\alpha)$. 

Suppose that the data are encoded through rotation angles given by linear combinations of the input features, for example
\begin{equation}
U(\alpha,x)
=
e^{-i\theta_\alpha(x)G},
\qquad
\theta_\alpha(x)
=
\sum_{j=1}^M \alpha_j x_j .
\end{equation}

If $
H_z(\alpha)
=
U(x,\alpha)^\dagger O(y)U(x,\alpha)$ with $
\|O(y)\|\le 1,$
then differentiating with respect to $\alpha_j$ gives a commutator-type expression. Since
$
\frac{\partial \theta_\alpha(x)}{\partial \alpha_j}
=
x_j$,
one obtains
\begin{equation}
\left\|
\frac{\partial H_z(\alpha)}{\partial \alpha_j}
\right\|
\le
2|x_j|\|G\|.
\end{equation}
If the features are uniformly bounded, $|x_j|\le B_x$, then
\begin{equation}
\left\|
\frac{\partial H_z(\alpha)}{\partial \alpha_j}
\right\|
\le
2B_x\|G\|.
\end{equation}

Therefore, moving coordinate by coordinate from $\alpha$ to $\beta$,
\begin{equation}
\|H_z(\alpha)-H_z(\beta)\|
\le
2B_x\|G\|
\sum_{j=1}^M
|\alpha_j-\beta_j|.
\end{equation}
In particular,
\begin{equation}
\|H_z(\alpha)-H_z(\beta)\|
\le
2MB_x\|G\|
\|\alpha-\beta\|_\infty .
\end{equation}
Thus, for Pauli generators with $\|G\|=1$, one may take $L\le 2MB_x$, up to convention-dependent factors such as the presence of a $1/2$ in the rotation generator.

\section{Gradient evaluation and the Hellmann–Feynman relation}\label{app: gradient_cost}
Here we justify why, when evaluating the gradient of the cost function in Eq.~\eqref{eq:cost-function}, the ground state can be treated as fixed at the value corresponding to the current parameters, despite its implicit dependence on $\alpha$. We formalize the result in Theorem~\ref{thm: gradient energy}.

\begin{theorem}[Energy cost gradient]\label{thm: gradient energy}
    Assume that the data Hamiltonian $H_{\mathcal T}(\alpha)
=
I-\frac{1}{N}\sum_{\ell=1}^{N}H_{z_\ell}(\alpha)$, as defined in Eq.~\eqref{eq: encoding_hamiltonian_T}, is continuously differentiable with a nondegenerate ground state.
    Define the energy cost as in Eq.~\eqref{eq:cost-function}
\begin{equation}
\widehat C_{\mathcal{T}}(\psi_{{\scriptscriptstyle\mathcal{T}},\alpha}, \alpha)
=
\langle\psi_{{\scriptscriptstyle\mathcal{T}},\alpha}|H_{\mathcal T}(\alpha)|\psi_{{\scriptscriptstyle\mathcal{T}},\alpha}\rangle.
\end{equation}

Then, for every parameter $\alpha_j$,
\begin{equation}
\left.
\frac{\partial}{\partial\alpha_j}
\widehat C_{\mathcal T}
(\psi_{{\scriptscriptstyle\mathcal T},\alpha},\alpha)
\right|_{\alpha=\alpha_t}
=
-\frac{1}{N}\sum_{\ell=1}^{N}
\langle\psi_{{\scriptscriptstyle\mathcal T},\alpha_t}|
\left.
\frac{\partial H_{z_\ell}(\alpha)}{\partial\alpha_j}
\right|_{\alpha=\alpha_t}
|\psi_{{\scriptscriptstyle\mathcal T},\alpha_t}\rangle.
\end{equation}
 
\end{theorem}
\begin{proof}
Remember that the cost function as a function of the parameter $\alpha$ is defined as:
\begin{equation}
\widehat C_{\mathcal T}(\psi_{{\scriptscriptstyle\mathcal{T}},\alpha},\alpha)
=
\bra{\psi_{{\scriptscriptstyle\mathcal{T}},\alpha}}
H_{\mathcal T}(\alpha)
\ket{\psi_{{\scriptscriptstyle\mathcal{T}},\alpha}},
\end{equation}
where $\ket{\psi_{{\scriptscriptstyle\mathcal{T}},\alpha}}$ is a normalized ground state of $H_{\mathcal T}(\alpha)$, satisfying
$
H_{\mathcal T}(\alpha)\ket{\psi_{{\scriptscriptstyle\mathcal{T}},\alpha}}
=
E_0(\alpha)\ket{\psi_{{\scriptscriptstyle\mathcal{T}},\alpha}}$ with $E_0(\alpha)$ the ground energy of the data Hamiltonian $H_{\mathcal{T}}(\alpha)$ at the point $\alpha$.

Assuming that the ground state is non-degenerate and differentiable with respect to $\alpha$, differentiating $\widehat C_{\mathcal T}(\psi_{{\scriptscriptstyle\mathcal{T}},\alpha},\alpha)$ with respect to $\alpha_j$ gives
\begin{align}
\frac{\partial \widehat C_{\mathcal T}(\psi_{{\scriptscriptstyle\mathcal{T}},\alpha},\alpha)}{\partial \alpha_j}
={}&
\bra{\partial_j\psi_{{\scriptscriptstyle\mathcal{T}},\alpha}}
H_{\mathcal T}(\alpha)
\ket{\psi_{{\scriptscriptstyle\mathcal{T}},\alpha}}
+
\bra{\psi_{{\scriptscriptstyle\mathcal{T}},\alpha}}
\frac{\partial H_{\mathcal T}(\alpha)}{\partial\alpha_j}
\ket{\psi_{{\scriptscriptstyle\mathcal{T}},\alpha}}
\nonumber\\
&+
\bra{\psi_{{\scriptscriptstyle\mathcal{T}},\alpha}}
H_{\mathcal T}(\alpha)
\ket{\partial_j\psi_{{\scriptscriptstyle\mathcal{T}},\alpha}},
\end{align}
where $\ket{\partial_j\psi_{{\scriptscriptstyle\mathcal{T}},\alpha}}
\equiv
\partial\ket{\psi_{{\scriptscriptstyle\mathcal{T}},\alpha}}/\partial\alpha_j$.
Using the eigenvalue equation, the first and third terms combine to give
\begin{equation}
    E_0(\alpha)\left(\bra{\partial_j\psi_{{\scriptscriptstyle\mathcal{T}},\alpha}}\ket{\psi_{{\scriptscriptstyle\mathcal{T}},\alpha}}+\bra{\psi_{{\scriptscriptstyle\mathcal{T}},\alpha}}\ket{\partial_j\psi_{{\scriptscriptstyle\mathcal{T}},\alpha}}\right).
\end{equation}

Since $\ket{\psi_{{\scriptscriptstyle\mathcal{T}},\alpha}}$ is normalized, differentiating
$\bra{\psi_{{\scriptscriptstyle\mathcal{T}},\alpha}}\ket{\psi_{{\scriptscriptstyle\mathcal{T}},\alpha}}=1$ yields
$\bra{\partial_j\psi_{{\scriptscriptstyle\mathcal{T}},\alpha}}\ket{\psi_{{\scriptscriptstyle\mathcal{T}},\alpha}}
+
\bra{\psi_{{\scriptscriptstyle\mathcal{T}},\alpha}}\ket{\partial_j\psi_{{\scriptscriptstyle\mathcal{T}},\alpha}}=0$.
Therefore,
\begin{equation}
\frac{\partial \widehat C_{\mathcal T}(\psi_{{\scriptscriptstyle\mathcal{T}},\alpha},\alpha)}{\partial\alpha_j}
=
\bra{\psi_{{\scriptscriptstyle\mathcal{T}},\alpha}}
\frac{\partial H_{\mathcal T}(\alpha)}{\partial\alpha_j}
\ket{\psi_{{\scriptscriptstyle\mathcal{T}},\alpha}},
\label{eq: hellmann_feynman}
\end{equation}
which is the Hellmann-Feynman theorem.

In particular, at iteration $t$ of Algorithm~\ref{alg:ham-embed}, after computing the ground state $\ket{\psi_{{\scriptscriptstyle\mathcal{T}},\alpha_t}}$ of $H_{\mathcal T}(\alpha_t)$, Eq.~\eqref{eq: hellmann_feynman} gives
\begin{equation}
\left.
\frac{\partial \widehat C_{\mathcal T}(\psi_{{\scriptscriptstyle\mathcal{T}},\alpha},\alpha)}{\partial\alpha_j}
\right|_{\alpha=\alpha_t}
=
\bra{\psi_{{\scriptscriptstyle\mathcal{T}},\alpha_t}}
\left.
\frac{\partial H_{\mathcal T}(\alpha)}
{\partial\alpha_j}
\right|_{\alpha=\alpha_t}
\ket{\psi_{{\scriptscriptstyle\mathcal{T}},\alpha_t}},
\end{equation}
Hence, although the ground state depends implicitly on the parameters $\alpha$, its derivative does not need to be evaluated when computing the gradient of the ground-state energy at the current parameter value. Equivalently, the frozen-state cost
$\widehat C_{\mathcal T}(\psi_{{\scriptscriptstyle\mathcal{T}},\alpha_t},\alpha)$
has the same gradient as the ground-state energy $E_0(\alpha)$ at $\alpha=\alpha_t$.

Finally, using
$H_{\mathcal T}(\alpha)
=
I - \frac{1}{N}\sum_{\ell=1}^{N}H_{z_\ell}(\alpha)$,
we obtain
\begin{equation}
\left.
\frac{\partial \widehat C_{\mathcal T}(\psi_{{\scriptscriptstyle\mathcal{T}},\alpha},\alpha)}{\partial\alpha_j}
\right|_{\alpha=\alpha_t}
=
-\frac{1}{N}
\sum_{\ell=1}^{N}
\bra{\psi_{{\scriptscriptstyle\mathcal{T}},\alpha_t}}
\left.
\frac{\partial H_{z_\ell}(\alpha)}
{\partial\alpha_j}
\right|_{\alpha=\alpha_t}
\ket{\psi_{{\scriptscriptstyle\mathcal{T}},\alpha_t}},
\end{equation}
which gives the gradient used in Eq.~\eqref{eq: derivative_cost} in the main text.
\end{proof}

\subsection{Gradients with approximate ground states}

In practice, the ground states in the optimization loop are obtained through quantum algorithms like QPE or SKQD. These may return only an approximate ground state
$\ket{\widetilde{\psi}_\alpha}$ of $H_{\mathcal T}(\alpha)$. Defining the corresponding approximate energy as
\begin{equation}
\widetilde{E}(\alpha)
=
\bra{\widetilde{\psi}_\alpha}
H_{\mathcal T}(\alpha)
\ket{\widetilde{\psi}_\alpha},
\end{equation}
its exact derivative with respect to $\alpha_j$ is
\begin{align}
\frac{\partial \widetilde{E}(\alpha)}{\partial \alpha_j}
={}&
\bra{\partial_j\widetilde{\psi}_\alpha}
H_{\mathcal T}(\alpha)
\ket{\widetilde{\psi}_\alpha}
+
\bra{\widetilde{\psi}_\alpha}
\frac{\partial H_{\mathcal T}(\alpha)}{\partial \alpha_j}
\ket{\widetilde{\psi}_\alpha}
\nonumber\\
&+
\bra{\widetilde{\psi}_\alpha}
H_{\mathcal T}(\alpha)
\ket{\partial_j\widetilde{\psi}_\alpha},
\end{align}
where
$\ket{\partial_j\widetilde{\psi}_\alpha}
\equiv
\partial\ket{\widetilde{\psi}_\alpha}/\partial\alpha_j$.
Therefore, unless $\ket{\widetilde{\psi}_\alpha}$ is an exact eigenstate, the
Hellmann-Feynman cancellation does not hold exactly for the derivative of
$\widetilde{E}(\alpha)$.

Nevertheless, the gradient obtained by keeping the approximate ground state
fixed provides a controlled approximation to the exact ground-state gradient.
Let $\ket{\psi_{{\scriptscriptstyle\mathcal{T}},\alpha}}$ denote the exact ground state of
$H_{\mathcal T}(\alpha)$ and define its fidelity with the approximate state as
$F(\alpha)=|\bra{\psi_{{\scriptscriptstyle\mathcal{T}},\alpha}}\ket{\widetilde{\psi}_\alpha}|^2$. Using Hölder's inequality together with the pure-state trace-distance relation
$\left\|
\ketbra{\widetilde{\psi}_\alpha}{\widetilde{\psi}_\alpha}
-
\ketbra{\psi_{{\scriptscriptstyle\mathcal{T}},\alpha}}{\psi_{{\scriptscriptstyle\mathcal{T}},\alpha}}
\right\|_1
=
2\sqrt{1-F(\alpha)}$,
we obtain
\begin{equation}
\left|
\bra{\widetilde{\psi}_\alpha}
\frac{\partial H_{\mathcal T}(\alpha)}{\partial\alpha_j}
\ket{\widetilde{\psi}_\alpha}
-
\bra{\psi_{{\scriptscriptstyle\mathcal{T}},\alpha}}
\frac{\partial H_{\mathcal T}(\alpha)}{\partial\alpha_j}
\ket{\psi_{{\scriptscriptstyle\mathcal{T}},\alpha}}
\right|
\nonumber 
\leq
2
\left\|
\frac{\partial H_{\mathcal T}(\alpha)}{\partial\alpha_j}
\right\|
\sqrt{1-F(\alpha)}.
\label{eq: approximate_hf_bound}
\end{equation}
Using the Hellmann--Feynman relation in Eq.~\eqref{eq: hellmann_feynman},
this gives
\begin{equation}
\left|
\bra{\widetilde{\psi}_\alpha}
\frac{\partial H_{\mathcal T}(\alpha)}{\partial\alpha_j}
\ket{\widetilde{\psi}_\alpha}
-
\frac{\partial E_0(\alpha)}{\partial\alpha_j}
\right|
\leq
2
\left\|
\frac{\partial H_{\mathcal T}(\alpha)}{\partial\alpha_j}
\right\|
\sqrt{1-F(\alpha)}.
\end{equation}
Hence, if the fidelity of the ground state found by the quantum solver is $F(\alpha)\geq 1-\epsilon$, the error in the gradient obtained from
the approximate ground state is at most
$2\left\|\partial H_{\mathcal T}(\alpha)/\partial\alpha_j\right\|
\sqrt{\epsilon}$. Moreover, as discussed in Appendix~\ref{app: barren plateaus}, for the standard Pauli encoding introduced in the main text, the operator norm $\left\|\partial H_{\mathcal T}(\alpha)/\partial\alpha_j\right\|$ is bounded.

\section{Absence of barren plateaus in the Hamiltonian embedding model}\label{app: barren plateaus}
One appealing feature of Hamiltonian embedding is that, for local encoding Hamiltonians, once the dataset-encoding ground state has been obtained, inference reduces to estimating a local observable on the learned state. Together with enabling classical-shadow-based inference, this locality structure also allows us to establish trainability guarantees for the model. In particular, for a particular class of Pauli-rotation encoding $U(x,\alpha)$, we can theoretically prove that a nontrivial subset of the trainable parameters avoids
barren plateaus.  We present the proof for the simple observable $A=1/n\sum_i Z_i$ and a family of shallow Hamiltonian embedding architectures. The same argument extends to observables that are sums of bounded-locality Pauli terms, provided that each Pauli term entering the bound anticommutes with at least one final-layer rotation generator acting on its support.

\begin{theorem}[Cost-level absence of barren plateaus]\label{thm: no bp}
Let $U(x,\alpha)$ be a depth-$K$ layered unitary of the form
$
U(x,\alpha)
=
V^{K}(\alpha^{K})U_{\rm enc}^{K}(x)
\cdots
V^{1}(\alpha^{1})U_{\rm enc}^{1}(x)$,
where $\alpha=(\alpha^{1},\ldots,\alpha^{K})$ contains $M$ trainable parameters. Assume
that $A=\frac{1}{n}\sum^n_{i=1}Z_i$ and the final trainable layer contains one independent  $X$- or $Y$-rotation on each
readout qubit, parameterized by the angles
$\alpha_1^K,\ldots,\alpha_n^K$ with $\alpha_i^K\in[0,2\pi]$. Let
\begin{equation}
C_t(\alpha)
:=
\widehat C(\psi_{{\scriptscriptstyle\mathcal{T}},\alpha_t},\alpha)
=
\mathbb E_{(x,y)\in\mathcal T}
\left[
\ell_{\rm eng}
\left(
\psi_{{\scriptscriptstyle\mathcal{T}},\alpha_t},
H_{(x,y)}(\alpha)
\right)
\right]
\end{equation}
be the empirical energy cost at optimization step $t$, where the state
$\psi_{{\scriptscriptstyle\mathcal{T}},\alpha_t}$ is kept fixed when differentiating with respect to $\alpha$.
Assume that, for some $\gamma>0$,
\begin{equation}\label{eq: energy_assumpion}
C_t(\alpha_t)
\leq
\frac12-\gamma.
\end{equation}
Then the coordinate-averaged variance of the gradients of the cost function satisfies
\begin{equation}\label{eq: bp result}
\frac1M
\sum_{\mu=1}^M
\operatorname{Var}_{\alpha_\mu}
\left(
\frac{\partial C_t(\alpha)}{\partial \alpha_\mu}
\right)
\geq
\frac{\gamma^2}{2M n},
\end{equation}
where $\operatorname{Var}_{\alpha_\mu}$ denotes the coordinate-wise variance with respect to
$\alpha_\mu\sim\operatorname{Unif}([0,2\pi])$ while keeping
all the other parameters fixed at their optimization-step values.
\end{theorem}

\begin{proof}
At optimization step $t$, we keep the state $\ket{\psi_{{\scriptscriptstyle\mathcal{T}},\alpha_t}}$ fixed and optimize only
the circuit parameters appearing in the data-dependent Hamiltonian. For a single training
example $z_\ell=(x_\ell,y_\ell)$, the fixed-state energy loss is defined as in \cref{eq: energy_loss}:
\begin{equation}
\ell_{\mathrm{eng}}(\psi_{{\scriptscriptstyle\mathcal{T}},\alpha_t},H_{z_\ell}(\alpha))
=
1-\bra{\psi_{{\scriptscriptstyle\mathcal{T}},\alpha_t}}
H_{z_\ell}(\alpha)
\ket{\psi_{{\scriptscriptstyle\mathcal{T}},\alpha_t}},
\end{equation}
where the state $\psi_{{\scriptscriptstyle\mathcal{T}},\alpha_t}$ is kept fixed when differentiating with respect to
$\alpha$. Let
$\rho_t
=
\ket{\psi_{{\scriptscriptstyle\mathcal{T}},\alpha_t}}\bra{\psi_{{\scriptscriptstyle\mathcal{T}},\alpha_t}}$.
Recalling the expression for $H_T(\alpha)$, with
$O(y)=\frac{I+(-1)^yA}{2}$ and $A=\frac1n\sum_{j=1}^n Z_j$,  we define
\begin{equation}\label{eq: hi}
h_j(x;\alpha)
:=
\operatorname{Tr}
\left[
\rho_t
U(x,\alpha)^\dagger Z_j U(x,\alpha)
\right].
\end{equation}
Let us define $s(x)=(-1)^{y(x)}$.
Then the fixed-state single-sample energy loss is
$\ell_{\mathrm{eng}}(\psi_{{\scriptscriptstyle\mathcal{T}},\alpha_t},H_{(x,y)}(\alpha))
=
\frac12
-
\frac{s(x)}{2n}
\sum_{j=1}^n
h_j(x;\alpha)$,
while the frozen total energy cost at optimization step $t$ is
\begin{equation}
C_t(\alpha)
=
\frac12
-
\frac{1}{2n}
\sum_{j=1}^n
\mathbb E_{x\in \T}
\left[
s(x)h_j(x;\alpha)
\right].
\end{equation}
The average is taken over the training set $\T$. By the assumption in \cref{eq: energy_assumpion},
$C_t(\alpha_t)
\leq
\frac12-\gamma$.
Equivalently,
\begin{equation}\label{eq: direct_low_energy_signal}
\frac1n
\sum_{j=1}^n
\mathbb E_{x\in \T}
\left[
s(x)h_j(x;\alpha_t)
\right]
\geq
2\gamma .
\end{equation}

We now isolate the dependence of the encoding unitary $U(x,\alpha)$ on the final trainable layer.
After fixing all previous layers, including both the data-encoding gates and the
trainable gates in layers $1,\ldots,K-1$, we write
$U(x,\alpha)
=
V_K(\alpha^K)U_{<K}(x,\alpha^{<K})$,
where $U_{<K}(x,\alpha^{<K})$ denotes the product of all layers before the final one, and
$V_K(\alpha^K)$ is the last trainable layer. At optimization step $t$, the parameters
$\alpha^{<K}$ are fixed at $\alpha_t^{<K}$, so we absorb them into the frozen unitary
$W_t(x)
:=
U_{<K}(x,\alpha_t^{<K})$.
Thus, for the purpose of analyzing the final layer at optimization step $t$, we write
$U(x,\alpha^K)
=
V_K(\alpha^K)W_t(x)$.
In particular, we take the final layer to consist of one trainable single-qubit rotation on each readout
qubit,
$\alpha^K
=
(\alpha_1^K,\ldots,\alpha_n^K)$.

By assumption, the rotation acting on qubit $i$ is generated by a Pauli operator that does
not commute with the measured observable $Z_i$. For example, if the final gate on qubit $i$ is an
$X_i$-rotation, then
$R_{X_i}(\alpha_i^K)^\dagger
Z_i
R_{X_i}(\alpha_i^K)
=
Z_i\cos \alpha_i^K
+
Y_i\sin \alpha_i^K$.
Equivalently, after freezing $\rho_t$, the dataset, all earlier parameters, and all
final-layer parameters except $\alpha_i^K$, there exist real coefficients $a_i(x)$ and
$b_i(x)$ such that we can rewrite $h_i$ in \cref{eq: hi} as
\begin{equation}\label{eq: direct_hi_sinusoid}
h_i(x;\alpha)
=
a_i(x)\cos \alpha_i^K
+
b_i(x)\sin \alpha_i^K .
\end{equation}
Note that the coefficients $a_i(x)$ and $b_i(x)$ are independent of $\alpha_i^K$.

We now pass to the empirical cost by averaging $\ell_{\text{eng}}$ over the input $x$. Since $\alpha_i^K$ is independent of the data point $x$, we can 
define the labeled empirical coefficients
$
\overline A_i
:=
\mathbb E_{x\in \T}
\left[
s(x)a_i(x)
\right]
$
and
$
\overline B_i
:=
\mathbb E_{x\in \T}
\left[
s(x)b_i(x)
\right]$ and obtain that:
\begin{equation}\label{eq: empirical_signal_direct}
\mathbb E_{x\in \T}
\left[
s(x)h_i(x;\alpha)
\right]
=
\overline A_i\cos \alpha_i^K
+
\overline B_i\sin \alpha_i^K .
\end{equation}
Define the empirical amplitude $
R_i
:=
\sqrt{\overline A_i^2+\overline B_i^2}$.
By the amplitude-phase form of a sinusoid, there exists a phase $\varphi_i$ such that
\begin{equation}\label{eq: empirical_signal_amplitude_phase}
\mathbb E_{x\in \T}
\left[
s(x)h_i(x;\alpha)
\right]
=
R_i\cos(\alpha_i^K+\varphi_i).
\end{equation}
This implies that, for every value of $\alpha_i^K$,
\begin{equation}\label{eq: empirical_signal_bounded_by_R}
\left|
\mathbb E_{x\in \T}
\left[
s(x)h_i(x;\alpha)
\right]
\right|
\leq
R_i .
\end{equation}
Evaluating \cref{eq: empirical_signal_bounded_by_R} at $\alpha=\alpha_t$ we obtain:
\begin{equation}
\mathbb E_{x\in \T}
\left[
s(x)h_i(x;\alpha_t)
\right]
\leq
\left|
\mathbb E_{x\in \T}
\left[
s(x)h_i(x;\alpha_t)
\right]
\right|
\leq
R_i .
\end{equation}
Therefore,
$
\frac1n
\sum_{i=1}^n
R_i
\geq
2\gamma $ and, by Jensen's inequality,
\begin{equation}\label{eq: bound Ri}
\frac1n
\sum_{i=1}^n
R_i^2
\geq
\left(
\frac1n
\sum_{i=1}^n
R_i
\right)^2
\geq
4\gamma^2 .
\end{equation}
Equivalently, $
\sum_{i=1}^n
R_i^2
\geq
4n\gamma^2 $.

Using \cref{eq: empirical_signal_direct}, the frozen empirical cost can be written as
\begin{equation}\label{eq: cost_direct_final_layer_form}
C_t(\alpha^K)
=
\frac12
-
\frac{1}{2n}
\sum_{i=1}^n
\left(
\overline A_i\cos\alpha_i^K
+
\overline B_i\sin\alpha_i^K
\right),
\end{equation}
where all earlier-layer parameters are fixed at $\alpha_t^{<K}$.

 The derivative of the frozen empirical cost with respect to the final-layer parameter $\alpha_i^K$ can be computed using Eq.~\eqref{eq: cost_direct_final_layer_form}
\begin{equation}
\frac{\partial C_t(\boldsymbol{\alpha}^K)}
{\partial \alpha_i^K}
=
\frac{1}{2n}
\left(
A_i \sin \alpha_i^K
-
B_i \cos \alpha_i^K
\right).
\label{eq:gradient_final_layer}
\end{equation}
Following the standard barren-plateau analysis, we evaluate the gradient variance with respect to independently and uniformly distributed final-layer angles
$\alpha_i^K\sim\mathrm{Unif}([0,2\pi])$, while keeping all other
parameters fixed. Since
\begin{equation}
\mathbb{E}_{\alpha_i^K}
\left[
\frac{\partial C_t}{\partial \alpha_i^K}
\right]
=0,
\end{equation}
we obtain
\begin{align}
\operatorname{Var}_{\alpha_i^K}
\left(
\frac{\partial C_t}{\partial \alpha_i^K}
\right)
&=
\mathbb{E}_{\alpha_i^K}
\left[
\left(
\frac{\partial C_t}{\partial \alpha_i^K}
\right)^2
\right]
\nonumber\\
&=
\frac{1}{4n^2}
\mathbb{E}_{\alpha_i^K}
\left[
\left(
A_i\sin\alpha_i^K-B_i\cos\alpha_i^K
\right)^2
\right]
\nonumber\\
&=
\frac{A_i^2+B_i^2}{8n^2}
=
\frac{R_i^2}{8n^2}.
\label{eq:gradient_variance_final_layer}
\end{align}
Summing over the $n$ final-layer parameters and using
Eq.~\eqref{eq: bound Ri}, we therefore find
\begin{equation}
\sum_{i=1}^n
\operatorname{Var}_{\alpha_i^K}
\left(
\frac{\partial C_t}{\partial \alpha_i^K}
\right)
=
\frac{1}{8n^2}
\sum_{i=1}^n R_i^2
\ge
\frac{\gamma^2}{2n}.
\label{eq:final_layer_gradient_variance}
\end{equation}

Finally, let $M$ denote the total number of trainable parameters in
the circuit. Since the variance associated with every parameter
direction is nonnegative, the average variance over all trainable
parameters is lower bounded by the contribution of the $n$
final-layer parameters alone:
\begin{align}
\frac{1}{M}
\sum_{\mu=1}^M
\operatorname{Var}_{\alpha_\mu}
\left(
\frac{\partial C_t}{\partial\alpha_\mu}
\right)
&\ge
\frac{1}{M}
\sum_{i=1}^n
\operatorname{Var}_{\alpha_i^K}
\left(
\frac{\partial C_t}{\partial\alpha_i^K}
\right)
\nonumber\\
&\ge
\frac{\gamma^2}{2Mn}.
\end{align}
This proves Eq.~\eqref{eq: bp result}.

If the ansatz has depth $K$ and $O(n)$ trainable parameters per
layer, then $M=O(Kn)$. Hence, there exists a constant $c>0$ such that
\begin{equation}
\frac{1}{M}
\sum_{\mu=1}^M
\operatorname{Var}_{\alpha_\mu}
\left(
\frac{\partial C_t}{\partial\alpha_\mu}
\right)
\ge
\frac{c\gamma^2}{Kn^2}.
\end{equation}

\end{proof}

We now show that the results in Theorem~\ref{thm: no bp} are enough to guarantee that a non-trivial fraction of parameters in the model have gradient variance that remains inverse-polynomial in the system size, and therefore avoid barren plateaus. 

\begin{corollary}[Non-negligible fraction of parameters avoids barren plateaus.]
\label{cor:bp-fraction}
Under the assumptions of Theorem~\ref{thm: no bp}, let
\begin{equation}
S_\gamma
:=
\left\{
i\in\{1,\ldots,n\}:
\operatorname{Var}_{\alpha_i^K}
\left(
\frac{\partial C_t}{\partial \alpha_i^K}
\right)
\ge
\frac{\gamma^2}{4n^2}
\right\}.
\end{equation}

Then
\begin{equation}
|S_\gamma|
\ge
\frac{2\gamma^2}{1-2\gamma^2}\,n.
\end{equation}
Hence, if the total number of trainable parameters satisfies
$M=O(Kn)$, an $\Omega(\gamma^2/K)$ fraction of all trainable
parameters has gradient variance at least $\gamma^2/(4n^2)$.
In particular, for $K=O(1)$ number of layers and $\gamma=\Theta(1)$,
a constant fraction of the trainable parameters has at least inverse-polynomial gradient variance.
\end{corollary}

\begin{proof}
From Eq.~\eqref{eq: bound Ri} in Theorem~\ref{thm: no bp}, we have
$
\sum_{i=1}^n R_i^2 \ge 4n\gamma^2.$
Moreover, $R_i\le 1$ for every $i$. Indeed, by Eq.~\eqref{eq: empirical_signal_amplitude_phase},
\begin{equation}
R_i
=
\max_{\alpha_i^K}
\left|
\mathbb{E}_{x\in\mathcal T}
\left[
s(x)h_i(x;\alpha)
\right]
\right|,
\end{equation}
and since $s(x)\in\{-1,+1\}$ and
$h_i(x;\alpha)=\operatorname{Tr}
[\rho_t U(x,\alpha)^\dagger Z_i U(x,\alpha)]\in[-1,1]$, the absolute value of the expectation is at most one.

Define the set of indices $i$ whose corresponding $R^2_i$ is greater than $2\gamma^2$
\begin{equation}
\widetilde S_\gamma
:=
\left\{
i\in\{1,\ldots,n\}:R_i^2\ge 2\gamma^2
\right\},
\end{equation}
and let $m:=|\widetilde S_\gamma|$. For the indices in
$\widetilde S_\gamma$ we use the bound $R_i^2\le 1$ derived above, while for the
remaining $n-m$ indices we have $R_i^2<2\gamma^2$. Therefore,
\begin{equation}
\sum_{i=1}^n R_i^2
\le
m+2\gamma^2(n-m).
\end{equation}
Combining this with Eq.~\eqref{eq: bound Ri} gives
$
4n\gamma^2
\le
m+2\gamma^2(n-m)$,
and hence
\begin{equation}
m(1-2\gamma^2)\ge 2n\gamma^2.
\end{equation}
Since the assumption $C_t(\alpha_t)\le 1/2-\gamma$ implies
$0<\gamma\le 1/2$, we have $1-2\gamma^2>0$, and thus the number of indices $i$ whose $R^2_i$ is greater than $2\gamma^2$ is  
\begin{equation}
m
\ge
\frac{2\gamma^2}{1-2\gamma^2}\,n.
\end{equation}

For every $i\in\widetilde S_\gamma$, the same argument as in Theorem~\ref{thm: no bp} allows us to bound the variance of the gradient for the parameter $\alpha_i^K$ using the quantity $R_i$
\begin{equation}
\operatorname{Var}_{\alpha_i^K}
\left(
\frac{\partial C_t}{\partial \alpha_i^K}
\right)
=
\frac{R_i^2}{8n^2}
\ge
\frac{\gamma^2}{4n^2}.
\end{equation}
Therefore $\widetilde S_\gamma\subseteq S_\gamma$, and hence
\begin{equation}
|S_\gamma|
\ge
|\widetilde S_\gamma|
\ge
\frac{2\gamma^2}{1-2\gamma^2}\,n.
\end{equation}

Finally, if $M=O(Kn)$, then
\begin{equation}
\frac{|S_\gamma|}{M}
=
\Omega\!\left(\frac{\gamma^2}{K}\right).
\end{equation}
Thus, for constant depth $K=O(1)$ and
$\gamma\ge 1/\operatorname{poly}(n)$, an inverse-polynomial fraction
of all trainable parameters has inverse-polynomial gradient variance.
\end{proof}

Our result concerns the cost evaluated for different $\alpha$ with the prepared ground state held fixed, as in line 21 of Algorithm~\ref{alg:ham-embed}. It shows that varying a non-negligible fraction of the parameters can produce cost differences of at least inverse-polynomial size. We note that, in general, the absence of barren plateaus alone does not guarantee successful training, as optimization may still encounter local minima or other stationary points.

\section{Universality of the Hamiltonian Embedding model}
\label{app:shallow-universality}

In this Appendix, we show that the Hamiltonian embedding model allows for universal construction. In particular, given a measurable target function $f:\Omega\subset \mathbb{R}^d \to\{0,1\}$ with inputs from $\Omega=[-\pi L,\pi L]^d$, where $L>0$ is an arbitrary parameter determining the size of the input domain, we construct an HE model with a fixed product encoding $U(x)$ whose ground-state classifier approximates $f$ up to $\varepsilon$ error under
the uniform input distribution on~$\Omega$.

\begin{theorem}[Universality with constant-depth encodings]
Let $\Omega=[-\pi L,\pi L]^d$ be the input space, $\mu$ be the uniform
distribution on $\Omega$ and let
$f:\Omega\to\{0,1\}$ be a measurable function. For every
$\varepsilon,\delta\in(0,1)$, there exist finite
$n,N\in\mathbb N$ and a constant depth encoding $U(x)$ such that given a training dataset
$\mathcal T=\{(x_\ell,f(x_\ell))\}_{\ell=1}^N$,
with $x\sim \mu$ the following holds. Let
$|\psi_{\mathcal T}\rangle$ be a ground state of
\begin{equation}
H_{\mathcal T}
=I-\frac1N\sum_{\ell=1}^N
U(x_\ell)^\dagger O(f(x_\ell))U(x_\ell),
\end{equation}
where $O(y)=(I+(-1)^yA)/2$ with $A=n^{-1}\sum_{i=1}^n Z_i$.
Then, the HE classifier $\tilde f_{\mathcal T}$ defined by the prediction rule in Eq.~\eqref{eq: prediction rule}, using the fixed encoding $U(x)$ and the ground state of $H_{\mathcal T}$, satisfies with at least $1-\delta$ probability over
the training dataset,
\begin{equation}
\Pr_{x\sim\mu}
[\tilde f_{\mathcal T}(x)\neq f(x)]
\leq\varepsilon.
\end{equation}
\end{theorem}

\begin{proof}

First, we write the target function as $s(x)=(-1)^{f(x)}$ and compute
its Fourier coefficients
$
\widehat s_k=\mathbb E_\mu[s(x)e^{-ik\cdot x/L}]=|\widehat s_k|e^{i\phi_k}$, where $\mu$ is the uniform distribution over $\Omega$.
Since $s$ is bounded and measurable, it belongs to $L^2(\mu)$. Completeness of the Fourier basis therefore ensures that, for every $\varepsilon\in(0,1)$, a finite real Fourier sum $p$ exists: 
\begin{equation}
p(x)=\sum_{k\in K}a_k\cos(k\cdot x/L+\phi_k),
\end{equation}
where $a_k=2|\widehat s_k|$ for $k\neq0$, $a_0=|\widehat s_0|$ for the constant term, and $\phi_k=\arg\widehat s_k$.
Parseval's identity then guarantees that with a finite $|K|$ it holds that $\|s-p\|_{L^2(\mu)}^2\leq\varepsilon/4$.  Here $K$ is the set of frequencies retained from the original Fourier decomposition of $s$, excluding those with zero Fourier coefficients and keeping one frequency from each nonzero pair $\{k,-k\}$. For example, we can take $K$ to contain frequencies $k\in\mathbb Z^d$ whose entries have absolute value at most a cutoff $M$, chosen such that $\|s-p\|_{L^2(\mu)}^2\leq\varepsilon/4$.

We now construct the encoding $U$ such that the HE model approximates the function $p$. Set
$W=\sum_{k\in K}a_k$. For each frequency $k\in K$, we choose integer multiplicities $m_k$ such that $\sum_{k\in K}m_k=n$ and $|m_k/n-a_k/W|\leq1/n$, where $n$ is the number of qubits we use to encode the HE model. We then define the following HE model
\begin{equation}\label{eq: universal encoding}
U(x)=\bigotimes_{k\in K}R_Y(k\cdot x/L)^{\otimes m_k},
\qquad
A=\frac1n\sum_{i=1}^n Z_i,
\qquad
O(y)=\frac{I+(-1)^yA}{2}.
\end{equation}
Thus, each frequency $k\in K$ appearing in the truncated Fourier series $p$ specifies a single-qubit
rotation, and its Fourier magnitude determines the fraction
of qubits encoding that frequency. The phases $\phi_k$
will be assigned during training in the following way.

Let us consider a draw of independent inputs $x_\ell\sim\mu$ and labeled by $y_\ell=f(x_\ell)$. Define the empirical correlations
\begin{equation}
C_k=\frac1N\sum_{\ell=1}^N
(-1)^{y_\ell}e^{-ik\cdot x_\ell/L}.
\end{equation}
Notice that $C_k$ approaches the Fourier coefficient $\widehat{s}_k$ of the target $s$ as $N\rightarrow\infty$. The real and imaginary part of $C_k$ are given by Euler's formula: 
\begin{equation}\label{eq: real imag Ck}
\operatorname{Re}(C_k)
=\frac1N\sum_{\ell=1}^N(-1)^{y_\ell}
\cos(k\cdot x_\ell/L),
\qquad
\operatorname{Im}(C_k)
=-\frac1N\sum_{\ell=1}^N(-1)^{y_\ell}
\sin(k\cdot x_\ell/L).
\end{equation}

Since the encoding in Eq.~\eqref{eq: universal encoding} acts independently on each qubit, the
identity $R_Y(\theta)^\dagger ZR_Y(\theta)
=\cos\theta\,Z-\sin\theta\,X$ gives for each qubit $i$:
\begin{equation}
U(x)^\dagger Z_iU(x)
=\cos(k_i\cdot x/L)Z_i-\sin(k_i\cdot x/L)X_i,
\end{equation}
where $k_i$ is the frequency assigned to qubit $i$.

We can now write the embedded Hamiltonian in terms of $C_k$. Substituting these identities and the readout $O(y)$ into
the definition of $H_{\mathcal T}$ in Eq.~\eqref{eq: encoding_hamiltonian_T} yields
\begin{equation}\label{eq: Ham universal}
H_{\mathcal T}
=I-\frac1N\sum_{\ell=1}^N
U(x_\ell)^\dagger O(y_\ell)U(x_\ell)
=\frac12I-\frac1{2n}\sum_{i=1}^n
\left[
\operatorname{Re}(C_{k_i})Z_i
+\operatorname{Im}(C_{k_i})X_i
\right].
\end{equation}
By writing $C_k$ in polar form,  $C_k=|C_k|e^{i\theta_k}$, it is easy to show that the ground state of the embedded Hamiltonian in Eq.~\eqref{eq: Ham universal} is given by product state of one-qubit states $R_Y(\arg C_{k_i})|0\rangle$, for the qubits encoding the frequency $k_i$.
In fact, each one-qubit Hamiltonian term is real and symmetric in the computational basis, so a ground state can be chosen with real amplitudes. Since every normalized real qubit state can be written as $R_Y(\beta)|0\rangle$, it suffices to minimize the energy over this one-parameter family. Its Pauli expectations are
$\langle Z\rangle=\cos\beta$ and
$\langle X\rangle=\sin\beta$, so the corresponding
energy term has expectation
\begin{equation}
E_k(\beta)
=\frac{1}{2n}[1-|C_k|
\left(\cos\theta_k\cos\beta+\sin\theta_k\sin\beta\right)]
=\frac{1}{2n}[1-|C_k|\cos(\beta-\theta_k)].
\end{equation}
This reaches the lowest eigenvalue, $E_k=1/(2n)[1-|C_k|]$,
when $\beta=\theta_k$. Therefore, the one-qubit ground
state is $R_Y(\arg C_k)|0\rangle$. Since the Hamiltonian
terms act on separate qubits, their ground states can
be combined into a product ground state.
It is important to note that the empirical phase $\theta_k$ is distinct from the
target Fourier phase $\phi_k=\arg\widehat s_k$.
However, as $C_k\to{\widehat s_k}$, the learned phase
$\theta_k$ approaches $\phi_k$ modulo $2\pi$.

We now show that the HE classifier obtained from the learned
ground state approximates the target labels encoded by $s$.
The full ground state is $
|\psi_{\mathcal T}\rangle
=\bigotimes_{k\in K}
\left(R_Y(\arg C_k)|0\rangle\right)^{\otimes m_k}$, which implements the following prediction score function:
\begin{equation}
g_{\mathcal T}(x)
:=\langle\psi_{\mathcal T}|
U(x)^\dagger A U(x)|\psi_{\mathcal T}\rangle
=\sum_{k\in K}\frac{m_k}{n}
\cos\!\left(k\cdot x/L+\arg C_k\right).
\end{equation}
Indeed, the encoding rotation and the ground-state rotation
combine into $R_Y(k\cdot x/L+\arg C_k)$ on each qubit,
whose $Z$ expectation is the corresponding cosine.
The model predicts label $0$ when $g_{\mathcal T}(x)\geq0$
and label $1$ otherwise.

For fixed $n$, define the limiting score for $N\rightarrow\infty $
\begin{equation}
g_\infty(x)
=\sum_{k\in K}\frac{m_k}{n}
\cos(k\cdot x/L+\phi_k).
\end{equation}
Since each cosine has magnitude at most one,
\begin{equation}
\|Wg_\infty-p\|_\infty
\leq\sum_{k\in K}\left|\frac{Wm_k}{n}-a_k\right|
\leq\frac{W|K|}{n}.
\end{equation}
We can then choose $n$ such that this bound is at most
$\sqrt{\varepsilon}/4$.
Secondly, for a finite dataset, $\arg C_k$ will only be an approximation of $\phi_k=\arg \widehat{ s}_k$. However, from the fact that $C_k\rightarrow\widehat{ s}_k $ as $N\rightarrow \infty$, we have that for any $\delta>0$, we can first choose $n$ such that
$W|K|/n\leq\sqrt{\varepsilon}/4$, and then a finite $N$
such that, with probability at least $1-\delta$,
\begin{equation}
\|Wg_{\mathcal T}-Wg_{\infty}\|_\infty
\leq\sqrt{\varepsilon}/4.
\end{equation}

On this event, we can bound the average approximation error of the HE model with respect the target $s$ as:
\begin{align}
\|s-Wg_{\mathcal T}\|_{L^2(\mu)}
&\leq
\|s-p\|_{L^2(\mu)}
+\|p-Wg_\infty\|_\infty
+\|Wg_\infty-Wg_{\mathcal T}\|_\infty\\
&\leq
\frac{\sqrt{\varepsilon}}2
+\frac{\sqrt{\varepsilon}}4
+\frac{\sqrt{\varepsilon}}4
=\sqrt{\varepsilon}.
\end{align}

Since $W>0$, a mistaken prediction implies
$|s(x)-Wg_{\mathcal T}(x)|\geq1$. Therefore,
$
\Pr_{x\sim\mu}
[\tilde f_{\mathcal T}(x)\neq f(x)]
\leq\|s-Wg_{\mathcal T}\|_{L^2(\mu)}^2
\leq\varepsilon.
$

\end{proof}
The construction guarantees universality with finite,
target-dependent number of qubits $n$ and training points $N$. We note that a polynomial scaling of $n$ and $N$ is not guaranteed for arbitrary
target functions.

\section{Solving a classically hard problem with Hamiltonian embedding}\label{app: DCR}

\subsection{The discrete cube root problem}

We first recall the discrete cube root problem and the corresponding classical hardness assumption. Let $p$ and $q$ be two distinct large prime numbers of the form $3k+2$ and $3k'+2$, respectively, and assume that they are represented by approximately the same number of bits. Let $N=pq$ be their product, which we take to be an $n$-bit integer, and denote by $\mathbb{Z}_N=\{0,\ldots,N-1\}$ the integers modulo $N$.

We consider the discrete cube function
\begin{equation}
    f_N:\mathbb{Z}_N\to\mathbb{Z}_N,
    \qquad
    f_N(y)=y^3 \mod N.
\end{equation}
For primes of the form specified above, this map is a bijection on $\mathbb{Z}_N$, and therefore its inverse is well defined. We denote this inverse by
\begin{equation}\label{eq: dcr function}
    g_N(x)=f_N^{-1}(x)=\sqrt[3]{x}\mod N.
\end{equation}
Indeed, since $p=3k+2$ and $q=3k'+2$, one has $\gcd(3,p-1)=\gcd(3,q-1)=1$. Hence cubing is invertible modulo both $p$ and $q$, and therefore also modulo $N=pq$ by the Chinese remainder theorem.

\begin{definition}[Discrete cube root problem]
Let $p$ and $q$ be two distinct primes of the form $3k+2$ and $3k'+2$, and let $N=pq$ be an $n$-bit integer. Given as input $N$ and $x\in\mathbb{Z}_N$, the discrete cube root problem asks to output $y\in\mathbb{Z}_N$ such that
\begin{equation}
    y^3=x \mod N.
\end{equation}
\end{definition}

The forward function $f_N$ can be computed efficiently on a classical computer using modular exponentiation. By contrast, computing the inverse function $g_N$ from only $N$ and $x$ is believed to be classically hard when the factorization of $N$ is unknown. This is formalized by the following assumption, which is a special case of the RSA assumption with exponent $e=3$.

\begin{definition}[Discrete cube root assumption]
For any polynomial $P(\cdot)$, there is no classical algorithm $A$ running in time $P(n)$ such that, given as input an $n$-bit integer $N=pq$, where $p$ and $q$ are random primes of the form $3k+2$ and $3k'+2$, and a random $x\in\mathbb{Z}_N$, outputs a value $y\in\mathbb{Z}_N$ satisfying
\begin{equation}
    y^3=x \mod N
\end{equation}
with probability at least $1/P(n)$. The probability is taken over the random choice of the primes $p,q$, the input $x\in\mathbb{Z}_N$, and the internal randomness of $A$.
\end{definition}

The discrete cube function is also a trapdoor function. If the factors $p$ and $q$ of $N$ are known, then one can compute Euler's totient $\varphi(N)=(p-1)(q-1)$. Since $p=3k+2$ and $q=3k'+2$, both $p-1$ and $q-1$ are congruent to $1$ modulo $3$, and therefore $\gcd(3,\varphi(N))=1$. It follows that $3$ has a multiplicative inverse modulo $\varphi(N)$. Hence, by Euclid's algorithm, one can efficiently find an integer $d$ such that $3d=1 \mod \varphi(N)$, or equivalently $3d=r\varphi(N)+1$ for some integer $r$.

We now show why this value of $d$ gives a secret key for inverting the discrete cube function. Let $x=f_N(y)=y^3\mod N$ for some $y\in\mathbb{Z}_N$. Then
\begin{equation}\label{eq: dcr_d}
x^d=(y^3)^d=y^{3d}=y^{r\varphi(N)+1}\mod N.
\end{equation}
 Using $\varphi(N)=(p-1)(q-1)$ together with Fermat's little theorem modulo $p$ and modulo $q$, one obtains $y^{r\varphi(N)+1}=y \mod p$ and $y^{r\varphi(N)+1}=y \mod q$. Since $p$ and $q$ are distinct primes, congruence modulo both $p$ and $q$ implies congruence modulo $N=pq$. Therefore, $x^d=y \mod N$. Thus, for every $x\in\mathbb{Z}_N$, the discrete cube root can be computed efficiently as $g_N(x)=x^d \mod N$.

The integer $d$ therefore acts as a trapdoor key: it can be computed efficiently from the factorization of $N$, but is believed to be inaccessible to efficient classical algorithms given only $N$. Hence, knowledge of the factorization makes the inverse efficiently computable, whereas without this trapdoor the problem is assumed to be classically intractable.

\subsection{Definition of the learning problem}

Based on the DCR function in \cref{eq: dcr function}, we define the following function that classifies the input $x\in\mathbb{Z}_N$ of our learning problem.

For a fixed shift parameter $s\in\mathbb Z_N$, define
$g_{N,s}:\mathbb Z_N\to \{0,1\}$,
by
\begin{equation}\label{eq: g_sN}
y(x)=g_{N,s}(x)=
\begin{cases}
1, & \text{if } g_N(x)\in \left[s,\, s+\dfrac{N-1}{2}\right]\\[6pt]
0, & \text{otherwise}
\end{cases}
\end{equation}
where, as in Eq.~\ref{eq: dcr_labeling_function}, the interval
is interpreted cyclically in $\mathbb{Z}_N$.

Equivalently, using the key $d$ as in \cref{eq: dcr_d},
\begin{equation}\label{eq: g_Ns^d}
g_{N,s}^{d}(x)=
\mathbf 1\!\left[x^d \bmod N \in \left[s,\, s+\dfrac{N-1}{2}\right]\right].
\end{equation}

Based on the labeling function above, we define the following learning problem. 
Let $N$ be a large semiprime of the form $N=pq$ as above, we formally define the concept class $\mathcal{C}_N$ as the family of all concepts function $g_{N,s}$ in \cref{eq: g_sN}:
\begin{equation}\label{eq: DCR_concept_class_n}
\mathcal{C}_N
=
\left\{
g_{N,s}
\;\middle|\;
s\in\mathbb{Z}_N
\right\}.
\end{equation}
Then, the learner is given training data labeled by one of the concept in $\mathcal{C}_N$, 
\begin{equation}
\mathcal{T}=\{(x_\ell,y_\ell)\}_{\ell=1}^m,
\qquad x_\ell\in \mathbb Z_N,
\qquad y_\ell=g_{N,s^\star}^{d^\star}(x_\ell)\in\{0,1\},
\end{equation}
for some hidden target parameters $(d^\star,s^\star)$, where $g_{N,s}^d$ denotes the representation of $g_{N,s}$ associated with the secret key $d$. Importantly, the learner does not know the true $(d^\star,s^\star)$, it can only retrieve information about them through the training data. We also assume that the learner has knowledge of $N$ from the specification of the concept class $\mathcal{C}_N$ or from the training data (for example we could trivially append it to all the training points $x_\ell \in\mathbb Z_N$ as $(x_\ell,N,y_\ell)$).
The goal of the learning task is to implement a learning algorithm $A$ that meets the following learning condition.

\begin{definition}[Efficient PAC learning]
\label{def:pac-learning}
Let $\mathcal{C}_N$ be the DCR concept class. We say that $\mathcal{C}_N$ is efficiently classically learnable if there exists a classical learning algorithm $A$ with the following property. For every target concept $g_{N,s}\in\mathcal{C}_N$, every probability distribution $\mathcal{D}$ over $\mathbb{Z}_N$, and every $\varepsilon,\delta>0$, given a training set
$\mathcal{T}=\{(x_\ell,g_{N,s}(x_\ell))\}_{\ell=1}^m$ with $x_\ell\sim\mathcal{D}$ drawn independently, the algorithm outputs a hypothesis $h:\mathbb{Z}_N\to\{0,1\}$ such that
\begin{equation}\label{eq: learning_condition}
    \Pr_{x\sim\mathcal{D}}
    \left[
        h(x)\neq g_{N,s}(x)
    \right]
    \leq
    \varepsilon
\end{equation}
with probability at least $1-\delta$ over the random draw of the training set and the internal randomness of $A$. Moreover, $A$ runs in time polynomial in $n$, $1/\varepsilon$, and $1/\delta$.
\end{definition}

\subsection{Classical hardness of the learning problem}

The classical hardness result for the DCR learning task relies on the result of~\cite{alexi1988rsa}. In particular, they show that even a weak predictor for suitable bits of the discrete cube root is sufficient to recover the full inverse function. Based on this, \cite{jerbi2024shadows} showed that no efficient classical learner can achieve accuracy significantly better than random guessing on the concept class defined from the discrete cube root function. We first recall the notion of an $\epsilon(n)$-oracle.

\begin{definition}[$\epsilon(n)$-oracle]
Let $O_{N,s}$ be a probabilistic oracle such that, on input $x\in\mathbb{Z}_N$, it outputs the label $g_{N,s}(x)$ correctly with probability $1/2+\epsilon(n)$ over the random choice of $x$ and the internal randomness of the oracle. We say that $O_{N,s}$ is an $\epsilon(n)$-oracle for $g_{N,s}$.
\end{definition}

The result of Alexi et al.~\cite{alexi1988rsa} implies that a weak oracle with inverse-polynomial advantage over random guessing is already sufficient to break the discrete cube root assumption.

\begin{lemma}[Corollary (a) to Theorem 1 in~\cite{alexi1988rsa}]
\label{lem:weak-oracle-dcr}
For any concept $g_{N,s}$ in the DCR concept class, given access to a $1/\operatorname{poly}(n)$-oracle $O_{N,s}$ for $g_{N,s}$, there exists a polynomial-time classical algorithm that uses $O_{N,s}$ to compute the discrete cube root $g_N(x)=\sqrt[3]{x}\mod N$ for every $x\in\mathbb{Z}_N$ with constant success probability.
\end{lemma}

Using this result, the authors in~\cite{jerbi2024shadows} obtain the following classical hardness statement for the corresponding learning problem.

\begin{lemma}[Classical hardness of the DCR learning task]
\label{lem:dcr-learning-hardness}
Assume the discrete cube root assumption. Then no polynomial-time classical learning algorithm can, with constant success probability over the random generation of the training data and its internal randomness, output a hypothesis $h$ whose expected error on the DCR concept $g_{N,s}$ is smaller than random guessing by an inverse-polynomial amount. More precisely, for the uniform distribution over $\mathbb{Z}_N$, no efficient classical learner can achieve
\begin{equation}
    \Pr_{x\sim U(\mathbb{Z}_N)}
    \left[
        h(x)\neq g_{N,s}(x)
    \right]
    \leq
    \frac{1}{2}
    -
    \frac{1}{\operatorname{poly}(n)} .
\end{equation}
\end{lemma}

Indeed, if a classical learner satisfying the learning condition in Definition~\ref{def:pac-learning} existed, then the hypothesis it outputs would define a $1/\operatorname{poly}(n)$-oracle for the label function $g_{N,s}$. By Lemma~\ref{lem:weak-oracle-dcr}, this oracle could then be used to compute the full discrete cube root function $g_N(x)$ in polynomial time with constant success probability. This would contradict the discrete cube root assumption. Therefore, under this assumption, the DCR learning task cannot be solved efficiently by classical learning algorithms with non-negligible advantage over random guessing.

\subsection{Quantum solution with Hamiltonian embedding}
 We now show how this problem can be encoded within the Hamiltonian-embedding framework. In particular, we construct an embedding Hamiltonian of the form $H_T = \sum_\ell H(x_\ell)$ whose ground state can be used to correctly classify new inputs, in the same way as described in Section~\ref{sec: ham embedding model}. This allows us to show in the following lemma the quantum learnability of the DCR task within the Hamiltonian embedding model. Together with Lemma~\ref{lem:dcr-learning-hardness}, this formally establishes the result anticipated in the main text as Theorem~\ref{thm: dcr advantage}.

\begin{lemma}\label{lemma: dcr quantum learnability}
Let $C_N$ be the DCR concept class, and let $\mathcal T=\{(x_\ell,y_\ell)\}_{\ell=1}^m$ be the training data associated to a concept $g_{N,s^\star}^{d^\star}\in C_N$ with $x_\ell$ sampled independently from the uniform distribution $\mathcal U(\mathbb Z_N)$. For every $\epsilon,\delta>0$, there exists
$m=\operatorname{poly}(n,1/\epsilon,\log(1/\delta))$, an encoding unitary
$U(x)$, and a local label observable $O(y)$, such that, with probability at
least $1-\delta$ over the draw of $\mathcal T$, the ground state of the
corresponding embedding Hamiltonian $H_{\mathcal T}$ induces a hypothesis
$\tilde f:\mathbb Z_N\to\{0,1\}$ as in \cref{eq: prediction rule} satisfying
\begin{equation}
\Pr_{x\sim\mathcal U(\mathbb Z_N)}
\!\left[\tilde f(x)\neq g_{N,s^\star}^{d^\star}(x)\right]
\leq \epsilon .
\end{equation}
Moreover, such a ground state can be prepared efficiently by a quantum algorithm, and once obtained, the induced hypothesis $\tilde f$ can be efficiently evaluated classically.
\end{lemma}

\begin{proof}
The embedded Hamiltonian is defined on the Hilbert space
$
\mathcal H
=
\mathrm{span}\{|d,s\rangle : d,s\in \mathbb Z_N\}\otimes \mathbb C^2$,
where the last qubit serves as an ancilla with computational basis $\{|0\rangle,|1\rangle\}$.

For each input $x\in \mathbb{Z}_N$, define the unitary
\begin{equation}\label{eq: encoding_unitary}
U(x)\,|d,s\rangle|b\rangle
=
|d,s\rangle\,|b\oplus g_{N,s}^{d}(x)\rangle .
\end{equation}
We note that $g_{N,s}^{d}(x)$ can be evaluated efficiently classically, since its computation only requires modular exponentiation and a classical comparison. Hence, there exists a polynomial-size classical circuit computing $g_{N,s}^{d}(x)$ from $x$, $d$, and $s$. By reversibly implementing this computation, with at most polynomial overhead, one obtains a polynomial-size quantum circuit realizing the unitary $U(x)$.
Define now the observable
\begin{equation}\label{eq: observable}
O(y)=I\otimes |y\rangle\langle y|,
\qquad y\in\{0,1\}.
\end{equation}

We can now define the embedding Hamiltonian $H_\mathcal{T}$ for the training set $\mathcal{T}$. For one training example $(x_\ell,y_\ell)$, define the Hamiltonian term
\begin{equation}\label{eq: Hxy}
H(x_\ell,y_\ell)
=
I-U(x_\ell)^\dagger O(y_\ell)U(x_\ell)+I\otimes |1\rangle\langle 1|.
\end{equation}
The dataset Hamiltonian $H_\mathcal{T}$ is then:
\begin{equation}\label{eq: total H}
H_\mathcal{T}=\sum_{\ell=1}^m H(x_\ell,y_\ell).
\end{equation}
We now show that the ground state of $H_\mathcal{T}$ can be used to correctly predict the labels of unseen inputs. In particular, in analogy to \cref{eq: prediction rule}, for a new point $x$ the predicted label is the value of $y$ that maximizes \\ $\operatorname{Tr}\!\left[
O(y)\,
U(x)\,
|d^\star,s^\star,0\rangle\langle d^\star,s^\star,0|
\,U(x)^\dagger
\right]$, with $|d^\star,s^\star,0\rangle$ ground state of $H_\mathcal{T}$.  

We first show that the energy of each term $H(x_\ell,y_\ell)$ is minimized by states of the form $|d,s,0\rangle$ with $(d,s)$ such that $g_{N,s}^{d}(x_\ell)=y_\ell$. Then, since the training labels are generated by the target concept $y_\ell=g_{N,s^\star}^{d^\star}(x_\ell)$, the state $|d^\star,s^\star,0\rangle$ associated with the true parameters $d^\star$ and $s^\star$ minimizes every term in $H_{\mathcal T}$ and thus has zero energy.

Let us consider each term $H(x_\ell,y_\ell)$ in \cref{eq: Hxy}.  We can expand $U(x_\ell)^\dagger O(y_\ell)U(x_\ell)$ as follows:
\begin{equation}\label{eq: UOU}
U(x_\ell)^\dagger O(y_\ell)U(x_\ell)
=
\sum_{(d,s):\, g_{N,s}^{d}(x_\ell)=y_\ell}
|d,s\rangle\langle d,s|\otimes |0\rangle\langle 0|
+
\sum_{(d,s):\, g_{N,s}^{d}(x_\ell)\neq y_\ell}
|d,s\rangle\langle d,s|\otimes |1\rangle\langle 1|.
\end{equation}
\cref{eq: UOU} follows from computing the action of $U(x_\ell)^\dagger O(y_\ell)U(x_\ell)$ on generic states $\ket{d,s,b}$ using the definitions of $U(x)$ and $O(y)$ in \cref{eq: encoding_unitary,eq: observable}:
\begin{equation}
    U^\dagger(x_\ell) O(y_\ell)U(x_\ell)\ket{d,s,b}=\delta_{y_\ell,b\oplus g^{d}_{N,s}(x_\ell)}\ket{d,s,b}
\end{equation}

Therefore, by rewriting the Hamiltonian $H(x_\ell,y_\ell)$ in \cref{eq: Hxy} using \cref{eq: UOU} we obtain: 
\begin{equation}
\label{eq: H(x)}
H(x_\ell,y_\ell)
=
I
-
\sum_{(d,s):\, g_{N,s}^{d}(x_\ell)=y_\ell}
|d,s\rangle\langle d,s|\otimes |0\rangle\langle 0| \;
-
\sum_{(d,s):\, g_{N,s}^{d}(x_\ell)\neq y_\ell}
|d,s\rangle\langle d,s|\otimes |1\rangle\langle 1|
\;\;+
I\otimes |1\rangle\langle 1|
    \end{equation}

We now evaluate the energy of a generic basis state $\ket{d,s,b}$ to determine which states minimize the Hamiltonian terms $H(x_\ell,y_\ell)$.
\begin{itemize}
    \item 
If $b=1$, i.e. $|\psi\rangle=|d,s,1\rangle$, then
\begin{align}
\langle \psi|H(x_\ell,y_\ell)|\psi\rangle
=
\begin{cases}
2, & \text{if } g_{N,s}^{d}(x_\ell)=y_\ell,\\
1, & \text{if } g_{N,s}^{d}(x_\ell)\neq y_\ell.
\end{cases}
\end{align}

\item If $b=0$, i.e. $|\psi\rangle=|d,s,0\rangle$, then
\begin{align}
\langle \psi|H(x_\ell,y_\ell)|\psi\rangle
=
\begin{cases}
0, & \text{if } g_{N,s}^{d}(x_\ell)=y_\ell,\\
1, & \text{if } g_{N,s}^{d}(x_\ell)\neq y_\ell.
\end{cases}
\end{align}
\end{itemize}
Hence, each term $H(x_\ell,y_\ell)$ in the dataset Hamiltonian $H_{\mathcal T}$ is minimized precisely by states of the form $\ket{d,s,0}$ satisfying $g_{N,s}^{d}(x_\ell)=y_\ell$.

Since the total Hamiltonian is given by $ H_\mathcal{T}=\sum_{\ell=1}^m H(x_\ell,y_\ell)$, then for a generic basis state $|d,s,0\rangle$ one has
\begin{equation}
\langle d,s,0|H_\mathcal{T}|d,s,0\rangle
=
\sum_{\ell=1}^m
\mathbf 1\!\left[g_{N,s}^{d}(x_\ell)\neq y_\ell\right].
\end{equation}
Therefore, the energy of a state $\ket{d,s,0}$ is exactly the number of training examples misclassified by the hypothesis specified by $(d,s)$:
\begin{equation}\label{eq: dcr gound energy}
\langle d,s,0|H_T|d,s,0\rangle
=
\abs{\{\text{training points misclassified by }g_{N,s}^{d}\}}.
\end{equation}
It follows that a basis state $\lvert d,s,0\rangle$ has zero energy if and only if the corresponding hypothesis $g^d_{N,s}$ correctly classifies every training example. Since the target hypothesis $g^{d^\star}_{N,s^\star}$ is consistent with the training data, the ground energy of $H_{\mathcal T}$ is zero, and its ground space is spanned by the states $\lvert d,s,0\rangle$ corresponding to hypotheses that are consistent with the training set. 

We note that the exponent $d^\star$ can be recovered efficiently by a quantum algorithm by factoring $N$~\cite{shor1997polynomial}. Once $d^\star$ is known, the remaining parameter $s$ can be learned efficiently from the training set by first computing the DCR of each input:
$
x_\ell \longmapsto x_\ell^{d^\star}\!\!\!\mod N
=
\sqrt[3]{x_\ell}\!\!\!\mod N
=
g_N(x_\ell)$, with $g_N$ the DCR function defined in Eq.~\eqref{eq: dcr function}.
The transformed dataset
$
\widetilde{\mathcal T}
=
\{(g_N(x_\ell),y_\ell)\}_{\ell=1}^m$
then reduces the problem to learning the shift of a one-dimensional interval classifier. Since the data are realizable, a shift $s'$ that is consistent with all the training examples can be found efficiently. While several shifts may be consistent with a polynomial-size training set, sufficiently many training samples guarantee that a suitable consistent shift $s'$ has small population error under the uniform distribution $\mathcal U(\mathbb Z_N)$.  In particular, \cite{jerbi2024shadows} shows that, with $O(\log(1/\delta)/{\epsilon'})$ uniformly sampled examples, one can identify a shift $s'$ satisfying
$
|s^\star-s'|\leq \epsilon' N$
with probability at least $1-\delta$, where $|s^\star-s'|$ denotes the cyclic distance between the two shifts in $\mathbb{Z}_N$.
Let us fix $\epsilon'=\epsilon/2$. This in turn guarantees
\begin{equation}\label{eq: learning_guarantee_dcr}
\Pr_{x\sim\mathcal U(\mathbb Z_N)}
\!\left[g_{N,s'}(x)\neq g_{N,s^\star}(x)\right]
\leq \epsilon.
\end{equation}
Therefore, the quantum learning algorithm finds the state $\ket{d^\star,s',0}$ which has zero energy with respect to the data Hamiltonian $H_\mathcal{T}$.

Once a ground state $|d^\star,s',0\rangle$ is obtained, a new point $x$ is classified by the same labeling rule as in \cref{eq: prediction rule}. In particular, the probability associated to the label $y$ for an input $x$ is given by
\begin{equation}\label{eq: prediction}
P_{d^\star,s'}(x,y)
=
\operatorname{Tr}\!\left[
O(y)\,
U(x)\,
|d^\star,s',0\rangle\langle d^\star,s',0|
\,U(x)^\dagger
\right].
\end{equation}

Since, by the construction of $U(x)$ in Eq.~\eqref{eq: encoding_unitary},
$
U(x)|d^\star,s',0\rangle
=
|d^\star,s',g_{N,s'}^{d^\star}(x)\rangle$,
it follows that
\begin{equation}
P_{d^\star,s'}(x,y)=
\begin{cases}
1, & \text{if } g_{N,s'}^{d^\star}(x)=y,\\
0, & \text{if } g_{N,s'}^{d^\star}(x)\neq y.
\end{cases}
\end{equation}
Hence, the prediction rule induced by the ground state
$|d^\star,s',0\rangle$ is precisely the hypothesis
$g_{N,s'}^{d^\star}$. By Eq.~\eqref{eq: learning_guarantee_dcr},
this hypothesis has population error at most $\epsilon$ under
$\mathcal U(\mathbb Z_N)$, which proves the lemma.  

Finally, we note that, once the state $|d^\star,s',0\rangle$ has been obtained quantumly, the prediction probability $P_{d^\star,s'}(x,y)$ in \cref{eq: prediction} can be evaluated entirely classically, since $U(x)$ acting on $|d^\star,s',0\rangle$ is efficiently computable by directly evaluating $g^{d^\star}_{N,s'}(x)$.
\end{proof}

In Lemma~\ref{lemma: dcr quantum learnability}, we constructed an embedding Hamiltonian $H_{\mathcal{T}}$ whose ground space contains a state $\ket{d^\star,s',0}$ associated with a hypothesis $g_{N,s'}^{d^\star}$ that solves the DCR learning task, and we showed that such a state can be found efficiently by a quantum algorithm. In the following lemma, we show that, more generally, every ground state of $H_{\mathcal{T}}$ induces a hypothesis with small population error on the DCR task.

\begin{lemma}
\label{lem:dcr-ground-state-generalization}
Let
$
\mathcal{H}_N
=
\left\{
g_{N,s}^{d}: d,s\in\mathbb{Z}_N
\right\},
$ be the hypothesis class of functions $g_{N,s}^{d}$ as in Eq.~\eqref{eq: g_Ns^d}.
Given the training set
$
\mathcal{T}
=
\{(x_\ell,y_\ell)\}_{\ell=1}^{m},
$  with  $
y_\ell
=
g_{N,s^\star}^{d^\star}(x_\ell)$,
 generated by drawing the inputs independently from the uniform
distribution over $\mathbb{Z}_N$, let $H_\mathcal{T}$ be the embedding Hamiltonian as in Eq.~\eqref{eq: total H}.

For any $\epsilon,\delta>0$, if
\[
m
\ge
\frac{1}{\epsilon}
\left(
2\log N+\log\frac{1}{\delta}
\right),
\]
then, with probability at least $1-\delta$ over the draw of the training set,
every zero-energy basis state $\lvert d,s,0\rangle$ of
$H_{\mathcal T}$ induces a hypothesis $g_{N,s}^{d}$ satisfying
\[
\Pr_{x\sim U(\mathbb{Z}_N)}
\left[
g_{N,s}^{d}(x)
\neq
g_{N,s^\star}^{d^\star}(x)
\right]
\le
\epsilon.
\]
In particular, the conclusion holds independently of whether
$d=d^\star$.
\end{lemma}

\begin{proof}
By Eq.~\eqref{eq: dcr gound energy}, for every basis state $\lvert d,s,0\rangle$,
\begin{equation}
\langle d,s,0|
H_{\mathcal T}
|d,s,0\rangle
=
\sum_{\ell=1}^{m}
\mathbf{1}
\left[
g_{N,s}^{d}(x_\ell)\neq y_\ell
\right].
\end{equation}
Since the target hypothesis $g_{N,s^\star}^{d^\star}$ is consistent with
the training data, the ground energy of $H_{\mathcal T}$ is zero.
Therefore, any zero-energy basis state $\lvert d,s,0\rangle$
corresponds to a hypothesis $g_{N,s}^{d}$ that is consistent with all
training examples.

Now consider a fixed hypothesis $h\in\mathcal{H}_N$ whose population
error with respect to the target hypothesis is larger than $\epsilon$,
namely
\begin{equation}
R(h)
:=
\Pr_{x\sim U(\mathbb{Z}_N)}
\left[
h(x)\neq g_{N,s^\star}^{d^\star}(x)
\right]
>
\epsilon.
\end{equation}
The probability that $h$ nevertheless agrees with all $m$
independently drawn training examples is
\begin{equation}
(1-R(h))^m
\le
(1-\epsilon)^m
\le
e^{-m\epsilon}.
\end{equation}
Since $d,s\in\mathbb{Z}_N$, the hypothesis class satisfies $
|\mathcal{H}_N|\le N^2$.

Hence, by a union bound,
\begin{equation}
\Pr_{\mathcal T}
\left[
\exists h\in\mathcal{H}_N:
R(h)>\epsilon
\ \text{and}\
\widehat R_{\mathcal T}(h)=0
\right]
\le
N^2 e^{-m\epsilon}.
\end{equation}
If
\begin{equation}
m
\ge
\frac{1}{\epsilon}
\left(
2\log N+\log\frac{1}{\delta}
\right),
\end{equation}
then the right-hand side is at most $\delta$.
Therefore, with probability at least $1-\delta$, every hypothesis in
$\mathcal{H}_N$ that is consistent with the training set has population
error at most $\epsilon$.

Since every zero-energy basis state of $H_{\mathcal T}$ corresponds to
such a consistent hypothesis, the claimed result follows.
\end{proof}

We note that the same argument can be extended to arbitrary ground states that are superpositions of zero-energy basis states. Indeed, for a ground state
$\ket{\psi}=\sum_{d,s} c_{d,s}\ket{d,s,0}$, the prediction probabilities are convex combinations of those associated with the individual hypotheses $g_{N,s}^d$, since the states $\ket{d,s}$ are mutually orthogonal. Consequently, if every zero-energy basis hypothesis has population error at most $\epsilon/2$, the classifier induced by $\ket{\psi}$ has population error at most $\epsilon$, since a misclassification can occur only when at least half of the total weight is assigned to basis hypotheses that err on that input.

We note that this example establishes the existence of a hard task that can be solved within the Hamiltonian embedding framework, for which a quantum algorithm can efficiently prepare the ground state of the embedded Hamiltonian, while doing so classically is believed to be hard. This is sufficient to demonstrate an advantage for that specific task. For more general learning problems, however, we do not expect to be able to design a problem-specific algorithm for preparing the ground state of the embedded Hamiltonian. Instead, one would likely rely on more generic quantum ground-state preparation methods, such as QPE or SKQD.

Finally, for the ground state $\ket{d^\star,s',0}$ constructed by the quantum learning algorithm, inference can be performed entirely classically. Indeed, once the parameters $d^\star$ and $s'$ are known, the corresponding hypothesis $g_{N,s'}^{d^\star}(x)$ can be evaluated efficiently using classical modular exponentiation and a classical comparison. This differs from the settings considered in Section~\ref{sec: classical inference}, where the classical surrogates are constructed using task-agnostic procedures. 

\section{A construction with highly entangled ground states}\label{sec: example_entanglement}
In Section~\ref{sec: observable O}, we discussed that the Hamiltonian embedding framework allows for different choices of the Hermitian operator $A$ entering the observable $O(y)$ in Eq.~\eqref{eq: observable O}. The choice of the label observable $O(y)$ can significantly affect the performance of the model, and different operators may be more suitable depending on the learning task under consideration. Although identifying which choices of $O(y)$ are best suited to a given learning problem is itself an interesting question, in this Appendix and Appendix~\ref{app: provable sep surrogate} we show that the choice of $O(y)$ can influence properties of the ground state $\psi_{{\scriptscriptstyle\mathcal{T}},\alpha^\star}$ found by the Hamiltonian embedding model, such as its degree of entanglement. In this Appendix, we provide a simple example of a label dependent observable $O(y)$ for which the ground states obtained within the Hamiltonian embedding framework exhibit a high degree of entanglement. We note that, in this particular construction, the resulting states remain efficiently classically representable, as they are derived from Clifford states. More elaborate choices of $O(y)$ could in principle give rise to ground states with more complex structure. The purpose of this example is therefore not to establish classical hardness results as we did in Section~\ref{sec: hard problem}, but rather to illustrate how the choice of $O(y)$ can directly influence the properties of the ground state found by the model. 

\subsection{Entangled graph states as ground states}
In order to obtain entangled ground states, we start from the parent Hamiltonian of a graph state, which can exhibit a large amount of bipartite entanglement. Specifically, for $n$ qubits, we consider a graph composed of two cycles of $n/2$ vertices each, connected by a perfect matching between corresponding vertices. To each vertex $a$ we associate the graph-state stabilizer
\begin{equation}
K_a = X_a \prod_{q\in N(a)} Z_q ,
\end{equation}
where $N(a)$ denotes the set of vertices adjacent to $a$. The corresponding parent Hamiltonian is constructed by introducing one local check for each qubit $
h_a=\frac{1}{2}\left(I-K_a\right),$ and defining $
H=\frac{1}{n}\sum_{a=1}^{n} h_a .
$

Since $K_a$ has eigenvalues $\pm1$, each check contributes energy $0$ when $K_a=+1$ and energy $1$ when $K_a=-1$. The ground state is therefore the unique simultaneous $+1$ eigenstate of all the stabilizers $K_a$, namely the graph state associated with the chosen graph.

In order to parameterize this construction with respect to the label $y$, we associate the two classes $y=0$ and $y=1$ with two different stabilizer syndromes. In particular, we choose the two syndromes $
\beta^{(0)}=(0,\ldots,0)
$, $
\beta^{(1)}=(1,0,1,0,\ldots,1,0)
$ with both $\beta^{(0)}, \beta^{(1)}\in\{0,1\}^n$
and define the label-dependent checks as

\begin{equation}
h_a(y)
=
\frac{1}{2}
\left[
I-(-1)^{\beta_a^{(y)}}K_a
\right],
\end{equation}

We can define then the label-dependent observable as the parent Hamiltonian of the graph states with the signature $\beta^{(0)}$ or $\beta^{(1)}$ according to the label $y$:

\begin{equation}\label{eq: label_observable}
O(y)
=
\frac{1}{n}\sum_{a=1}^{n}h_a(y)
=
\frac{1}{2}
\left[
I-\frac{1}{n}\sum_{a=1}^{n}(-1)^{\beta_a^{(y)}}K_a
\right].
\end{equation}

For $y=0$, the ground state satisfies $K_a=+1$ for every qubit $a$, whereas for $y=1$ it satisfies $K_a=(-1)^{\beta_a^{(1)}}$, corresponding to the alternating stabilizer syndrome.

To complete the construction of the Hamiltonian embedding model, we now dress the observable $O(y)$ with a particular choice of encoding unitary $U(x,\alpha)$. In particular, for $n$ qubits we choose the simple unitary 
\begin{equation}
    U(x,\alpha)
=
\bigotimes_{a=1}^{n}
R_Z\!\left(\alpha_a^Z x_a\right)
R_Y\!\left(\alpha_a^Y x_a\right).
\end{equation}
with $\alpha=\{\alpha_a^Y,\alpha_a^Z\}_{a=1}^{n}\in\mathbb{R}^{2n}$ denoting the set of trainable parameters.
We notice that there are no entangling gates in $U(x,\alpha)$. This, however, does not preclude the embedding Hamiltonian from having entangled ground states, as we demonstrate below.
For a datapoint $(x,y)$ the encoding Hamiltonian is then
\begin{equation}\label{eq: Hxy graph}
H_{x,y}(\alpha)
=
\frac{1}{n}
\sum_{a=1}^{n}
\frac{
I-(-1)^{\beta^{(y)}_a}\widetilde K_a(x,\alpha)
}{2},
\end{equation}
where
$
\widetilde K_a(x,\alpha)
=
U(x,\alpha)^\dagger K_a U(x,\alpha)$. For a training set
$
\mathcal T=\{(x_i,y_i)\}_{i=1}^N$,
the total Hamiltonian is
$
H_{\mathcal T}(\alpha)
=
\frac{1}{N}
\sum_{i=1}^N
H_{x_i,y_i}(\alpha)$.

\subsection{Entanglement properties of the ground state}

Our numerical experiments show that, for this construction, the ground states learned by the model exhibit a high degree of entanglement. This can be understood directly from the structure of the label Hamiltonians.
The ground state $\ket{G_{\beta^{(y)}}}$ of the labeled observable $O(y)$ in Eq.~\eqref{eq: label_observable} is the simultaneous eigenstate satisfying
$
    K_a\ket{G_{\beta}}
    =
    (-1)^{\beta_a}\ket{G_{\beta}}$, $\forall a$.
For $\beta^{(0)}$, this is the graph state
\begin{equation}\label{eq: graph_0}
    \ket{G_0}
    =
    \prod_{(i,j)\in E}CZ_{ij}\ket{+}^{\otimes n},
\end{equation}
while for $\beta^{(1)}$ we have
\begin{equation}
    \ket{G_1}
    =
    \left(\prod_{a:\beta_a^{(1)}=1} Z_a\right)\ket{G_0}.
\end{equation}

The embedding Hamiltonian associated with the full training set is
\begin{equation}\label{eq: data_hamiltonian_graph}
    H_{\mathcal T}(\alpha)
    =
    \frac{1}{|\mathcal T|}
    \sum_{(x_\ell,y_\ell)\in\mathcal T}
    H_{(x_\ell,y_\ell)}.
\end{equation}
Its ground state $\ket{\psi_{{\scriptscriptstyle\mathcal{T}},\alpha}^\star}$ is in general determined by the competition between the different training terms and therefore does not coincide with the ground state of any individual term.

The case of zero energy loss is achieved when the learned parameters $\alpha^\star$ and ground state $\ket{\psi_{{\scriptscriptstyle\mathcal{T}},\alpha^\star}}$ satisfy
\begin{equation}
    U(x_\ell,\alpha^\star)
    \ket{\psi_{{\scriptscriptstyle\mathcal{T}},\alpha^\star}}
    = \ket{G_{y_\ell}},
\end{equation}
for every training example $(x_\ell,y_\ell)$.
Since $U(x,\alpha)$ is a product of single-qubit rotations, it is a local unitary and therefore preserves the bipartite entanglement spectrum. Hence, we can expect that the learned ground state $\ket{\psi_{{\scriptscriptstyle\mathcal{T}},\alpha^\star}}$ follows the same bipartite entanglement as the corresponding graph states $\ket{G_0}$ and $\ket{G_1}$.

\subsection{Numerical demonstration}
To study the entanglement properties of the ground states reached by the constructed HE model, we consider the following constructed classification task. As input data $x$, we use images from the Fashion-MNIST dataset and apply PCA to reduce their dimensionality to $n$, matching the number of qubits in the model. Rather than using the original Fashion-MNIST class labels, we assign labels through a planted teacher model. Specifically, we fix a highly entangled graph-state ground state $\ket{\psi_{\mathrm{T}}}$ together with a set of teacher parameters $\alpha_{\mathrm{T}}$, and we label each input $x$ according to the class $y\in\{0,1\}$ whose corresponding Hamiltonian term $H_{x,y}(\alpha_T)$ in Eq.~\eqref{eq: Hxy graph} has the lower energy, i.e.
\begin{equation}
 y(x) = \arg\min_{y\in\{0,1\}} \bra{\psi_T} H_{x,y}(\alpha_T) \ket{\psi_T}
 \end{equation}

We note that, in this Appendix, $O(y)$ is defined as the label-dependent parent Hamiltonian, so that the state associated with the correct label minimizes its expectation value and prediction is performed via an $\arg\min_y$ rule. This differs only by convention from the main text, where $O(y)$ denotes the complementary score observable and prediction is therefore performed via an $\arg\max_y$ rule. Moreover, the operators used here do not satisfy $O(0)+O(1)= I$. Prediction selects the label with the lower expectation value, without interpreting these values as probabilities.
The resulting labeled dataset is then used to train a Hamiltonian-embedding student model of the same form, without providing it with either the teacher state or the teacher parameters. 
We notice that the purpose of this experiment is to study the entanglement properties of ground states that can be learned within the Hamiltonian-embedding
model, rather than the ability of the outer optimization to locate such
solutions from arbitrary initializations. We therefore initialize the
student parameters in the neighborhood of the planted teacher,
$\alpha^{(0)}=\alpha_T+\epsilon$, with
$\epsilon_j\sim\mathcal{N}(0,0.25^2)$.
In particular, we choose the teacher parameters $\alpha_T$ by sampling each entry independently from a uniform distribution over $[-\pi,\pi]$, and take as the teacher state $\ket{\psi_T}$ a highly entangled graph-syndrome state,
\begin{equation}
    \ket{\psi_T}=Z^{\beta_T}\ket{G_0},
\end{equation}
where $\ket{G_0}$ is the graph state introduced in Eq.~\eqref{eq: graph_0} and $\beta_T\in\{0,1\}^n$ is a fixed nontrivial stabilizer syndrome for each system size.

We use $N=800$ training samples, together with $200$ validation samples and $400$ test samples for each run.
For the outer parameter optimization, we use the energy-based analogue of the hinge loss introduced in Eq.~\eqref{eq: hinge loss}. Since classification is performed by selecting the label with the lowest energy, we define the class-energy margin as
$ \Delta_\alpha(x,y)=E_{1-y}(x)-E_y(x), $
and optimize
$ \ell_{\mathrm{hinge}}^{(\gamma)}(x,y) =\max\{0,\gamma-\Delta_\alpha(x,y)\} $, with $\gamma=0.05$.
At each optimization step, the ground state of the corresponding training Hamiltonian is recomputed exactly and then kept fixed while updating the encoding parameters.
After training, we study how accurately the final exact learned ground
state can be represented by an MPS of bounded bond dimension. We fix the
qubit ordering used in the simulation and convert the exact state into
an MPS, truncating its Schmidt decompositions so that the maximum bond
dimension is bounded by $\chi$. For each value of $\chi$, we evaluate
both the fidelity of the resulting truncated MPS with the exact learned
state and the test accuracy obtained when this MPS is used in place of
the exact state.

As shown in \cref{fig:mps-bond-scaling}, low-bond-dimension MPS approximations in this fixed ordering have low fidelity with the learned ground state, and
the corresponding classification accuracy is substantially reduced.
Increasing $\chi$ progressively recovers both the exact state and the classification performance.

\begin{figure}[t]
    \centering
    \includegraphics[width=0.9\linewidth]{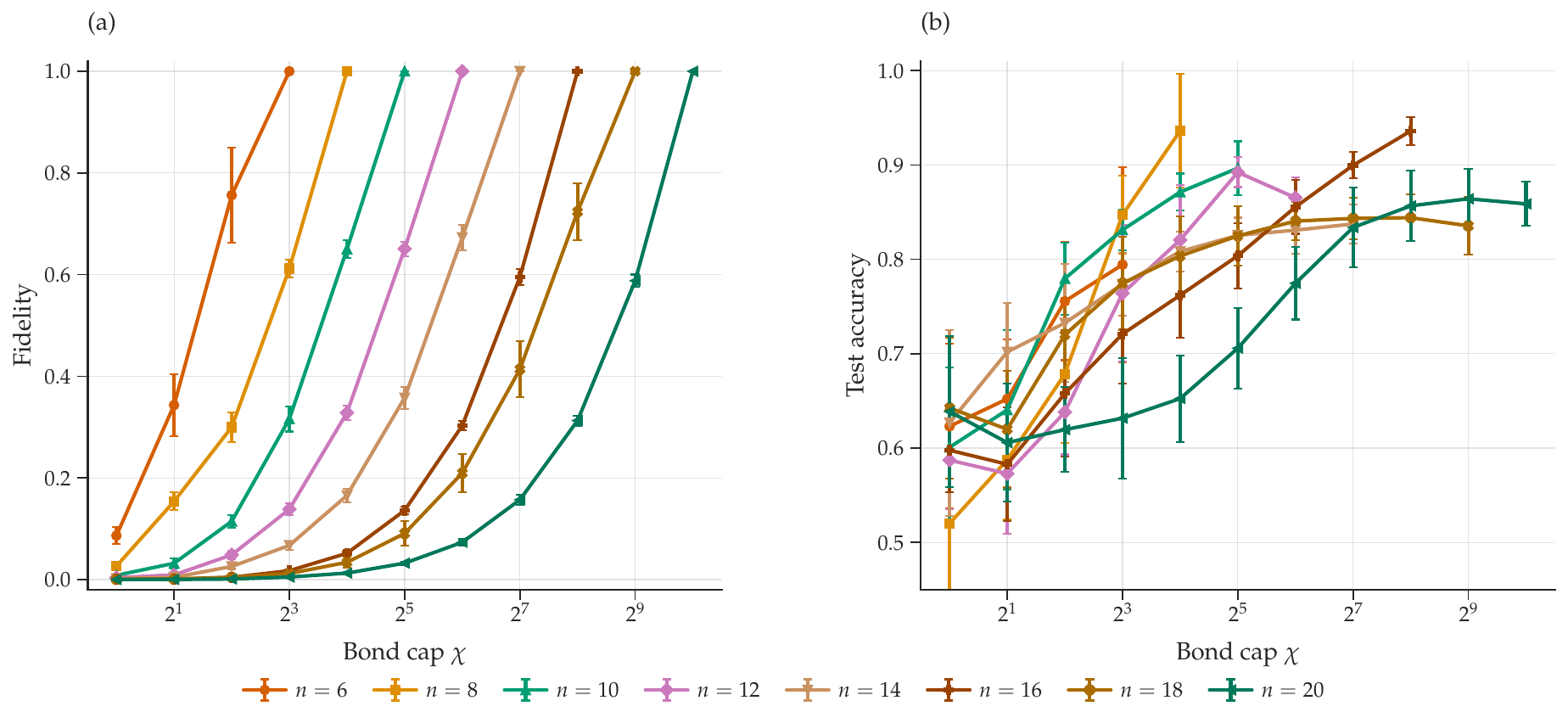}
    \caption{
    Left: fidelity between the final exact learned ground state and its MPS
approximation obtained by Schmidt truncation at maximum bond dimension
$\chi$, using the fixed qubit ordering of the simulation.
Right: test accuracy obtained by replacing the exact learned state with
the corresponding truncated MPS state.
    }
    \label{fig:mps-bond-scaling}
\end{figure}

\section{A provable separation between Pauli surrogates and  Hamiltonian embedding}\label{app: provable sep surrogate} 

In Section~\ref{sec: classical inference}, we showed that, for certain choices of the encoding unitary $U(x,\alpha)$, the trained Hamiltonian embedding model can be evaluated classically once a suitable classical surrogate has been constructed from the final state $\psi_{{\scriptscriptstyle\mathcal{T}},\alpha^\star}$. This naturally raises the question of whether such a surrogate could instead be trained directly, thereby removing the need to obtain the quantum state altogether.
In this Appendix, we present an example in which we analytically show that the physical constraints imposed by the Hamiltonian embedding model lead to better predictions on unseen data than the corresponding Pauli classical surrogate. Moreover, the efficiency guarantee for the classical surrogate
based on reduced density matrices introduced in
Eq.~\eqref{eq: density surrogate} does not extend to this construction.
In order to prove a separation between the two models, we use properties of stabilizer codes and stabilizer states.

Similar to Appendix~\ref{sec: example_entanglement}, we consider a stabilizer code with stabilizers $K_a$, $a\in [n]$.
Concretely, in our example, we have a graph $G$ with stabilizers of the form $K_a = X_a \prod_{q \in N_G(a)}Z_a$, where $N_G(a)$ is the set of neighboring vertices of $a$ in the graph.
This graph and the set of stabilizers is known and specifies the ``concept class'' of the labeling functions.
In particular, each concept is associated with a distinct assignment specifying whether each $K_a$ acts as a stabilizer or an antistabilizer of the ground state. Thus, there are $2^n$ possible functions for a given graph, indexed by the syndrome vector $\beta$ whose coordinates $\beta_a$ are +1 (-1) depending on whether $K_a$ is a stabilizer (antistabilizer).

We now describe the function more explicitly.
Let the data points have coordinates $x_i = \abs{x}_i \zeta_i$ with $\abs{x}_i \in [0,1]$ and $\zeta_i = \text{sign}(x_i)$.
For an input vector $ x$, we form a query by considering which coordinates $i$ have a magnitude above some cutoff $\xi$.
Define this to be the ``active set'' $S = \{i: \abs{x_i}\geq \xi\}$.
We imagine the active set $S$ to be a query to the function since the function depends only on the signs $\{\zeta_i, i \in S\}$ (and not on $\abs{x_i}$).
For each such input, construct an ``error vector'' $W_\zeta := \prod_{i \in S:\zeta_i =-1}Z_i$, supported on the collection of places where the query has a negative sign.
Now define $\eta(\zeta,a)$ to be $+ 1$ if the stabilizer $K_a$ commutes with the error $W_\zeta$ and $-1$ otherwise.
The function to be learned is:
\begin{align}
y_\beta(x) &= \prod_{a \in S} y_{\beta(x), a}, \quad \text{where } \\
y_{\beta(x), a} &=  \beta_a \eta(\zeta,a). \nonumber
\end{align}

In our case where the stabilizers are of the form $K_a = X_a\prod_{q \in N_G(a)}Z_q$, whether an error vector $W_\zeta$ commutes with $K_a$ or not is decided by whether $\zeta_a = \pm 1$. In other words, we have
\begin{align}
    W_\zeta^\dag K_a W_\zeta =  \zeta_a K_a,
\end{align}
so in this case $\eta(\zeta,a)=\zeta_a$. 

Therefore, the function is 
\begin{align}
y_\beta(x)= \left(\prod_{a \in S} \beta_a \right) \cdot \left(\prod_{a \in S} \zeta_a \right) \\
= \left(\prod_{a : \abs{x}_a \geq \xi} \beta_a  \cdot \text{sign}(x_a) \right),
\end{align}
where we have made the dependence on $x$ explicit by writing out $S$ and $\zeta_a$ in terms of $x$.

We define some more operators for notational convenience.
Let $L_a = -i Z_a K_a = Y_a \prod_{q \in N_G(a)} Z_q$.
This operator anticommutes both with $Z_a$ and $K_a$.
Also, define $L_S = \prod_{a \in S} L_a$ and similarly $Z_S$ and $K_S$. 
Note that there is no ambiguity with respect to ordering in these definitions because all the different $L_a$ operators (respectively, $Z_a$ and $K_a$) commute with each other.
For a given active set $S$, let $R_S$ be a Clifford operator that applies the map
\begin{align}
    R_S^\dag Z_1 R_S = L_S.
\end{align}
Now consider the unitary $U(x)=  R_S \prod_{a \in S}\exp[-i \alpha_a x_a Z_a/2]$. 
The form of $U(x)$ matches the definition in this work of a unitary that is efficiently constructible from the input $x$.
Moreover, the gates that depend on the magnitudes of $x$ have the same form as those considered in the numerical sections of this work, namely Pauli rotations whose angles depend linearly on the input features,
$ \exp\!\left[-\frac{i}{2}(\vec{\alpha}\cdot\vec{x})P\right], $
where $P$ denotes the corresponding Pauli generator.

Now 
\begin{align}
U(x)^\dag Z_1 U(x) &= \prod_{a \in S} \exp[+i \alpha_a x_a Z_a/2] R_S^\dag Z_1  R_S \prod_{b \in S} \exp[-i \alpha_b x_b Z_b/2] \\
&= \prod_{a \in S} \exp[+i \alpha_a x_a Z_a/2] L_S \prod_{b \in S} \exp[-i \alpha_b x_b Z_b/2]. 
\end{align}
Since $Z_a$ commutes with all $L_b$ except $L_a$, we get 
\begin{align}
U(x)^\dag Z_1 U(x) &= \prod_{a \in S} \left(\cos \frac{\alpha_a x_a}{2} I + i \sin \dfrac{\alpha_a x_a}{2} Z_a \right) L_a \left(\cos \frac{\alpha_a x_a}{2} I - i \sin \frac{\alpha_a x_a}{2} Z_a \right) \\
 &= \prod_{a \in S} (\cos{(\alpha_a x_a )} L_a +  i \sin {(\alpha_a x_a)} Z_a L_a) \\
 &= \prod_{a \in S} (\cos{(\alpha_a x_a )} L_a +  \sin {(\alpha_a x_a)} K_a).
\end{align}

Recall $x_a = \abs{x_a}\zeta_a$, $\zeta_a=\pm 1$. Then
\begin{align}
U(x)^\dag Z_1 U(x)
=\prod_{a\in S}
\left(\cos(\alpha_a \abs{x_a})L_a
+\zeta_a\sin(\alpha_a \abs{x_a})K_a\right).
\label{eq:graph-dressed-observable}
\end{align}
We note that this construction also illustrates the limitation of the reduced-density-matrix surrogate in \cref{eq: density surrogate} in Section~\ref{sec: dequantizability surrogate}. Here, the readout is $A=Z_1$, and the support of the dressed observable $U(x,\alpha)^\dagger Z_1U(x,\alpha)$ depends on the active set $S(x)$ and contains its qubits. Consequently, a fixed region supporting all training observables must contain $\bigcup_{\ell=1}^{N}S(x_\ell)$, the set of qubits active in at least one training example. Whenever the sampled active sets cover all $n$ qubits, the fixed supporting region for all the training observables $U(x)^\dag Z_1 U(x)$ must contain the entire register. Therefore, the corresponding density matrix $\sigma_1$ has dimension $2^n$. Thus, the polynomial-time argument based on explicitly diagonalizing matrices on small fixed regions does not apply. We remark that the accuracy separation proved below concerns the Pauli surrogate introduced in Eq.~\eqref{eq: pauli surrogate base}.

\subsection{Idealized embedding Hamiltonian and its ground state: perfect classification}
The Hamiltonian term corresponding to this data point is
\begin{align}
  H(x,y) &= \frac{I - y U^\dag Z_1 U}{2} \\  
  &= \frac{I - y \prod_a (\cos{(\alpha_a \abs{x_a})} L_a +  \zeta_a \sin {(\alpha_a \abs{x_a})} K_a)}{2} \\
  &= \frac{I - \prod_a \beta_a \zeta_a \left(\cos{(\alpha_a \abs{x_a})} L_a +  \zeta_a \sin {(\alpha_a \abs{x_a})} K_a \right)}{2}.
\end{align}
Rewrite this in terms of binary variables $z \in \{0,1\}^\abs{S}$ as
\begin{align}
 H(x,y) = \frac12 \times \left[ {I - \left(\prod_{a \in S} \beta_a \right) \sum_{z \in \{0,1\}^\abs{S}}   P_z } \right],
\end{align}
where 
$P_z = \prod_{a:z_a=0} \zeta_a \cos{(\alpha_a \abs{x_a})} L_a \times \prod_{a:z_a=1}\sin{(\alpha_a \abs{x_a})} K_a.$
If we average this over all data points from the distribution, then considering the different signs of $x_a$ at the same magnitude, i.e. averaging over $\zeta_a =\pm 1$ while keeping $\abs{x_a}$ fixed:
we get
\begin{align}
  P_z + P_{z_{\bar a}} = 0 \quad \text{if } z_a = 0, 
\end{align}
where $z_{\bar a}$ is the binary vector $z$ with the bit at position $a$ flipped: $z_{{\bar a}_a}=1-z_a$.
The average over signs gives 
\begin{align}
  \sum_{z \in \{0,1\}^n} P_z =  \prod_{a \in S} \sin{(\alpha_a \abs{x_a})} K_a.
\end{align}
This means the embedding Hamiltonian averaged over different signs for the same $\abs{x_a}$ gives
\begin{align}
H_\mathcal{D}(\abs{x}) &=  \frac{1}{2} \times \left[ I -  \left(\prod_{a \in S} \beta_a \right) \prod_{a \in S} \sin{(\alpha_a \abs{x_a})} K_a  \right] \\
&= \frac{I - \prod_{a \in S} (\beta_a  \sin {(\alpha_a \abs{x_a})} K_a)}{2}.
\end{align}

For the full embedding Hamiltonian, we should also average over the distribution of $x$ for all coordinates, denoted $g(x)$:
\begin{align}
H_\mathcal{D} = \int d^n x g(x) \frac{I - \prod_{a \in S} \beta_a  (\sin {(\alpha_a \abs{x_a})} K_a)}{2}. \label{eq: full_embedding_Hamiltonian_graph}
\end{align}

This is a commuting Hamiltonian whose terms are products of $\beta_a K_a$.
The number of products that appear is precisely $\abs{S}$, which is the sparsity of the input vector $x$ above the cutoff $\xi$.
If the distribution over the input data is supported on 1-sparse inputs $x$, then these stabilizers $\beta_a K_a$ are represented in the Hamiltonian.
Since these stabilizers commute with each other, the full embedding Hamiltonian is frustration-free. Its ground space is the +1 eigenspace of the operators $\cup_S \{ \prod_{a \in S}\beta_a K_a \}$.
Note that if the distribution that the input data $x$ is drawn from is not supported on all coordinates, the group generated by the stabilizer generators $\cup_S \{\prod_{a \in S}\beta_a K_a\}$ may not be full rank.

Now consider classification.
Assume that the ground state $\ket{\psi}$ of the embedding Hamiltonian is found by the ground state solver (which we treat as a black box in this Appendix).
The expectation value of $U(x)^\dag Z_1 U(x)$ in this state is:
\begin{align}
\bra{\psi} U(x)^\dag Z_1 U(x) \ket{\psi} = \bra{\psi} \prod_{a\in S} (\cos{(\alpha_a \abs{x_a})} L_a +   \zeta_a \sin {(\alpha_a \abs{x_a})} K_a )\ket{\psi}.
\end{align}
We have $\bra{\psi} L_a \ket{\psi} =0$ because $L_a$ anticommutes with $\beta_a K_a$, which is the stabilizer of $\ket{\psi}$.
Therefore,
\begin{align}
\bra{\psi} U(x)^\dag Z_1 U(x) \ket{\psi} = \bra{\psi} \prod_{a\in S} (\cos{(\alpha_a \abs{x_a})} L_a +   \zeta_a \sin {(\alpha_a \abs{x_a})} K_a )\ket{\psi} = \prod_{a\in S}\beta_a \zeta_a \sin {(\alpha_a \abs{x_a})}.
\end{align}
The sign of this is $\prod_{a\in S} \beta_a \zeta_a \sign(\alpha_a)$, which matches the function $y_\beta(x)$ whenever $\sign(\alpha_a)=1$ for all $a\in S$. Therefore, classification is perfect as long as $\alpha_a > 0$.

\subsection{Embedding Hamiltonian with finite training data}
In this subsection, we specify the distribution over inputs and prove that there exists a regime of number of training points such that the empirical Hamiltonian embedding model classifies the data nearly perfectly; in this same regime we can also prove an upper bound on the accuracy of the classical Pauli surrogates described in \cref{sec: Pauli surrogates}. This separation between the classical surrogate and the Hamiltonian embedding model proves that this model generalizes better when the model is suited to the task.

We choose the input training data to be drawn from the following distribution. 
First, we choose the size of the active set according to the distribution
\begin{align}
\Pr(\abs{S}=2)=\frac{4}{7},\qquad \Pr(\abs{S}=3)=\frac{2}{7},\qquad \Pr(\abs{S}=4)=\frac{1}{7}. \label{eq: query_distribution}
\end{align}
Then, for a given size $\abs{S}$, we pick the active coordinates uniformly.
Let $a$ denote a location that is active, that is $a\in S$.
The input vector $x$ at coordinate $a$ has magnitude $\abs{x}_a = r_a$, where $r_a$ is a random number in $[\xi,1]$.
This number $r_a$ is fixed for the duration of the experiment, that is, any other training point with $a$ an active location will have $\abs{x}_a = r_a$. 

Now consider the empirical embedding Hamiltonian as opposed to the idealized one in \cref{eq: full_embedding_Hamiltonian_graph}.
The difference is that in practice we may not have queried every possible input, which is precisely the case where we expect the Hamiltonian embedding model to generalize better by virtue of its inductive bias.
The empirical embedding Hamiltonian is
\begin{align}
 H_\T(\alpha) = \frac{1}{N} \sum_{\ell} H(x_\ell,y_\ell) = \sum_{\ell} \frac{1}{2N} \times \left[ {I - \left(\prod_{a \in S^\ell} \beta_a \right) \sum_{z \in \{0,1\}^\abs{S^\ell}}   P_{z^\ell}(\alpha) } \right],
\end{align}
where 
$P_{z^\ell}(\alpha) = \prod_{a:z_a^\ell=0} \zeta_a^\ell \cos{(\alpha_a \abs{x^\ell_a})} L_{a} \times \prod_{a:z_a^\ell=1}\sin{(\alpha_a \abs{x_{a}^\ell})} K_a$ and $\ell =1,2, \ldots N$ denotes the different input training data points.
Since, as mentioned before, for different data points we have the same magnitude $\abs{x^\ell_a}= r_a \forall \ell$, we can rewrite this as:
\begin{align}
P_{z^\ell}(\alpha) = \prod_{a:z_a^\ell=0} \zeta_a^\ell \cos{(\alpha_a r_a)} L_{a} \times \prod_{a:z_a^\ell=1}\sin{(\alpha_a r_a)} K_a.
\end{align}
Notice that in this case, since we do not assume that the different signs $\zeta_a$ queried at the same magnitude $r_a$ cancel out, and thus the cosine term survives.
This can potentially be problematic because $L_a$ anticommutes with $K_a$ and would make the Hamiltonian frustrated.
Nevertheless, the empirical embedding Hamiltonian can have the same ground state $\ket{\psi}$ as the full embedding Hamiltonian if the angles $\alpha$ can be trained.
In this case, setting $\alpha^*_a  = \frac{\pi}{2r_a}$ ensures that the cosine term gives 0 and thus only the all-sine term $z^\ell = 1\ldots 1$ survives:
$P_{z^\ell}(\alpha) = \prod_{a \in S^\ell} K_a$, 
giving
\begin{align}
 H_\T(\alpha^*) = \sum_{\ell} \frac{1}{2N} \times \left[ {I - \left(\prod_{a \in S^\ell} \beta_a K_a \right) } \right].
\end{align}
As we have seen before, the Hamiltonian has as ground space the +1 eigenspace of the stabilizer group $\cup_{\ell} \{\prod_{a \in S^\ell} \beta_a K_a\}$, with energy $\expval{H_\T}=0$.
As we have restricted ourselves to positive semidefinite Hamiltonians, this proves that $\alpha^*$ is an optimum.

\subsubsection{Empirical accuracy of Hamiltonian embedding}
Consider the binary matrix $A$, whose $\ell$th row records the active set $S^\ell$.
We are interested in the event that $A$ has rank $n$ over $\mathbb F_2$.
This is equivalent to whether the sampled stabilizer products generate the full stabilizer
group $\langle\beta_a K_a:a\in[n]\rangle$.
In this case, the empirical embedding Hamiltonian would have the unique ground state $\ket{\psi}$.
Note that merely requiring every coordinate $a$ to appear in some sampled query is not sufficient for this spanning condition.
On the other hand, if $N$ is large enough that every possible input in the support of the distribution is observed, the classical surrogate can also classify perfectly by
memorization.

Therefore, we now investigate the case where $N$ is bounded. 
Specifically, we choose $N=4n\log n$.
Our analysis of accuracy depends on the probability that the stabilizer group $\cup_{\ell} \{\prod_{a \in S^\ell} \beta_a K_a\}$ has full rank.
Consider the matrix $A_{N\times n}$ whose rows correspond to training data.
For each data point $\ell$, we have $A_{\ell, a}=1$ if $a$ is in the active set $S^\ell$, and 0 otherwise.
Now $A$ is not full rank if and only if there exists a nonzero binary vector $v$ such that $A \cdot v = 0$.
Denote by $T$ the set of locations $i$ where $v_i = 1$ and let $\abs{T}=t$.
Then
\begin{align}
A_{\ell} \cdot v = 0 \Leftrightarrow \abs{S^\ell \cap T} = 0 \bmod 2.
\end{align}
Suppose we fix a $v$ and consider whether for a random query $\ell$ from the distribution $v$ is ``detected'' by $\ell$, i.e.~whether $\abs{S \cap T} = 1 \bmod 2$.
The probability of this event, conditioned on $\abs{S}=2$, is:
\begin{align}
 \Pr\left[\abs{S \cap T} = 1 \bigg \vert \abs{S}=2\right] = \frac{t \times (n-t)}{\binom{n}{2}}.
\end{align}
Similarly, conditioning on $\abs{S}=3$, we get the two cases $\abs{S \cap T} = 1 $ or $3$, giving 
\begin{align}
\Pr\left[\abs{S \cap T} = 1 \bmod 2 \bigg \vert \abs{S}=3 \right] = \frac{t \times \binom{n-t}{2}}{\binom{n}{3}} + \frac{\binom{t}{3} \times \binom{n-t}{0}}{\binom{n}{3}}.
\end{align}
Lastly, when $\abs{S}=4$, we again have the possibilities  $\abs{S \cap T} = 1 $ or $3$, with the probability 
\begin{align}
\Pr\left[\abs{S \cap T} = 1 \bmod 2 \bigg \vert \abs{S}=4 \right] = \frac{t \times \binom{n-t}{3}}{\binom{n}{4}} + \frac{\binom{t}{3} \times \binom{n-t}{1}}{\binom{n}{4}}.
\end{align}
Denote the combined probability
\begin{align}
r_t = \Pr\left[\abs{S \cap T} = 1 \bmod 2 \right] &= \frac{4}{7} \Pr\left[\abs{S \cap T} = 1 \bigg \vert \abs{S}=2\right] + \frac{2}{7} \Pr\left[\abs{S \cap T} = 1 \bmod 2 \bigg \vert \abs{S}=3\right]  \nonumber \\
& + \frac{1}{7}\Pr\left[\abs{S \cap T} = 1 \bmod 2 \bigg \vert \abs{S}=4\right],
\end{align}
where we have made use of the distribution over the active set sizes.
The probability that no row detects $v$ is given by $(1-r_t)^N \leq e^{-N r_t}$.
By a union bound over all nonzero $v\in\mathbb F_2^n$,
\begin{align}
\Pr[\operatorname{rank}(A)<n] \leq \sum_{t=1}^n \binom nt e^{-Nr_t}. \label{eq:rank-deficiency-union-bound}
\end{align}
We set $N=\lceil4n\log n\rceil$ and divide the sum into two regimes.

First, suppose $1\leq t\leq n/2$. Retaining only the contribution from 2-sparse queries gives
\begin{align}
r_t &\geq \frac47\frac{t(n-t)}{\binom n2}\\
&=
\frac{8t(n-t)}{7n(n-1)}
\geq
\frac{4t}{7n}.
\end{align}
Consequently,
\begin{align}
Nr_t\geq\frac{16}{7}t\log n.
\end{align}
Using $\binom nt\leq(en/t)^t$, the contribution of this regime to
Eq.~\eqref{eq:rank-deficiency-union-bound} is at most
\begin{align}
\sum_{t=1}^{\lfloor n/2\rfloor}
\binom nt e^{-Nr_t}
&\leq
\sum_{t=1}^{\lfloor n/2\rfloor}
\left(\frac{en}{t}\right)^t n^{-16t/7}\\
&\leq
\sum_{t\geq1}
\left(\frac{e}{n^{9/7}}\right)^t
=
O\!\left(n^{-9/7}\right).
\end{align}
Next, suppose $n/2<t\leq n$. A 3-sparse query lying entirely inside
$T$ has odd intersection with $T$, and hence
\begin{align}
r_t
\geq
\frac27\frac{\binom t3}{\binom n3}
\geq
\frac{1}{56}
\end{align}
for all sufficiently large $n$. The contribution of this regime is
therefore at most
\begin{align}
2^n e^{-N/56}
\leq
\exp\left(
n\log2-\frac{n\log n}{14}
\right)
=o(1).
\end{align}
Combining the two regimes gives
\begin{align}
\Pr[\operatorname{rank}(A)<n]=o(1),
\end{align}
so $A$ has rank $n$ with probability $1-o(1)$.
This means that the accuracy of the Hamiltonian embedding model, assuming that the ground-state solver is perfect, would be $1-o(1)$.

\subsubsection{Accuracy of classical surrogates}
For an active set $S$, define
\begin{align}
P_{S,z}
&:=
\prod_{a\in S:z_a=0}L_a
\prod_{a\in S:z_a=1}K_a,\\
c_{S,z}(x,\alpha)
&:=
\prod_{a\in S:z_a=0}\cos(\alpha_a\abs{x_a})
\prod_{a\in S:z_a=1}
\zeta_a\sin(\alpha_a\abs{x_a}),
\end{align}
where $z\in\{0,1\}^{\abs{S}}$. The inference observable and its
Pauli surrogate are then
\begin{align}
U(x,\alpha)^\dag Z_1 U(x,\alpha)
&=\sum_{z\in\{0,1\}^{\abs{S}}}c_{S,z}(x,\alpha)P_{S,z},\\
g_S(x,\alpha)
&=\sum_{z\in\{0,1\}^{\abs{S}}}c_{S,z}(x,\alpha)m_{S,z}.
\label{eq:graph-state-pauli-surrogate}
\end{align}
Here, the expectation value of $P_{S,z}$, namely $\bra{\phi}P_{S,z}\ket{\phi}$, is replaced by the coefficient $m_{S,z}$, which is unconstrained.
For the graph stabilizers considered here, $P_{S,z}$ has an $X$ or $Y$ on precisely the qubits in $S$. Pauli observables belonging to different active sets are therefore distinct, so the free surrogate has a separate coefficient block for each possible query $S$.

For fixed $\alpha$, the loss function used in ridge regression is
\begin{align}
J(m;\alpha)
&=
\frac{1}{2N}
\sum_{\ell=1}^N
\left[
y_\ell-
\sum_{z\in\{0,1\}^{\abs{S^\ell}}}
c_{S^\ell,z}(x^\ell,\alpha)m_{S^\ell,z}
\right]^2 + \frac{\lambda}{2} \sum_{S \in \mathcal{Q}}\sum_{z\in\{0,1\}^{\abs{S}}}m_{S,z}^2,
\label{eq:graph-state-ridge-objective}
\end{align}
where the second sum $S \in \mathcal{Q}$ is over all supports $\mathcal{Q}=\{S \subseteq [n]\}$.
If a support $S$ is absent from the training set, its coefficient block does not enter the data-fitting term. The unique ridge solution is therefore
\begin{align}
m_{S,z}^{\star}(\alpha)=0
\qquad
\text{for all }z,\quad S\notin\mathcal Q_{\mathcal T},
\label{eq:unseen-support-zero}
\end{align}
where $\mathcal Q_{\mathcal T}$ is the set of observed supports. This holds
for every $\alpha$, and hence also when $\alpha$ and $m$ are optimized
jointly for the classical surrogate.

We can evaluate the coefficients explicitly when $\alpha$ is fixed to the
Hamiltonian-embedding optimum constructed above,
$\alpha_a^\star r_a=\frac{\pi}{2}$.
Writing $\beta_S:=\prod_{a\in S}\beta_a$, we then have
\begin{align}
c_{S,\mathbf 1_S}(x,\alpha^\star)
=
\prod_{a\in S}\zeta_a,
\qquad
c_{S,z}(x,\alpha^\star)=0
\quad
(z\neq\mathbf 1_S),
\end{align}
and the label is
$
y_\beta(x)=\beta_S\prod_{a\in S}\zeta_a
$.

If $N_S$ denotes the number of training examples with support $S$, the
ridge solution is found by setting 
\begin{align}
\frac{\partial J(m;\alpha)}{m_{S,z}} &= 0 \\
&=  -\frac{N_S}{N} \beta_S \delta_{z,\mathbf{1}_S} + \frac{N_S}{N} m_{S,z} \delta_{z,\mathbf{1}_S} + {\lambda} m_{S,z},
\end{align}
giving 
\begin{align}
m_{S,z}^{\star}(\alpha^\star)
=
\begin{cases}
\displaystyle
\beta_S\frac{N_S}{N_S+\lambda N},
&N_S>0\ \text{and}\ z=\mathbf 1_S,\\[2mm]
0,&\text{otherwise}.
\end{cases}
\label{eq:graph-state-ridge-solution}
\end{align}
Consequently, every observed support is classified correctly:
\begin{align}
y_\beta(x)g_S(x,\alpha^\star)
=
\frac{N_S}{N_S+\lambda N}>0,
\qquad N_S>0,
\end{align}
whereas an unobserved support has zero score.

Let $M_k(\mathcal T)$ be the number of distinct $k$-element supports
observed during training. A new query picked at random according to the distribution in
Eq.~\eqref{eq: query_distribution}, will have been seen previously with probability
\begin{align}
q_{\rm seen}(\mathcal T) = \sum_{k=2}^4 p_k\frac{M_k(\mathcal T)}{\binom nk}.
\end{align}
On a previously unobserved support the accuracy is $1/2$. Therefore, for any fixed or independently optimized $\alpha$, the ridge classifier's accuracy is
\begin{align}
A_{\rm ridge}(\mathcal T;\alpha)
&= \frac{1-q_{\rm seen}(\mathcal T)}{2} + q_{\rm seen}(\mathcal T) A_{\rm seen}(\mathcal T;\alpha),
\label{eq:ridge-occupancy-bound}
\end{align}
where $A_{\rm seen}$ is the classifier's accuracy on previously seen points.
This gives 
\begin{align}
\abs{A_{\rm ridge}(\mathcal T;\alpha)-\frac12} & \leq \frac12q_{\rm seen}(\mathcal T).
\end{align}
When $\alpha$ is the optimal Hamiltonian-embedding solution, \cref{eq:graph-state-ridge-solution} gives $A_{\rm seen}=1$, and hence
\begin{align}
A_{\rm ridge}(\mathcal T;\alpha^\star) = \frac12+\frac12q_{\rm seen}(\mathcal T). \label{eq:target-angle-ridge-accuracy}
\end{align}
Thus even if the surrogate optimizes $\alpha$ independently its accuracy remains bounded; moreover, fixing $\alpha=\alpha^\star$ shows that the bound is achievable.

Lastly, we bound $q_{\rm seen}$. For a particular $k$-element support, the probability of appearing at least once in $N$ independent training queries is
\begin{align}
1-\left(1-\frac{p_k}{\binom nk}\right)^N.
\end{align}
Averaging \cref{eq:ridge-occupancy-bound} over the training set gives
\begin{align}
\mathbb E_{\mathcal T}[A_{\rm ridge}]
&\leq
\frac12+ \frac12\sum_{k=2}^4p_k \left[ 1-\left(1-\frac{p_k}{\binom nk}\right)^N \right]\\
&\leq \frac12+ \frac12\sum_{k=2}^4 \frac{ N p_k^2}{\binom nk} \\
&\leq \frac12+ \frac{N}{98} \left[ \frac{16}{\binom n2} +\frac4{\binom n3} +\frac1{\binom n4} \right],
\label{eq:expected-ridge-occupancy}
\end{align}
where the first inequality is an equality at $\alpha^\star$. In particular,
for $N=\Theta(n\log n)$,
\begin{align}
A_{\rm ridge}
=
\frac12+O\!\left(\frac{\log n}{n}\right).
\end{align}

The difference between the surrogate and the Hamiltonian embedding model is that the latter combines the observed stabilizer products: once their incidence vectors (i.e.\ the rows in $A$) have full rank, the underlying generator signs $\beta_a$ are fixed and every support can be classified.
The fact that a single state is consistent with all Pauli expectation values and also informs the expectation values of unseen Paulis by virtue of the Pauli algebra is what provides the inductive bias in this case, allowing Hamiltonian embedding to generalize better than the classical surrogate.

\subsection{Required bond dimension for good classification with Hamiltonian embedding}
In this subsection, we empirically study this previous example by instantiating a family of graphs and their corresponding codes.
We randomly generate a simple, connected 3-regular graph, which fixes the stabilizer code family, at each system size.
As mentioned earlier, the stabilizers are $\beta_i X_i \prod_{j\in N(i)}Z_j$, with the signs $\beta \in \pm 1$ that describe the function also drawn uniformly at random at every vertex.
We set the cutoff $\xi = 0.6$ and draw the radius of every active coordinate $\abs{x_i} \sim \operatorname{Unif}[0.6,1]$ for every realization and fix them thereafter.
The signs of the input coordinates are again chosen uniformly at random.
This fixes the distribution over input vectors $\in \mathbb{R}^n$.
We take the same regime of $N=4n\log n$, with the same distribution over queries as \cref{eq: query_distribution}.

In order to train the Hamiltonian embedding model, we adopt a different numerical technique from the one considered in \cref{sec: numerics}.
Specifically, we apply imaginary time evolution on a classical MPS representation of the state.
The MPS has a fixed ordering and no attempt was made at optimizing the ordering in order to better represent the ideal stabilizer ground state, since this strategy does not scale in general.
After each application of imaginary time evolution, we compress the MPS up to a maximum bond dimension $\chi$, with a numerical cutoff of $10^{-8}$ for the singular value decomposition (SVD).
For the numerical optimization, we parameterize the 
rotation angles as
\begin{align}
\alpha_a=\alpha_c+\Delta_\alpha u_a,\qquad u_a\in[-1,1],
\end{align}
where
\begin{align}
\alpha_c=\frac{\pi}{4}(1+\xi^{-1}),\qquad
\Delta_\alpha=\frac{\pi}{4}(\xi^{-1}-1).    
\end{align}
Thus $\alpha_a\in[\frac{\pi}{2},\frac{\pi}{2\xi}]$, an interval containing the analytic optimum $\alpha_a^\star=\frac{\pi}{2r_a}$ for $r_a\in[\xi,1]$. We initialize $u_a\sim\operatorname{Unif}[-0.75,0.75]$ and optimize $u$ using Adam, projecting it onto $[-1,1]^n$ after each update.
During this procedure, we optimize the energy of the parameterized Hamiltonian $H_T(\alpha)$, keeping the state fixed.

We illustrate the accuracy of the trained model as a function of the maximum bond dimension cap $\chi$ in \cref{fig: trainable alpha graph state}.
Writing $m(x)=y(x)\langle U(x)^\dagger Z_1U(x)\rangle$, each test example receives score $1$ when $m(x)\ge\eta$, $0$ when $m(x)\le-\eta$, and $1/2$ when $|m(x)|<\eta$.
The accuracy $A_\eta$ is the average of these scores.
We empirically find that increasing bond dimension is necessary in order to achieve the same accuracy at larger system sizes.
Similarly, with increasing bond dimension, the energy and the observable's margins also improve.

In \cref{fig: surrogate}, we also plot the empirical accuracy of the classical surrogates and compare with the one from Hamiltonian embedding.
For the points labeled as coming from HE, we show the highest test score observed over all bond dimensions.
We notice that the classical surrogates trained on the ridge loss show a gradual decay in accuracy with respect to system size $n$.

\begin{figure}
    \centering
    \includegraphics[width=0.9\linewidth]{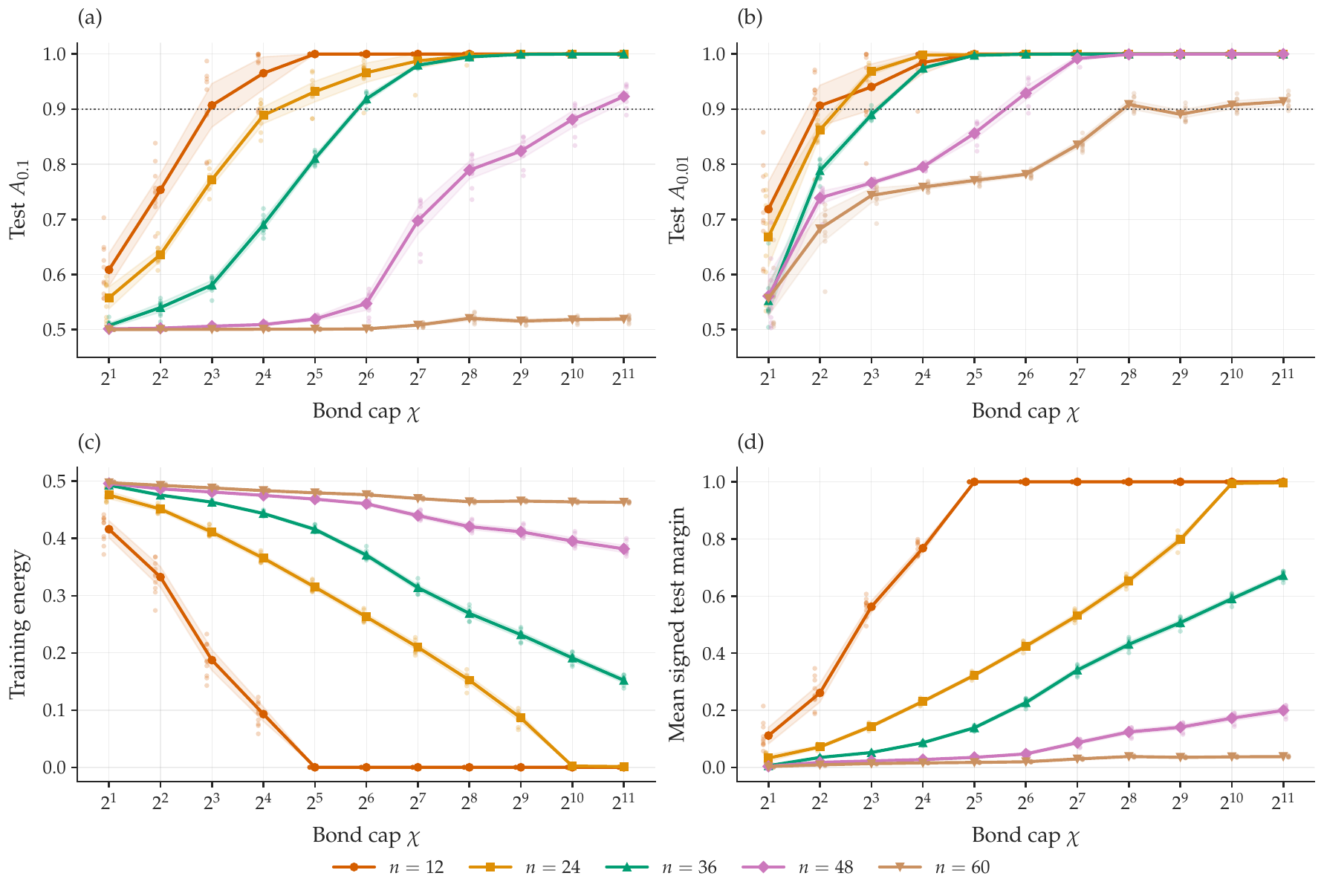}
    \caption{Test accuracy $A_\eta$ as a function of bond dimension at system sizes $n \in \{12,24,36,48, 60\}$, averaged over 10 trials, with (a) $\eta=0.1$, (b) $\eta = 0.01$. (c) Minimum energy seen during training as a function of bond dimension. (d) The signed margin ${y(x) \expval{U(x)^\dag Z_1 U(x)}}$.}
    \label{fig: trainable alpha graph state}
\end{figure}

\begin{figure}
    \centering
    \includegraphics[width=0.6\linewidth]{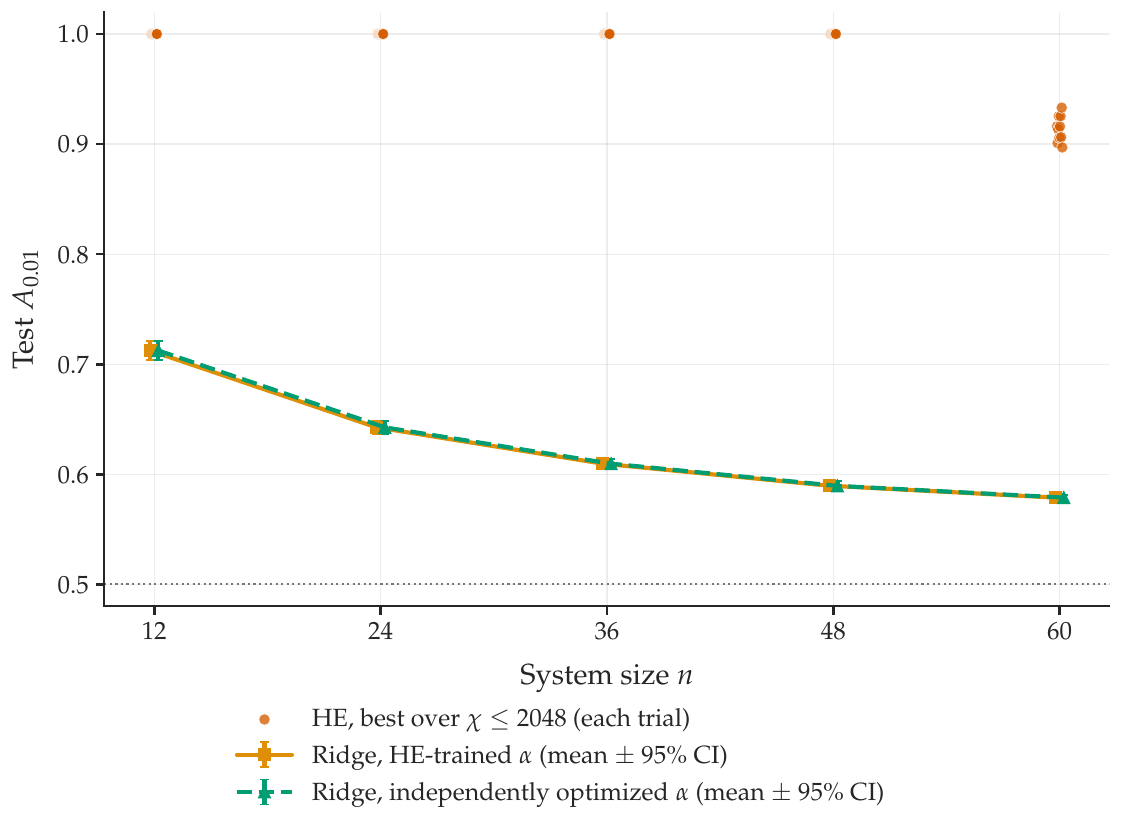}
    \caption{Accuracy of surrogate models. The green curve shows the surrogate that optimizes the $m_{S,z}$ from \cref{eq:graph-state-pauli-surrogate}, keeping $c_{S,z}(\alpha)$ fixed by taking $\alpha=\alpha^*$ obtained from the Hamiltonian embedding model. The orange curve optimizes over both $m_{S,z}$ and $\alpha$ (which determines $c_{S,z}(\alpha)$). Also shown in orange dots for comparison are the best (over bond dimensions $\chi$) accuracies for the Hamiltonian embedding model for each trial.}
    \label{fig: surrogate}
\end{figure}

\section{Numerical setup}\label{app: numerics setup}

In this Appendix, we discuss the Hamiltonian embedding model used in all the experiments of Section~\ref{sec: numerics}, together with the classical simulation of the training loop. The protocols of the individual tasks are given in Appendices~\ref{app: fashion} and~\ref{app: qpu}. 

\subsection{Encoding circuit and embedding Hamiltonian}

We use the encoding unitary of Eq.~\eqref{eq: encoding U} with $K=3$ layers. Layer $\ell$ applies to every qubit $j$ the sequence of single-qubit Pauli rotations $e^{-i\theta_{j,3}^{\ell}Z_j/2}e^{-i\theta_{j,2}^{\ell}Y_j/2} e^{-i\theta_{j,1}^{\ell}Z_j/2}$, followed by a ring of nearest-neighbor two-qubit rotations $e^{-i\theta_{j,4}^{\ell}Z_jZ_{j+1}/2}$ with periodic boundary conditions. All $12n$ angles are independently trainable, and each is a linear function of the input feature and of the trainable weights,
\begin{equation}
    \theta_g(x,\alpha)=\sum_{r=1}^{d}\alpha_{g,r}x_r+\alpha_{g,0},
    \qquad g=1,\ldots,12n,
\label{eq: numerics angles}
\end{equation}
so the trainable parameters are $\alpha\in\mathbb{R}^{12n\times(d+1)}$ and they are a total of $M=12n(d+1)$. Every trainable gate is generated by a Pauli operator with $P^2=I$, so the cost admits the parameter-shift rule of Section~\ref{sec: barren plateaus and gradient}.
The two-qubit rotations are nearest-neighbor, so the circuit remains of constant depth with a light cone of constant size.
Lastly, the label observable is set to $A=\frac{1}{n}\sum_{i=1}^{n}Z_i$ throughout.

We consider a constant-depth encoding unitary $U(\alpha, x)$ such that the observable $A$ can be evolved under $U(\alpha, x)$ classically. In particular, using backpropagation, we store a list of which Pauli terms appear in the embedding Hamiltonian together with the functional dependence of their coefficients from the features and trainable weights.

\subsection{Classical iterative solver}\label{app: solver} 
We describe the training procedure when using the classical iterative solver.
At the first training step, we start from the reference state $\ket{\mathrm{ref}}=\ket{0^n}$ and construct truncated approximations to the six vectors $\{H_{\mathcal T}^{k}(\alpha_0)\ket{\mathrm{ref}}\}_{k=0}^{5}$ obtained by repeatedly applying the data Hamiltonian evaluated at the initial parameters $\alpha_0$. After each application, we keep up to $8{,}000$ computational-basis bitstrings with the largest amplitudes. We then diagonalize $H_{\mathcal T}(\alpha_0)$ in the subspace spanned by the collected computational-basis states, obtaining an approximate ground state.
Let us refer to this procedure of applying the Hamiltonian, collecting bitstrings, and diagonalizing as one ``round''.
We perform three rounds in total at fixed $\alpha_0$, where the ground state from each round is used as the reference state for the next, and the diagonalization in each round is in the subspace spanned by all collected bitstrings from previous rounds.
This gives us the approximate ground state for $H_{\mathcal T}(\alpha_0)$. At subsequent training steps, we use an analogous procedure, taking the approximate ground state obtained at step $t-1$ as the reference state. Here, each step involves a single round with only one application of $H_{\mathcal T}(\alpha_t)$, yielding truncated approximations to the two Krylov states $\{H_{\mathcal T}^{k}(\alpha_t)\ket{\mathrm{ref}}\}_{k=0}^{1}$. Furthermore, for optimization steps $t>0$, the truncation limit is reduced from $8{,}000$ to $3{,s}000$ computational-basis components per state, retaining those with the largest amplitudes.

\subsection{Optimization and model selection}

The trainable parameters are initialized by drawing from independent Gaussian distributions $\alpha_{g,r}\sim\mathcal{N}(0,0.3^2)$, using a fixed seed
for each run to ensure reproducibility.
To limit overfitting, we add a squared-$\ell_2$ penalty to
the training objective,
\begin{equation}
    \widehat C_{\mathcal T}^{\mathrm{reg}}(\psi,\alpha)
    =
    \widehat C_{\mathcal T}(\psi,\alpha)
    + \frac{\lambda}{2}\|\alpha\|_2^2,
\end{equation}
where $\alpha$ denotes the full parameter vector and
$\lambda=0.005$. 
At iteration $t$, we hold the state $\ket{\psi_t}$ used for gradient evaluation fixed and compute
\begin{equation}
    g_t =
    \left.
    \bra{\psi_t}
    \nabla_\alpha H_{\mathcal T}(\alpha)
    \ket{\psi_t}
    \right|_{\alpha=\alpha_t}
    + \lambda\alpha_t.
\end{equation}
For the parameter update step, we use Adam with learning rate $0.05$ and constants $\beta_1=0.9$, $\beta_2=0.999$, and $\epsilon=10^{-8}$. Training is performed for $80$ iterations for the image-classification benchmark described in the next section.

Model selection after training is performed by scoring the performance at each training epoch against a validation set that is split off the training data. In practice we hold out $20\%$ of the latter, and the final reported model corresponds to the epoch with the highest validation accuracy. Classification uses the score $f(x)=\bra{\psi}U^\dagger(x,\alpha)A\, U(x,\alpha)\ket{\psi}$ with label $0$ assigned when $f(x)\ge0$ and label $1$ otherwise.

\section{Image classification benchmark}\label{app: fashion}

\subsection{Protocol}

The task considered is the binary classification of the Fashion-MNIST dataset~\cite{xiao2017fashion}, restricted to the classes \emph{dress} (label $3$) and \emph{shirt} (label $6$), which is the setting used in Appendix~L of Ref.~\cite{huang2021power}. Following that protocol, the two classes are pooled ($14{,}000$ images of $784$ pixels each), the pool is reduced by principal component analysis to $n$ components, each component is standardized to zero mean and unit variance, and $N_{\text{pool}}=800$ points are subsampled without replacement using the seed of the run. A stratified $70/30$ train--test split of the $800$ points is then drawn from the same seed, and a further stratified $80/20$ split of the training part provides the inner validation set used for model selection. For every pair (qubit count, seed) the subsample and the splits are drawn once and shared by every method, so all comparisons are paired.

\subsection{Experimental Baselines}\label{app: kernel sim}

The quantum baselines are those of Ref.~\cite{huang2021power}, evaluated on the same splits and with the same hyperparameters. Three data embeddings $x\mapsto\ket{\phi(x)}$ are used: a product of single-qubit rotations, one per feature (E1); an IQP-type circuit in which two layers of Hadamards sandwich diagonal phases with linear terms $x_jZ_j$ and all-to-all quadratic terms $x_jx_{j'}Z_jZ_{j'}$ (E2); and a Trotterized one-dimensional Heisenberg evolution for time $t=n/3$, with couplings set by the features, acting on a fixed Haar-random product state (E3). 
For each embedding we compute the fidelity quantum kernel $K_{ij}=|\bra{\phi(x_i)}\phi(x_j)\rangle|^2$ and the projected quantum kernel built from the single-qubit reduced density matrices, $K_{ij}=\exp[-\gamma\sum_{q}\|\rho_q(x_i)-\rho_q(x_j)\|_F^2]$, together with its generalizations in which the reduced density matrices of all pairs (2-RDM) or all triples (3-RDM) of qubits enter the exponent. Each kernel is fed to a support vector machine with a precomputed Gram matrix. A penalty term $C$ is scanned over an 18-point grid and the bandwidth over $\gamma\in\{0.25,0.5,1,2,3,4,5,20\}/(n\,\mathrm{Var}(X))$ (see Ref.~\cite{huang2021power} for more details). Finally to compare against the performance of hardware runs (in Fig. \ref{fig: qpu train overlay} (b)), the best kernel results are selected based on their validation accuracy.

Simulating these baselines faithfully at $50$ qubits requires embedding-specific methods, and not every embedding reaches every $n$. The E1 states are products of single-qubit states, so both kernels admit exact closed-form solutions. The E2 circuit generates volume-law entanglement through its all-to-all diagonal phases and admits no efficient tensor-network contraction, so it is simulated by exact statevector simulations, which bounds its reach by memory. The fidelity kernel is computed up to $n=16$, and the projected kernels up to $n=24$ ($n=20$ at 2-RDM order). 
The E3 evolution also generates extensive entanglement, and hence its exact matrix-product-state bond dimension grows as $2^{\lceil n/2\rceil}$. We simulate it by time-evolving-block-decimation with $20$ Trotter steps for the total time $t=n/3$, capping the bond dimension as $\chi=\min(256,2^{\lceil n/2\rceil})$. With this choice, the MPS can represent exactly any quantum state up to $n\le16$, and the only approximation error comes from the finite Trotter step, which we measured at $n=16$ as a projected-kernel feature error of $1.6\times10^{-2}$. We do not evaluate E3 beyond $n=16$, where $\chi=256$ would cease to be exact.  We note that the classical hardness of simulating the E2 and E3 embeddings at large $n$ does not, by itself, confer a learning advantage, as the corresponding kernels are precisely the ones that concentrate and fail on this task.

For the classical baseline, we consider support vector machines with radial basis function and linear kernels, random forests, gradient boosting, adaptive boosting, and a multilayer perceptron.
We use the same $n$-dimensional features for all methods. The hyperparameters for these methods follow the same grids of~\cite{huang2021power}.
For each model class, hyperparameters are tuned by three-fold cross-validation on the training split. The selected model is then evaluated once on the test split.

\section{Hardware experiments}\label{app: qpu}

\subsection{The hardware solver}

The hardware runs replace only the sampling stage of the solver of Appendix~\ref{app: solver} by a quantum device, whereas the construction of $H_{\T}$, the subspace union, the projection, the diagonalization, the gradient and the Adam update are performed classically on a laptop. At iteration $t$ the current $H_{\T}(\alpha_t)$ is converted into a Pauli operator and the circuits
\begin{equation}
    \ket{\mathrm{ref}}\longmapsto
    \left(e^{-i\,\delta t\,\widetilde H_{\T}(\alpha_t)}\right)^{k}\ket{\mathrm{ref}},
    \qquad k=0,\ldots,k_{\max},
\end{equation}
are built by first-order Lie--Trotter product formulas, with the step $\delta t=\pi/(2\sum_i|c_i|)$ of Appendix~\ref{app: solver}, and measured in the computational basis. Here $\widetilde H_{\T}$ is a pruned version of $H_{\T}$ keeping only its $n_P$ Pauli terms of largest coefficient magnitude (identity excluded). The pruned Hamiltonian $\widetilde H_{\T}$ is determined as follows.
First, Pauli terms with coefficients $|c_i|<10^{-3}\max_j|c_j|$ are dropped from the set of operators in the expansion.
Next, we keep the $n_P=400$ Pauli terms of largest magnitude and assess whether the deepest transpiled circuit has at most $4{,}000$ two-qubit gates.
We iteratively reduce $n_P$ by factors of $0.6$ until this is the case.
The size of this set of Pauli terms $n_P$ is kept fixed for the remaining training epochs.
The resulting value of $n_P$ is $240$ for the 50-qubit hardware run.
For comparison, the full $H_{\T}$ contains over $20{,}000$ terms. Note that the pruning applies only to the evolution operator, while the projection uses the full $H_{\T}$.

At the first iteration the reference is $\ket{0^n}$ and the thirteen circuits $k\in \{0,\ldots,12\}$ are measured with $25{,}000$ shots each. At later iterations the references are the four highest-weight bitstrings of the previous ground state, the circuits $k\in\{1,2,3\}$ are measured for each reference ($12$ circuits at $8{,}192$ shots), and the entire previous subspace is added to the union, so that the subspace can follow the training trajectory of the Hamiltonian. 

Circuits are mapped to a chain of $n$ physical qubits chosen from the calibration data at run time, to minimize the summed two-qubit-gate and readout errors along the chain by a greedy search over the coupling graph. These circuits are then executed with dynamical decoupling and measurement twirling enabled; no other error mitigation is used.

\subsection{Model training insights}
\label{app:ensemble}

For every hardware run we compare the quantumly trained HE model to the classical iterative solver applied to the same data cell (\cref{fig:qpu_vs_sim}). In all five runs, the accuracy of the hardware runs is greater, by $+0.004$ to $+0.050$ (mean $+0.023$). 

\begin{figure}[t]
  \centering
  \includegraphics[width=0.5\linewidth]{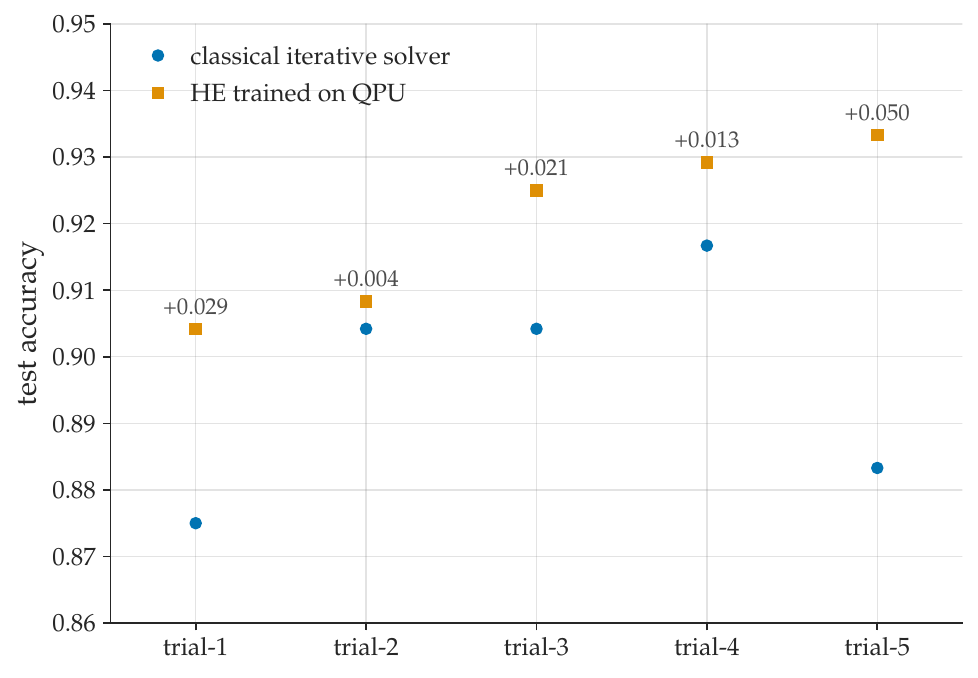}
  \caption{
  A comparison of the test accuracy from the Hamiltonian embedding model, when used in conjunction with a classical iterative solver (blue circles) and the SKQD method on hardware (orange squares) for the dataset of \cref{fig: qpu train overlay}(b).
  }
  \label{fig:qpu_vs_sim}
\end{figure}

Finally, to test whether the seed-to-seed variation reflects genuine differences in learned models or merely finite-sample noise in the $240$-example test set, we perform a cross-seed transfer analysis in simulation with the classical iterative solver. We evaluate twenty independently trained cells on all twenty held-out test splits (\cref{fig:cross_seed}). This allows us to compute the seed-to-seed variation for the model initialization (indexed by rows) and the data splits (indexed by columns). The near-uniformity of the map shows that the accuracy barely changes when a model is applied to a different data sample.
Both the native (diagonal) and transfer (off-diagonal) accuracies average $\approx\!0.885$.
Thus, a model trained on one data sample performs just as well on an independent sample, confirming that HE learns transferable, sample-independent structure from the data.

\begin{figure}[t]
  \centering
  \includegraphics[width=0.5\linewidth]{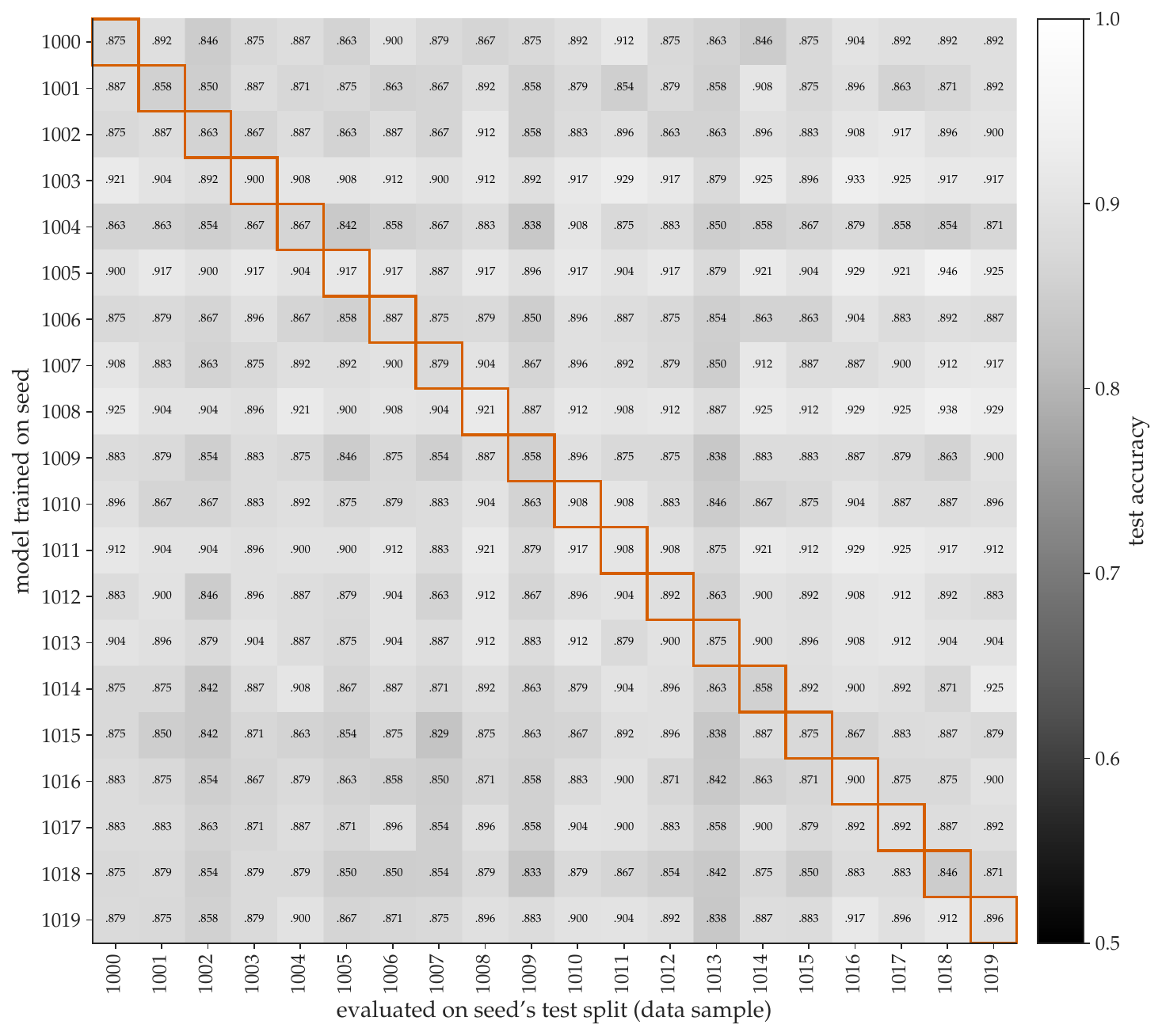}
  \caption{Cross-seed transfer at $n = 50$. Each cell is the test accuracy of a model trained on one seed (row) and evaluated on another seed's held-out data split (column); the red boxes mark the diagonal, where a model is tested on its own split.}
  \label{fig:cross_seed}
\end{figure}

\section{A numerical counterexample to Pauli classical surrogate models}\label{app: counter surrogate}
In Appendix~\ref{app: provable sep surrogate}, we presented a learning task for which the Hamiltonian embedding model provably requires fewer training examples than the corresponding classically regularized surrogate.
Here, we present a different example as a proof of principle, numerically showing that retaining the quantum state during learning yields a substantial performance advantage over the corresponding Pauli-based classical surrogate. 

Specifically, we compare the performance of a particular Hamiltonian embedding model, defined by a fixed choice of encoding unitary $U(x,\alpha)$ and observable $O(y)$, with that of the corresponding Pauli classical surrogate. The latter takes the form
\begin{equation}
g_{\alpha}(x)
=
\sum_i \phi_i(x,\alpha)\,\widetilde{m}_i(\alpha),
\end{equation}
where the coefficients $\phi_i(x,\alpha)$ are defined through the Pauli decomposition
\begin{equation}
U^\dagger(x,\alpha) A U(x,\alpha)
=
\sum_{i=1}^K \phi_i(x,\alpha) P_i,
\end{equation}
where $O(y)=\frac{I+(-1)^yA}{2}$.
Here, $P_i$ are the $K=\text{poly}(n)$ Pauli operators whose expectation values are measured by the Hamiltonian embedding model as explained in Section~\ref{sec: Pauli surrogates}, while $\widetilde{m}_i(\alpha)\in\mathbb{R}$ are unconstrained classical parameters replacing the physical quantum moments
$
m_i(\alpha)
=
\bra{\psi_{{\scriptscriptstyle\mathcal{T}},\alpha}} P_i \ket{\psi_{{\scriptscriptstyle\mathcal{T}},\alpha}}$ for the quantum model in Eq.~\eqref{eq:trained_quantum_model}.
Thus, the two models use the same input-dependent parameterized Pauli feature map, but differ in that the coefficients of the Hamiltonian embedding model must arise as expectation values of a single physical quantum state, whereas those of the classical surrogate are optimized freely. Schematically,
\begin{align}
 f_{HE}(x) &= \sum_i \phi_i(x,\alpha)\,m_i(\alpha),
&m_i(\alpha)=\bra{\psi_{{\scriptscriptstyle\mathcal{T}},\alpha}}P_i\ket{\psi_{{\scriptscriptstyle\mathcal{T}},\alpha}},\\
 g_{\rm ridge}(x) &= \sum_i \phi_i(x,\alpha)\,\tilde m_i(\alpha),
&\tilde m_i(\alpha)\in\mathbb{R}\ .
\end{align}
For a fixed value of the trainable parameter vector $\alpha$, the
coefficients of the classical surrogate are obtained by solving the ridge
regression problem
\begin{equation}\label{eq: ridge regression}
\tilde m^\star(\alpha)
=
\arg\min_{\tilde m }
\left[
\frac{1}{2L}
\sum_{\ell=1}^{L}
\left(
(-1)^{y_\ell}-\phi(x_\ell,\alpha)^T \tilde m
\right)^2
+
\frac{\lambda}{2}\|\Tilde m\|_2^2
\right].
\end{equation}
Since the feature map $\phi(x,\alpha)$ itself depends on $\alpha$, the
optimal vector $\tilde m^\star$ depends implicitly on the choice of
the representation parameters. The corresponding loss is then minimized with respect to the trainable parameters $\alpha$.

\subsection{Construction of the learning task}
The task is constructed such that the requirement that the Pauli moments
$m_i(\alpha)$ be jointly compatible with a physical quantum state provides
a useful inductive bias in the low-data regime. To emphasize this effect,
we consider a teacher--student setting in which the inputs $x$ are labeled
according to
\begin{equation}\label{eq: labels_surrogate}
y=\frac{1-s(x)}{2}\qquad\qquad
s(x)=\operatorname{sign}\!\left[
\bra{\psi^\star}
U^\dagger(x,\alpha^\star)
A
U(x,\alpha^\star)
\ket{\psi^\star}
\right],
\end{equation}
for a fixed teacher parameter $\alpha^\star$ and teacher state
$\ket{\psi^\star}$. We then use the same parameterized unitary
$U(x,\alpha)$ in the dataset Hamiltonian, with $\alpha$ initialized
randomly. The goal of the Hamiltonian embedding model is therefore to
learn parameters $\alpha$ and the corresponding ground state
$\ket{\psi_{{\scriptscriptstyle\mathcal{T}},\alpha}}$ such that the prediction rule in Eq.~\eqref{eq: prediction rule} reproduces
the teacher labels.

In particular, we choose the following setting. We consider a system of $n=2M$ qubits, grouped into $M$ two-qubit blocks $(2j-1,2j)$ for $j\in\{1,..,M\}$. The teacher state $\ket{\psi^\star}$ is the product of Bell states on each block 
\begin{equation}
 \ket{\psi^\star}=\ket{\Phi^+}^{\otimes M},
 \qquad
 \ket{\Phi^+}=\frac{\ket{00}+\ket{11}}{\sqrt{2}}.
\end{equation} 

The inputs $x$ are chosen to be a tuple of four different variables $x=(j,\theta,\beta,q)$ with different roles:
\begin{itemize}
  \item $j$ selects the active two-qubit block and it is sampled uniformly at random in $j\in\{1,...,M\}$;
\item $\theta$ is sampled from a two-pole Gaussian mixture centered at $0$ and $\pi$,
\begin{align}
\theta \sim
\frac{1}{2}\mathcal{N}(0,\sigma_\theta^2)
+
\frac{1}{2}\mathcal{N}(\pi,\sigma_\theta^2),
\qquad
\sigma_\theta=0.12,
\end{align}
where $\mathcal{N}(\mu,\sigma^2)$ denotes a normal distribution with mean $\mu$ and variance $\sigma^2$.
  \item $\beta\sim\mathrm{Unif}[-\pi,\pi]$;
  \item $q=(q_0,q_1,q_2,q_3)$ is a random unit vector such that $q^2_0+q^2_1+q^2_2+q^2_3=1$. It parametrizes a random single-qubit unitary $V(q)\in SU(2)$, which is applied to the active two-qubit block through $V(q)\otimes V(q)^*$. 
\end{itemize}

The label-dependent observable is defined as
\begin{equation}
O(y,j)=\frac{I+(-1)^yA_j}{2},
\qquad
A=\frac{2}{3}X\otimes X-\frac{1}{3}Y\otimes Y.
\end{equation}
and $A_j$ denotes the operator $A$ acting on the block $(2j-1,2j)$ and identity elsewhere.

The encoding unitary acts nontrivially only on the two-qubit block selected
by the input index $j$. Specifically,
\begin{equation}
U(x=(j,\theta,\beta,q),\alpha)
=
I^{\otimes 2(j-1)}
\otimes
U_{\mathrm{loc}}(\theta,\beta,q;\alpha_j)
\otimes
I^{\otimes 2(M-j)},
\end{equation}
where
\begin{equation}
U_{\mathrm{loc}}(\theta,\beta,q;\alpha_j)
=
C(\theta)
\left[
R_{\bm n(\theta)}
\!\left(\beta+\alpha_j h(\theta)\right)
\otimes I
\right]
S(q),
\end{equation}
with
\begin{align}
C(\theta)
&=
R_Y(\theta/2)\otimes R_Y(-\theta/2),\\
\bm n(\theta)
&=
\left(\sin(\theta/2),0,\cos(\theta/2)\right),\\
h(\theta)
&=
\sin^2(\theta/2),\\
S(q)
&=
V(q)\otimes V(q)^*,
\end{align}
and
\begin{equation}
V(q)
=
q_0 I-i\left(q_1X+q_2Y+q_3Z\right).
\end{equation}
We note that the unitary  $S(q)=V(q)\otimes V(q)^*$ has the Bell state $\ket{\Phi^+}$ as eigenstate. In particular $
 (V\otimes V^*)\ket{\Phi^+}=\ket{\Phi^+}$
for any single-qubit unitary $V$. Therefore $q$ can rotate the Pauli representation seen by the learner while leaving the Bell teacher prediction unchanged. Finally, we also note that the local encoding $U_{\mathrm{loc}}(\theta,\beta,q;\alpha_j)$ can be implemented entirely using standard single-qubit rotation gates. In particular, $R_{\bm n(\theta)}(\varphi)$ can be decomposed as $R_Y(\theta/2)R_Z(\varphi)R_Y(-\theta/2)$, while $V(q)$ is itself a single-qubit rotation parameterized by the quaternion $q$.

We then label the data by the rule in Eq.~\eqref{eq: labels_surrogate}, with the choice of $\alpha_j^\star=0$ for all $j\in\{1,...,M\}$. The teacher corresponding expectation value is then 
\begin{equation}\label{eq: teacher_exp_value}
 \bra{\Phi^+}U_{\rm loc}(\theta,\beta,q;\alpha^\star)^\dagger A_jU_{\rm loc}(\theta,\beta,q;\alpha^\star)\ket{\Phi^+}
 =\frac{1+2\cos\theta}{3}\cos\beta.
\end{equation}
We note that, although the observable $A_j$ explicitly depends on the input index $j$, the same construction can be written with a fixed observable by absorbing this dependence into the encoding unitary. Choosing $A_1$ as the fixed observable and letting $P_j$ swap the pair $(2j-1,2j)$ with $(1,2)$, so that $P_j^\dagger A_1P_j=A_j$, we can define $\widetilde U(x,\alpha)=U_{\mathrm{loc}}^{(1,2)}(\theta,\beta,q;\alpha_j)P_j$, which gives $\widetilde U^\dagger(x,\alpha)A_1\widetilde U(x,\alpha)=P_j^\dagger U_{\mathrm{loc}}^{(1,2)\dagger}A_1U_{\mathrm{loc}}^{(1,2)}P_j$ and therefore reproduces the original $j$-dependent construction.

From the labeling rule in Eq.~\eqref{eq: labels_surrogate} and Eq.~\eqref{eq: teacher_exp_value}, the labels are then given by Eq.~\eqref{eq: labels_surrogate} with
\begin{equation}\label{eq: surrogate_labeling}
\quad s=\text{sign}\!\left[(1+2\cos\theta)\cos\beta\right].
\end{equation}
The nuisance variable $q$ and the block index $j$ do not appear in the labeling rule, but they make the learning problem for the classical surrogate more demanding. In particular, the index $j$ determines which two-qubit block is probed by each example, with each block associated with its own vector of Pauli coefficients $\widetilde m_j$ in the classical surrogate. On the other hand, $q$ changes the Pauli decomposition of the encoded observable. Indeed, conjugation by the random unitary $V(q)$ rotates Pauli operators into different linear combinations of $X$, $Y$, and $Z$, so that different values of $q$ make different examples probe different combinations of the Pauli coefficients. In the low-data regime, the training set may therefore be compatible with many choices of the free surrogate coefficients $\widetilde m$, which can fit the observed data while generalizing differently to unseen inputs. In contrast, the Hamiltonian embedding model restricts the coefficients $m$ to be jointly compatible with the Pauli moments of a physical quantum state. This provides an additional inductive bias that favors coefficient vectors consistent with the teacher construction and, in our numerical experiment, leads to better generalization before the classical surrogate has received enough data to identify a well-generalizing solution. 

\subsection{Numerical experiment}
We compare the Hamiltonian embedding model and the Pauli classical surrogate for systems ranging from $4$ to $20$ qubits, with $n\in\{4,6,8,10,12,14,16,18,20\}$, using 20 independent random seeds and 5000 fresh test examples per run. For each system size, we vary the number of training examples supplied to both models. Specifically, for $M=n/2$ two-qubit pairs, the total training-set size is chosen as $L=Ms$, where $s$ is varied over a fixed grid of samples per pair. For each $(n,L,\text{seed})$, we also vary the strength $\lambda$ of the regularization term in the ridge regression of Eq.~\eqref{eq: ridge regression}. We then evaluate the two models on the same training/test data and with the same initial random $\alpha$.

As a common metric, for each model and system size we define
$L_{90}$ as the smallest training-set size $L$ for which the mean
test accuracy over the 20 random seeds reaches $90\%$.
This quantity is plotted in \cref{fig:lambda01-errorbars}.
Because $L$ is evaluated on a discrete grid, the reported values are first grid crossings.
To quantify seed-to-seed variability, for each seed $r$ we also define the first seed-wise crossing.

\begin{figure}[t]
\centering
\includegraphics[width=0.55\textwidth]{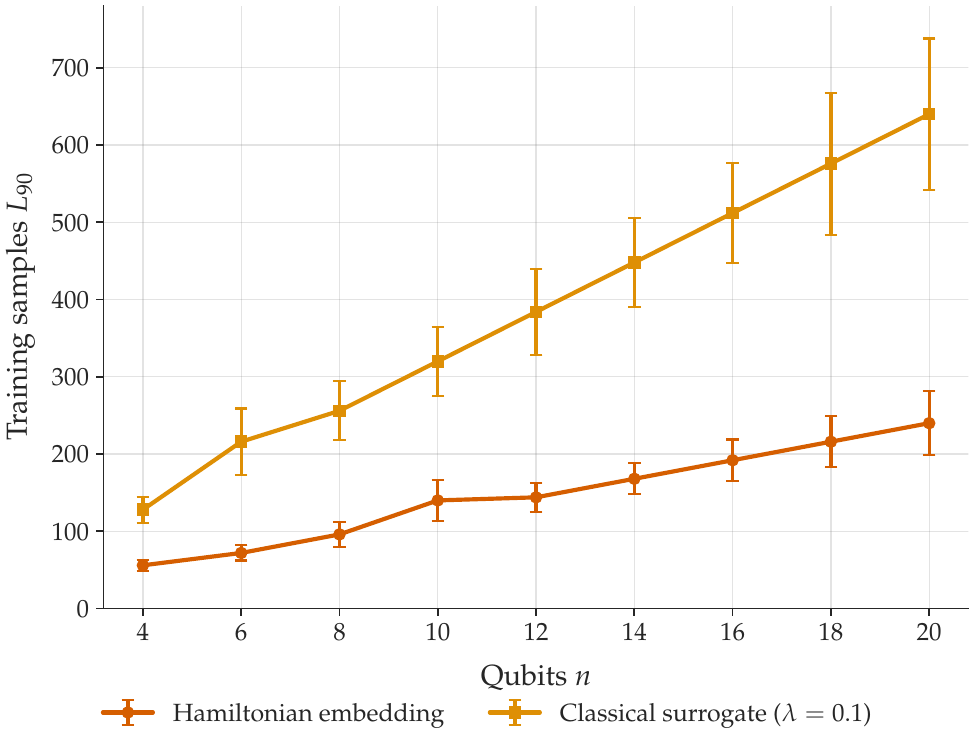}
\caption{Training examples needed to reach 90\% test accuracy for Hamiltonian embedding and classical surrogate with $\lambda=0.1$. The sweep uses 20 independent seeds. Central values are the first crossings of the mean test-accuracy curve across the 20 seeds, while the error bars are the standard deviations of the seed-wise first-crossing values. The bars therefore show seed-to-seed variability rather than a conventional confidence interval. }
\label{fig:lambda01-errorbars}
\end{figure}

The error bars used in \cref{fig:lambda01-errorbars} are the standard deviation of these 20 seed-wise first-crossing values. The plotted central point is the smallest training-set size $L$ for which the mean accuracy over the random seeds reaches $90\%$. Therefore the bars should be interpreted as a diagnostic of run-to-run variability, rather than as a confidence interval around the central point.

The complete results are shown in Table~\ref{tab:alllambda}. The first row is the Hamiltonian embedding model. The remaining rows give the classical surrogate $L_{90}$ for each ridge regularization value.

\begin{table}[ht]
\caption{First tested training size $L$ for which the mean test accuracy over all seeds reaches
$90\%$. The $\lambda$ sweep uses 20 independent seeds and 5000 test points
per run.}
\label{tab:alllambda}
\begin{ruledtabular}
\begin{tabular}{lccccccccc}
Model / $\lambda$ & $n=4$ & $6$ & $8$ & $10$ & $12$ & $14$ & $16$ & $18$ & $20$ \\
\colrule
HE        & 56  & 72  & 96  & 140 & 144 & 168 & 192 & 216 & 240 \\
$10^{-8}$ & 160 & 216 & 320 & 400 & 432 & 560 & 640 & 648 & 800 \\
$10^{-7}$ & 160 & 216 & 320 & 400 & 432 & 560 & 640 & 648 & 800 \\
$10^{-6}$ & 160 & 216 & 320 & 400 & 432 & 560 & 640 & 648 & 800 \\
$10^{-5}$ & 160 & 216 & 320 & 400 & 432 & 560 & 640 & 648 & 800 \\
$10^{-4}$ & 160 & 216 & 320 & 400 & 432 & 560 & 640 & 648 & 800 \\
$10^{-3}$ & 160 & 216 & 320 & 400 & 432 & 560 & 640 & 648 & 800 \\
$10^{-2}$ & 128 & 216 & 288 & 360 & 432 & 448 & 576 & 648 & 720 \\
$10^{-1}$ & 128 & 216 & 256 & 320 & 384 & 448 & 512 & 576 & 640 \\
$1$       & 144 & 240 & 256 & 360 & 432 & 448 & 640 & 576 & 640 \\
\end{tabular}
\end{ruledtabular}
\end{table}
The weakly regularized regime, from $10^{-8}$ through $10^{-3}$, gives the same first-crossing thresholds on the tested grid. Moderate regularization improves ridge regression. The strongest improvement is typically around $\lambda=10^{-1}$, although the optimum is not identical at every system size.

To make the comparison maximally favorable to the classical surrogate, we also report the comparison between the $L_{90}$ value for the Hamiltonian embedding and the classical surrogate at the best regularization strength $\lambda$, for each system size, i.e.  
$
 L_{90}^{\rm ridge,best}(n)=\min_{\lambda\in\Lambda}L_{90}^{\rm ridge}(n,\lambda)$,
where $\Lambda$ is the tested grid.  As the best performance for the classical surrogate is found at $\lambda=0.1$ for almost all system sizes, we plot the comparison with the Hamiltonian embedding model in \cref{fig:lambda01-errorbars}.

\begin{table}[t]
\caption{Hamiltonian embedding compared with the best ridge threshold observed over the tested $\lambda$ grid. The final column gives the ridge-to-HE ratio averaged over system sizes.}
\label{tab:bestlambda}
\begin{ruledtabular}
\begin{tabular}{lcccccccccc}
 & 4 & 6 & 8 & 10 & 12 & 14 & 16 & 18 & 20 & mean ratio \\
\colrule
$L_{90}^{\text{HE}}$
& 56 & 72 & 96 & 140 & 144 & 168 & 192 & 216 & 240 & -- \\
$L_{90}^{\rm ridge,best}$
& 128 & 216 & 256 & 320 & 384 & 448 & 512 & 576 & 640 & $2.62\times$ \\
\end{tabular}
\end{ruledtabular}
\end{table}

Finally, we emphasize that the labeling rule in Eq.~\eqref{eq: surrogate_labeling} is deliberately simple, and a classical learner provided with additional structural information could easily solve the learning task. Nevertheless, the example serves to demonstrate that there exist settings in which the physical bias imposed by requiring the Pauli moments $m_i$ to arise from a valid quantum state can improve the performance of the Hamiltonian embedding model relative to the corresponding classical surrogate.

\end{document}